\documentclass[a4paper,UKenglish,thm-restate]{lipics-v2021}

\usepackage{bussproofs}
\usepackage{logicproof}[2014/03/20]
\usepackage{amssymb}
\usepackage{amsmath}
\usepackage{amsfonts}
\usepackage{wrapfig}
\usepackage{hyperref}
\usepackage{pdfpages}
\usepackage{cleveref}
\usepackage{enumitem}
\PassOptionsToPackage{table,dvipsnames}{xcolor}
\usepackage{xcolor}
\usepackage{colortbl}  
\usepackage{tabularx}

\hypersetup{
    colorlinks=true,
    linkcolor=blue,   
    citecolor=blue,   
    urlcolor=blue     
}

\ifpdf
  \usepackage{tikz}
  \usetikzlibrary{decorations.pathreplacing,decorations.markings,calligraphy}
  \usetikzlibrary{arrows,automata}
  \usetikzlibrary{positioning}
 \else
  \AtBeginDocument{%
    \RemoveFromHook{shipout/firstpage}[hyperref]%
    \RemoveFromHook{shipout/before}[hyperref]%
  }
\fi

\usepackage{float}
\usepackage{microtype}
\usepackage{graphicx}
\usepackage{algorithm}
\usepackage{algpseudocode} 
\algblockdefx[MATCH]{Match}{EndMatch}%
  [1]{\textbf{match} #1 \textbf{with}}%
  {\textbf{end match}}
\algdef{SxN}[CASE]{Case}{EndCase}%
  [1]{{#1}:}
\algdef{SxN}[IFTHEN]{IfThen}{EndIfThen}%
  [2]{\textbf{if}~{#1}~\textbf{then} {#2}}
\algdef{SxN}[FORONELINE]{ForOne}{EndForOne}%
  [2]{\textbf{for}~{#1}~\textbf{do} {#2~\textbf{end for}}}
\algdef{SxN}[CHOOSE]{Choose}{EndChoose}%
  {\textbf{choose non-deterministically}}
\algdef{SxN}[WHEN]{When}{EndWhen}%
  [2]{\textbf{when}~{#1}~\textbf{do}~{#2}}
\newcommand{\Output}{\textsf{output} }
\newcommand{\goto}{\textsf{goto} }
\newcommand{\true}{\textsf{true}}
\newcommand{\false}{\textsf{false}}

\hideLIPIcs  
\nolinenumbers

\renewcommand{\frac}[1]{\{ #1 \}}

\newcommand{\Gg}{\mathcal{G}}

\newcommand{\cA}{\mathcal{A}}
\newcommand{\cB}{\mathcal{B}}
\newcommand{\cN}{\mathcal{N}}
\newcommand{\cM}{\mathcal{M}}

\newcommand{\Nat}{\mathbb{N}}

\newcommand{\xra}[1]{\xrightarrow{#1}}

\newcommand{\release}[1]{\mathsf{release}(#1)}
\newcommand{\A}{\mathcal{A}}
\newcommand{\leqlt}{\mathrel{\triangleleft}}
\newcommand{\nleqlt}{\mathrel{\not\triangleleft}}

\newcommand{\init}{\mathsf{init}}

\newcommand{\prog}{\mathsf{prog}}

\newcommand{\TA}{\textrm{TA}}
\newcommand{\GTA}{\textrm{GTA}}
\newcommand{\GTAs}{\textrm{GTAs}}

\newcommand{\GTAfull}{generalized timed automata}
\newcommand{\GTAp}{\textrm{LiveRGTA}}

\newcommand{\tool}{\textsc{Tempora}}

\newcommand{\MTL}{\sf{MTL}}
\newcommand{\LTL}{\sf{LTL}}
\newcommand{\CTL}{\sf{CTL}}
\newcommand{\TPTL}{\sf{TPTL}}

\newcommand{\MITL}{\sf{MITL}}

\newcommand{\mitl}{\MITL}   

\newcommand{\Prop}{\mathsf{Prop}}

\newcommand{\MTLfp}{\MITL^{+p}}
\newcommand{\detMTL}{\mathsf{det}\MTLfp}

\newcommand{\Xh}{X^H}
\newcommand{\Xf}{X^F}

\newcommand{\X}{\mathop{\mathsf{X}\vphantom{a}}\nolimits}
\newcommand{\Y}{\mathop{\mathsf{Y}\vphantom{a}}\nolimits}
\newcommand{\YP}{\mathop{\mathsf{YP}\vphantom{a}}\nolimits}
\newcommand{\U}{\mathbin{\mathsf{U}}}
\renewcommand{\S}{\mathbin{\mathsf{S}}}

\newcommand{\Release}{\mathbin{\mathsf{R}}}
\newcommand{\F}{\mathop{\mathsf{F}\vphantom{a}}\nolimits}
\newcommand{\G}{\mathop{\mathsf{G}\vphantom{a}}\nolimits}

\newcommand{\Next}{\X}
\newcommand{\Since}{\S}
\newcommand{\Eventually}{\F}
\newcommand{\Always}{\G}
\newcommand{\Past}{\mathop{\mathsf{P}\vphantom{a}}\nolimits}
\newcommand{\History}{\mathop{\mathsf{H}\vphantom{a}}\nolimits}

\newcommand{\Red}[1]{\textcolor{red}{#1}}

\newcommand{\AY}{\cA_{\Y}}
\newcommand{\AS}{\cA_{\S}}
\newcommand{\ASfirst}{\AS^{\mathsf{first}}}
\newcommand{\ASlast}{\AS^{\mathsf{last}}}

\newcommand{\ASgenI}{\cA_{\S_I}^{\mathsf{gen}}}
\newcommand{\AX}{\cA_{\X}}
\newcommand{\AU}{\cA_{\U}}
\newcommand{\AUfirst}{\AU^{\mathsf{first}}}
\newcommand{\AUlast}{\AU^{\mathsf{last}}}
\newcommand{\AUgenI}{\cA_{\U_I}^{\mathsf{gen}}}
\newcommand{\AisatX}{\cA^{isat}_{\X_I}}
\newcommand{\AisatU}{\cA^{isat}_{\U_I}}

\newcommand{\xinit}{x_{\sf{init}}}
\newcommand{\Ainit}{\cA_{\sf{init}}}
\newcommand{\xlast}{x_{\sf{last}}}
\newcommand{\Alast}{\cA_{\sf{last}}}
\newcommand{\xnext}{x_{\sf{next}}}
\newcommand{\Anext}{\cA_{\sf{next}}}
\newcommand{\Aisatp}[1]{\cA^{isat}_{#1}}

\newcommand{\bad}{\sf{Bad}}

\newcommand{\releasedclocks}{\textsc{releasedclocks}}
\newcommand{\CurrentReleasedClocks}{\textsc{CurrentReleasedClocks}}
\newcommand{\labels}{{\textsc{labels}}}
\newcommand{\countint}{{\textsc{count}}}
\newcommand{\Roots}{{\textsc{Roots}}}
\newcommand{\Active}{{\textsc{Active}}}
\newcommand{\Todo}{{\textsc{Todo}}}

\newcommand{\released}{{\mathsf{released}}}
\newcommand{\suc}{{\textsc{succ}}}
\newcommand{\dfsnum}{{\mathsf{dfsnum}}}
\newcommand{\curr}{{\mathsf{current}}}
\newcommand{\nxt}{{\mathsf{next}}}

\newcommand{\Tchecker}{\textsc{Tchecker}}
\newcommand{\MightyL}{\textsc{MightyL}}
\newcommand{\Opaal}{\textsc{OPAAL}}
\newcommand{\UPPAAL}{\textsc{UPPAAL}}
\newcommand{\LTSmin}{\textsc{LTSmin}}

\newcommand{\Acacia}{\mathsf{Acacia}}
\newcommand{\Reqgrant}{\mathsf{Req\text{-}grant}}
\newcommand{\Spotex}{\mathsf{Spot\text{-}ex}}

\renewcommand{\true}{\mathsf{true}}
\renewcommand{\false}{\mathsf{false}}
\newcommand{\argone}{b_1}
\newcommand{\argtwo}{b_2}
\newcommand{\outv}{\mathsf{out}}
\newcommand{\state}{\mathsf{st}}
\newcommand{\pre}[2]{{#1}.{#2}}
\newcommand{\post}[2]{{#1}.#2^{\bullet}}
\newcommand{\now}[2]{{#1}.{#2}}
\newcommand{\postM}[1]{#1^{\bullet}}
\newcommand{\Absolut}[1]{|#1|}

\newcommand{\req}{\mathsf{req}}
\newcommand{\wait}{\mathsf{wait}}
\newcommand{\cs}{\mathsf{cs}}
\newcommand{\eat}{\mathsf{eating}}

\newcommand{\gtasystem}{\mathcal{A}}
\newcommand{\finiteabs}{\mathcal{G}}

\title{ \textsc{TEMPORA}: Efficient Verification of Metric Temporal Properties with Past in Pointwise Semantics}

\titlerunning{Efficient Verification of $\MTL$ with Past in Pointwise Semantics}

\author{S Akshay}{Department of CSE, Indian Institute of Technology Bombay, Mumbai, India}
{akshayss@cse.iitb.ac.in}{https://orcid.org/0000-0002-2471-5997}{}

\author{Prerak Contractor}{Department of CSE, Indian Institute of Technology Bombay, Mumbai, India} 
{prerak@cse.iitb.ac.in}{https://orcid.org/0009-0004-7452-9940}{}

\author{Paul Gastin}{Université Paris-Saclay, ENS Paris-Saclay, CNRS, LMF, 91190,
	Gif-sur-Yvette, France \and CNRS, ReLaX, IRL 2000, Siruseri, 
  India}{paul.gastin@ens-paris-saclay.fr}{https://orcid.org/0000-0002-1313-7722}{}

\author{R Govind}{The Institute of Mathematical Sciences, Chennai, India
  \and Homi Bhabha National Institute, Anushaktinagar, Mumbai, India
  \and CNRS, ReLaX, IRL 2000, Siruseri, India} 
{govind@imsc.res.in}{https://orcid.org/0000-0002-1634-5893}{}

\author{B Srivathsan}{Chennai Mathematical Institute, India
  \and CNRS, ReLaX, IRL 2000, Siruseri, 
  India} {sri@cmi.ac.in}{https://orcid.org/0000-0003-2666-0691}{}

\authorrunning{S. Akshay, P. Contractor, P. Gastin, R. Govind and B. Srivathsan}
\Copyright{S. Akshay, P. Contractor, P. Gastin, R. Govind and  B. Srivathsan}


\category{}

\relatedversion{} 

\supplement{}

\begin{document}
\maketitle

\begin{abstract}
  Model checking for real-timed systems is a rich and diverse topic.  Among the different
  logics considered, Metric Interval Temporal Logic ($\mitl$) is a powerful and commonly
  used logic, which can succinctly encode many interesting timed properties especially
  when past and future modalities are used together.  In this work, we develop a new
  approach for $\mitl$ model checking in the pointwise semantics, where our focus is on
  integrating past and maximizing determinism in the translated automata.

  Towards this goal, we define synchronous networks of timed automata with shared
  variables and show that the past fragment of $\mitl$ can be translated in linear time to
  synchronous networks of deterministic timed automata.  Moreover determinism can be
  preserved even when the logic is extended with future modalities at the top-level of the
  formula.  We further extend this approach to the full $\mitl$ with past, translating it
  into networks of generalized timed automata (\GTA) with future clocks (which extend
  timed automata and event clock automata). 
  We present an SCC-based liveness algorithm to analyse \GTA.
  We implement our translation in a prototype tool which handles both finite and infinite
  timed words and supports past modalities.  Our experimental evaluation demonstrates that
  our approach significantly outperforms the state-of-the-art in $\mitl$ satisfiability
  checking in pointwise semantics on a benchmark suite of 72 formulas.  Finally, we
  implement an end-to-end model checking algorithm for pointwise semantics and demonstrate
  its effectiveness on two well-known benchmarks.
\end{abstract}

\section{Introduction}\label{sec:intro}

Linear Temporal Logic ($\LTL$) is a standard specification formalism.  A central
ingredient in the model-checking approach is a translation from $\LTL$ to automata.  The
study of algorithmic techniques to optimize $\LTL$-to-automata translations spans almost
four decades -- see~\cite{TsayVardi2021} for a recent survey of this topic.  Many of the
theoretical techniques have also been incorporated in practical model-checking
tools~\cite{Spin,NuSMV,Spot,PAT,OPAAL}.  For model-checking systems with timing
constraints, the de-facto automaton model would be \emph{timed automata}
(TA)~\cite{AD94,AlurD90}.  State-of-the-art timed automata
tools~\cite{UPPAAL-1,TChecker,PAT,OPAAL} implement analysis methods for models represented
as (networks of) timed automata (see Section 7 of~\cite{handbook-mc} for the practical
applications of these tools).  Properties checked on timed models are typically
reachability, liveness, fragments of timed $\CTL$~\cite{UPPAAL-1} or (untimed) $\LTL$.  On
the other hand, there is a very large body of foundational work on \emph{timed
logics}~\cite{timed-logics-Raskin-thesis,HirshfeldR04,SchobbensRH02}.  The two main timed
logics are Metric Temporal Logic~\cite{Koymans90} ($\MTL$) and Timed Propositional
Temporal Logic~\cite{tptl-AlurH94} ($\TPTL$).  In order to use these logics for real-time
model-checking, we need an efficient translation from timed logics to timed automata.  The
foundations for $\MTL$-to-TA (resp.  $\TPTL$-to-TA) have been laid out
in~\cite{AFH-MITL-J} (resp.~\cite{tptl-AlurH94}).

Broadly, there are two ways to interpret $\MTL$ formulae: over continuous
(dense/super-dense) timed signals~\cite{AFH-MITL-J,MalerNP06} or pointwise timed
words~\cite{AlurH90,Wilke94,MightyL}.  In \emph{continuous (state-based) semantics}, the
formula is interpreted over \emph{signals}, which are functions that associate the
valuation of the propositions to \emph{all} timestamps.  Over the \emph{pointwise
semantics}, the formula is interpreted over sequences of events (timed words), each
consisting of a timestamp and a set of propositions.  Thus, in the continuous semantics,
the system is under observation at all times, while in the pointwise semantics, it is
observed only when it executes a discrete transition.  However, as noted by Ferr\'ere et
al.~\cite{jacm/FerrereMNP19}, in the continuous semantics, arbitrary $\MTL$ formulae can
be simplified to the $[0, \infty]$ fragment, while as pointed out in~\cite{AGGS-CONCUR24},
this simplification is not possible in the pointwise semantics.
In this work, we only consider pointwise semantics as our focus is on models represented
in tools like UPPAAL, PAT, TChecker whose behaviours are given as timed words, as in the
classical timed automata model of Alur and Dill~\cite{AD94}.

In the pointwise semantics, it is known that $\MTL$ model checking is undecidable over
infinite words, and over finite words it is decidable but the complexity is
Ackermanian~\cite{Survey-AlurH91,MTL-OW07J,MTL-OW07C,OuaknineW06}.
Thus, the literature often considers the more efficient $\MITL$ fragment, where
$\mathsf{Until}$ and $\mathsf{Since}$ operators must use non-punctual intervals (either
$[0,0]$ or non-singleton intervals).  The classical approach for model checking for
$\MITL$ involves translating $\MITL$ formula into timed automata (TA) and then using
state-of-the-art TA tools such as UPPAAL~\cite{UPPAAL-1,UPPAAL-2} or
TChecker~\cite{TChecker}.  The best known tool that accomplishes this translation
in the pointwise semantics, called \MightyL~\cite{MightyL}, goes via alternating 1-clock
automata and compiles to timed automata.  Recently, an extension of timed automata called
\emph{generalized timed automata} (GTA) was introduced in~\cite{AGGJS-CAV23}.  GTA are
equipped with richer features than TA,
and enable a different translation from $\MITL$ to GTA, which can
then be analyzed using a GTA analysis tool from~\cite{AGGJS-CAV23}.  The theoretical
foundations for such an $\MITL$-to-GTA translation which goes via transducers was presented
in~\cite{AGGS-CONCUR24}, where it was argued that the $\MITL$-to-GTA translation has clear
theoretical benefits over the $\MITL$ to 1-clock Alternating Timed Automata
translation~\cite{BrihayeEG13,BrihayeEG14} on which \MightyL\ is based on.
However, this latter approach was not implemented.  Moreover, neither of these approaches
considered Past modalities, and were restricted only to Future.

On the other hand, from the classical work of Maler et al~\cite{MalerNP05}, in the signal
semantics, it is known that having only Past modalities leads to efficient algorithms and
in particular it gives a translation to {\em deterministic} (timed) automata albeit over
signals.  Our first question is whether the same holds under the pointwise semantics, and
whether determinism can be extended beyond just the past fragment.  Next, if we cannot
guarantee determinism, can we reduce non-deterministic branching?  Indeed, even if it does
not change the expressiveness, several works have highlighted the usefulness and
succinctness of combining Past and Future operators, both in the untimed
setting~\cite{LPZ85,LTL-past-Markey,ALP24} and in the timed
setting~\cite{KKP11,jacm/FerrereMNP19}.

In this work, we address these questions both for past and future in the pointwise
semantics and to build the foundations for a new efficient model checking algorithm for
$\MITL$ with past that exploits determinism as much as possible.  We show that for the
past fragment of $\MITL$, we can obtain a linear-time translation to networks of
deterministic timed automata.  Further, even when we have to handle future operators,
where non-determinism cannot be avoided, we propose techniques to reduce non-deterministic
branching in the resulting automata, using various ideas including sharing components,
automata, clocks and predictions.  We start by defining a formal semantics for a model of
synchronous network of timed automata with shared variables, where each component of the
network owns some clocks which only it can reset but can be tested by all components.
Moreover it allows for {\em instantaneous programs on transitions} which have sequences of
tests and updates, like in~\cite{AGGS-CONCUR24}.  Our main contributions are the
following:
\begin{enumerate}
  \item We identify a fragment of metric temporal logic with past and future for which we
  provide a linear time translation to networks of deterministic timed automata.  In
  this fragment, that we call $\detMTL$, the outermost
  temporal modalities are future (next, until) and can even have punctual intervals, while
  inner temporal modalities should be past: Since with non-punctual interval or Yesterday.
  This fragment already subsumes many safety properties, extending e.g., the
  $\mathsf{G(pLTL)}$ fragment defined in~\cite{AGGMM23}.
  
  \item We present a new linear-time algorithm to translate $\detMTL$ formula to networks
  of deterministic timed automata.  This algorithm is in two steps.  First we show how the
  outermost temporal modalities can be handled with respect to \emph{initial} satisfiability.
  Next we show how the past temporal modalities are handled for general satisfiability.
  We obtain a synchronous network of deterministic automata with shared variables that
  captures the semantics.  We start with Yesterday allowing an arbitrary interval
  constraint.  Then, we deal with Since operators with one-sided intervals of the form
  $[0,a]$ or $[a,\infty)$.  For these simpler cases, we construct networks where each
  automaton has at most 1 clock and 2 states, but critically uses the sharing feature
  across components.  Finally, the automaton for Since with a general two-sided interval
  [a,b] is the most complicated part of the construction and results in a network whose
  one component uses $\mathcal{O}(a/(b\!-\!a))$ clocks.  Our construction works for both
  finite and infinite words.
  
  \item Next, we extend this algorithm to handle future modalities.  To achieve this we
  move from networks of timed automata to networks of generalized timed automata (\GTA)
  ~\cite{AGGJS-CAV23}, which have future clocks in addition to usual (history) clocks.  We
  obtain a translation from general $\MTLfp$ (with Past and allowing punctual intervals at
  the outermost level) to networks of $\GTA$ with shared variable.  While non-determinism
  cannot be avoided here, we show that we can reduce non-deterministic branching in two
  ways: first by using shared variables to share predictions, and second, by using
  non-Zenoness assumption and the semantics of future clocks to obtain a 2-state automaton
  for Until (which results in an exponential improvement in running time for the model
  checking algorithm).  Again our model checking algorithm works for infinite words and is
  easily adaptable for finite timed words.

  \item Since our model checking approach ultimately reduces to checking reachability and
  liveness for \GTA s, we propose an improved liveness-checking algorithm that extends
  Couvreur's SCC algorithm~\cite{Couvreur} (which is the state-of-the-art for timed
  automata) to the setting of \GTA s, marking the first adaptation of Couvreur's procedure
  to automata with event clocks or diagonal constraints.

  \item Finally, we implement our algorithms in a prototype tool, called \tool, built on
  top of the open-source timed automata library \Tchecker.  \tool\ supports satisfiability
  and model checking for $\detMTL$ and general $\MTLfp$, for both finite and infinite
  words, in the pointwise semantics.
  We compare against $\MightyL$ which translates $\MITL$ formula to timed automata (and
  uses $\UPPAAL$ to check satisfiability with finite words and $\Opaal$ for infinite
  words).  Our comparison on a suite of 72 benchmarks (22 new and 50 taken from the
  literature) shows that our pipeline (\tool+\Tchecker) significantly outperforms
  $\MightyL$ pipeline; it is faster in all benchmarks for finite words and all but 2 for
  infinite words, and most often explores a smaller state space (in zones).  We also
  implement the theoretical translation to $\GTA$ presented in~\cite{AGGS-CONCUR24} and
  demonstrate an order of magnitude improvement.  Finally, we also implement an end-to-end
  model checking algorithm for $\MTLfp$ over the pointwise semantics and evaluate our
  model-checking pipeline on two models: Fischer's mutual exclusion algorithm and the
  Dining Philosophers problem.
\end{enumerate}

\textbf{Related works.} 
Model checking for real-time systems has been extensively studied,
with Metric Interval Temporal Logic (MITL) emerging as a widely used specification
language due to its expressiveness for capturing timed properties.  A central challenge
lies in efficiently translating MITL formulas into automata-based representations,
enabling algorithmic verification.  For a comprehensive overview of the topic, we refer
the reader to the survey by Bouyer~\cite{Bouyer09}.

Over continuous/signal semantics, we already discussed several works that focus on
obtaining networks of timed automata for various fragments of MITL with and without past
including~\cite{jacm/FerrereMNP19,MalerNP05,MalerNP06}.  Most of these employ construction
of compositional testers or transducers for timed modalities, but do not focus on sharing,
e.g., of networks and subformula. 
Despite the difference in the semantics and hence in the reasoning, the automata
constructed, e.g., the event-recorder automaton in~\cite{MalerNP05} for the general Since
operator, are useful even for our setting.  In this sense, our work can be seen as an
adaptation of these classical works to the pointwise setting. A different purely logic
based approach was considered in~\cite{MBRP20}, where the timed automata and $\MITL$
formula, still with signal semantics, are translated to an intermediate logical language,
which is then encoded into an SMT solver.  This work also included synchronization
primitives and was implemented in a Java tool.

Coming back to pointwise semantics, in~\cite{AFH-MITL-J} that introduced $\MITL$, the
authors propose the first such translation, which is known to be notoriously complicated,
spanning several pages, and acknowledged to be difficult in the follow-up works by
different authors~\cite{MalerNP06}.  The current state-of-the-art for $\MITL$ model
checking is the procedure proposed by Brihaye et al.\ in~\cite{MightyL}, where they
translate $\mitl$ formulae to a network of timed automata, which can be plugged in a
standard timed automata tool implementing non-emptiness.  Their procedure and the
associated tool, called \MightyL~\cite{MightyL} is based on a series of works that study a
translation of $\MITL$ to timed automata~\cite{BrihayeEG13,BrihayeEG14}, via alternating
timed automata~\cite{OuaknineW06}.  \cite{DBLP:conf/fossacs/BouyerSV25} studies a
zone-based algorithm directly for 1-clock alternating timed automata, and focusses on a
fragment of $\MTL$.

Several tools support $\MITL$ model checking under continuous semantics,
including~\cite{mitlbmc,nuxmv0,nuxmv1,Zot,stl,MBRP20}.  Under pointwise semantics, to the
best of our knowledge, MightyL~\cite{MightyL} is the only tool that can handle the full
future fragment of MITL (even if it only solves satisfiability and does not handle past
operators).  The other works (to our knowledge) that offer tool support are for restricted
fragments such as $\MITL_{(0.\infty)}$~\cite{BulychevDLL14}, or $\MITL$ over untimed
words~\cite{ZhouMB16}.

\textbf{Structure of the paper.} In \Cref{sec:sync-network-TA}, we formalize
the notion of a synchronous network of timed automata with shared variables, whose
transitions are instantaneous programs, and define the model checking problem in terms of
such networks.  In \Cref{sec:detMTL} we describe the $\detMTL$ logic followed by
the translation to synchronous networks of deterministic timed automata in
\Cref{sec:fastmtl-to-ta}.  In \Cref{sec:future}, we explain how to integrate
future modalities to get a translation from $\MTLfp$ to generalized timed automata with
future clocks.  
\Cref{sec:gta-liveness} describes our adaptation of Couvreur's SCC algorithm to 
the setting of GTAs.
\Cref{sec:experiments} contains our experimental evaluation and
comparisons and we end with a brief conclusion in \Cref{sec:conclusion}.  

To improve the flow and readability of the paper, we have moved some proofs and technical
details to the Appendix.  Moreover, some benchmark details and extended results are
presented in the Appendix.

\section{Synchronous network of timed automata with shared variables}
\label{sec:sync-network-TA}

We define the model of network of timed systems that we use in this paper.  The
presentation is different from previous
works~\cite{DBLP:conf/sfm/BehrmannDL04,DBLP:conf/atva/BouyerHR06,DBLP:conf/concur/GovindHSW19,AGGS-CONCUR24},
and is geared towards providing a clean and concise explanation of our logic-to-automata
translation.  Before giving the formal syntax, we highlight the main aspects in our
definition.

Each component in the network is an automaton with a finite set of locations, real valued variables
(called \emph{clocks}) and finitely valued variables (mainly boolean variables and the
\emph{state} variable ranging over the finite set of locations).
The semantics of the network is \emph{synchronous}: either time elapses in all components
of the network (all clocks are increased by the same amount), or all components execute
simultaneously an instantaneous transition.
The components of the system communicate via \emph{shared variables} (clocks, states, boolean 
variables, \ldots). The paradigm is concurrent-read owner-write (CROW). Hence, each 
variable is \emph{owned} by a unique component of the system. A component may read all 
variables of the system, but it may only modify its owned variables.

A \emph{transition} of a component $\cA$ is of the form $(\ell,\prog,\ell')$ where
$\ell,\ell'$ are locations of $\cA$ and $\prog$ is an \emph{instantaneous} program, which
is a sequence of \emph{tests} and \emph{updates}. 
For instance, if $x,y$ are clocks
owned by $\cA$ and $\outv$ is a boolean variable owned by $\cA$, then a purely local
program could be $\now\cA{x}\in(3,5];\now\cA{y}:=2;\now\cA\outv:=0$ where
$\now\cA{x}\in(3,5]$ is a test which succeeds if the current value of clock $x$ is within
the interval, and $\now\cA{y}:=2$ and $\now\cA\outv:=0$ are updates.  Tests may also use
variables owned by other components.  For instance, if component $\cB$ owns clock $z$,
then a program of $\cA$ could be of the form
$$
\now\cA{x}\in(3,5] \wedge \post\cB\outv=1 \wedge \pre\cB{z}>2 
\wedge \post\cB{\state}=\ell_{1}; \now\cA{y}:=2; \now\cA\outv:=0
$$
where $\now\cB\state$ means the state of $\cB$.
Since we use a synchronous semantics, all components of the network simultaneously
take a transition and execute the associated programs. By default, a variable such as 
$\pre\cA{x}$ or $\pre\cB{z}$ used in a test refers to the value of the variable 
\emph{before} the synchronous transition is taken.
It is sometimes convenient to refer to the value a variable will have \emph{after} the 
synchronous transition is taken. We do so by writing $\post\cB\outv$ or $\post\cB\state$.

We give now the formal syntax and semantics for a network of timed systems. 
A concrete example follows in \Cref{fig:example1} and \Cref{ex:network-formula}.

\begin{definition}[Timed automaton with shared variables]\label{def:ta}
  A timed automaton with shared variables is a tuple
  $\cA=(Q,X,V,\Delta,\init)$ where $\cA$ is the name,
  $Q$ is a finite set of locations over which ranges the state variable $\now\cA\state$,
  $X$ is a finite set of clock variables \emph{owned} by $\cA$, 
  $V$ is a finite set of boolean variables owned by $\A$, 
  $\Delta$ is the finite set of transitions, 
  and $\init$ is an initial condition for the variables owned by $\cA$.
  
  A transition is a triple $(\ell,\prog,\ell')$ where $\ell,\ell'\in Q$ are locations and 
  $\prog$ is an instantaneous program. 
  A program is a sequence of tests and updates.
  An atomic update for a boolean variable $v\in V$ owned by $\cA$ has the form
  $\now\cA{v}:=c$ (or simply $v:=c$ when $\cA$ is clear from the context) with
  $c\in\{0,1\}$.
  An atomic update for a clock variable $x\in X$ owned by $\cA$ has the form $\now\cA{x}:=c$
  (or simply $x:=c$) with $c\in\overline{\mathbb{N}}=\mathbb{N}\cup\{+\infty\}$.
  We assume that each variable is updated at most once in a program.
  An atomic test for a clock variable $x\in X$ owned by $\cA$ has the form $\now\cA{x}\in I$
  (or simply $x\in I$) where $I$ is an integer bounded interval whose end points come from
  $\overline{\mathbb{N}}$ (e.g., $[0,3)$ or
  $[2,+\infty)$ or $(2,5]$, etc).
  $\cA$ may also test a clock variable owned by another component with $\pre\cN{x}\in I$
  or $\post\cN{x}\in I$ where $\cN$ is a name and $I$ is an integer bounded interval.
  Similarly, an atomic test for a boolean variable $b$ owned by $\cA$ has the form $b=0$ or
  $b=1$ (also written $b$ and $\neg b$) and if the variable is not owned by $\cA$ we use
  $\pre\cN{b}$ or $\post\cN{b}$.
  We write $\pre\cN\state=\ell$, $\pre\cN\state\neq\ell$, $\post\cN\state=\ell$ and
  $\post\cN\state\neq\ell$ for testing whether component $\cN$ is in location $\ell$ or not.
  Finally, a test in a program is a boolean combination of atomic tests.
  
  The initialisation $\init$ of $\cA$ is a conjunction of atomic tests, one for each
  variable of $\cA$, including the state variable.  For instance, $\init$ could be
  $\state=\ell_{0} \wedge y=+\infty \wedge \outv=0$.
\end{definition}

\begin{definition}[Synchronous network of timed automata]\label{def:networks-ta}
  A network is a tuple $\overline{\cA}=(\cA_{1},\ldots,\cA_n)$ where each
  $\cA_i=(Q_i,X_i,V_i,\Delta_i,\init_i)$ is a timed automaton with shared variables.
  The network is \emph{closed} if the references to external variables in programs are all of the 
  form $\pre{\cA_i}{\state}$ or $\post{\cA_i}{\state}$, 
  $\pre{\cA_i}{x}$ or $\post{\cA_i}{x}$ with $x\in X_{i}$, 
  $\pre{\cA_i}{v}$ or $\post{\cA_i}{v}$ with $v\in V_{i}$,
  for some $1\leq i\leq n$.
  The network is \emph{partially ordered} if there is a partial order $\preceq$ on its 
  components such that whenever component $\cA$ refers to a variable owned by component 
  $\cB$, we have $\cB\preceq\cA$. In this paper, we will only use partially ordered 
  networks.
\end{definition}

\begin{figure}
  \centering
  \raisebox{2mm}{\includegraphics[page=1,scale=0.9]{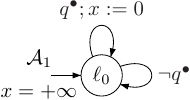}}\hfil
  \includegraphics[page=2,scale=0.9]{gastex-figures-pics.pdf}
  
  \medskip
  \includegraphics[page=3,scale=0.9]{gastex-figures-pics.pdf}
  \hfil
  \raisebox{8mm}{\includegraphics[page=4,scale=0.9]{gastex-figures-pics.pdf}}
  \caption{Synchronous Network of timed automata with shared variables.  The model $\cM$
  is not drawn and owns the boolean variables $\Prop=\{p,q,r\}$.  To lighten the figures
  above, we simply write $\postM{p}$, $\postM{q}$ and $\postM{r}$ instead of
  $\post\cM{p}$, $\post\cM{q}$ and $\post\cM{r}$.
  The automata $\cA_1$, $\cA_2$, $\cA_3$ own respectively the clock variables $x$, $y$ and
  $z$.  The names and initial conditions are depicted with an incoming arrow to a
  location.  For instance, the initial condition of $\cA_2$ is $\state=\ell_{0}\wedge
  y=+\infty$.  Notice that $\cA_3$ refers to the post values of the state
  $\post{\cA_2}\state$ and clock $\post{\cA_2}{y}$ of $\cA_2$.  }
  \label{fig:example1}
\end{figure}

We describe now the semantics of a network.  A \emph{configuration} $C$ of
$\overline{\cA}$ is a valuation of all variables of the network: for all $1\leq i\leq n$,
we have $C(\now{\cA_i}\state)\in Q_{i}$,
$C(\now{\cA_i}x)\in\overline{\mathbb{R}}_+=\mathbb{R}_+\cup\{+\infty\}$ for $x\in X_{i}$, and
$C(\now{\cA_i}v)\in\{0,1\}$ for $v\in V_{i}$.
A configuration is \emph{initial} if it satisfies all initial conditions of its 
components. For instance, if $\init_2$ is 
$\state=\ell_{0} \wedge y=+\infty \wedge \outv=0$, then it is satisfied at $C$ if
$C(\now{\cA_2}\state)=\ell_{0}$, $C(\now{\cA_2}y)=+\infty$ and $C(\now{\cA_2}\outv)=0$.

The network is \emph{synchronous}, i.e., all components move simultaneously, and has two
types of transitions: time elapse and
discrete.  A \emph{time elapse} by some non-negative real number $\delta\geq0$ is written
$C\xra{\delta}C'$ where $C'$ coincides with $C$ on all variables other than clocks, and the
value of each clock is advanced by the same quantity $\delta$:
$C'(\now{\cA_i}x)=C(\now{\cA_i}x)+\delta$ for all $1\leq i\leq n$ and $x\in X_i$ (with
$+\infty+\delta=+\infty$).

A \emph{discrete} transition occurs in the network when each component executes
simultaneously a transition.  Let $\overline{t}=(t_{1},\ldots,t_{n})$ with
$t_{i}=(\ell_i,\prog_i,\ell'_i)\in\Delta_i$ for all $1\leq i\leq n$.  Let $C,C'$ be two
configurations.  There is a discrete transition $C\xra{\overline{t}}C'$ if
the following holds:
\begin{itemize}[left=0.1em]
  \item  $C(\now{\cA_i}\state)=\ell_{i}$ and $C'(\now{\cA_i}\state)=\ell'_{i}$
   for all $1\leq i\leq n$,

  \item  for all $i$ and $x\in X_{i}\cup V_{i}$, we have 
  $C'(\now{\cA_i}{x})=c$ if $\prog_i$ contains an update $\now{\cA_i}{x}:=c$ and
  $C'(\now{\cA_i}{x})=C(\now{\cA_i}{x})$ otherwise (recall that each variable is updated 
  at most once),

  \item  all tests occurring in the programs evaluate to true:
  \begin{itemize}[leftmargin=0pt, labelindent=0pt, itemindent=0pt]
    \item for all $i$ and clock $x\in X_{i}$ and integer bounded interval $I$, an atomic
    test $\pre{\cA_i}{x}\in I$ (resp.\ $\post{\cA_i}{x}\in I$) evaluates to true if
    $C(\now{\cA_i}{x})\in I$ (resp.\ $C'(\now{\cA_i}{x})\in I$),
  
    \item for all $i$ and boolean variables $v\in V_{i}$ and $c\in\{0,1\}$, an atomic test
    $\pre{\cA_i}{v}=c$ (resp.\ $\post{\cA_i}{v}=c$) evaluates to true if
    $C(\now{\cA_i}{v})=c$ (resp.\ $C'(\now{\cA_i}{v})=c$),
  
    \item for all $i$, an atomic test $\pre{\cA_i}\state=\ell$ (resp.\
    $\post{\cA_i}\state=\ell$) evaluates to true if $C(\now{\cA_i}\state)=\ell$ (resp.\
    $C'(\now{\cA_i}\state)=\ell$).
  \end{itemize}
\end{itemize}

\begin{example}\label{ex:network-formula}
  We consider the specification given by the sentence
  $$
  \Phi = \G_{\leq 100} \big( r \longrightarrow (\YP_{<2} q \wedge p\S_{\geq 5}q) \big) \,.
  $$
  We assume that the model $\cM$ owns the boolean variables $\Prop=\{p,q,r\}$, also called
  atomic propositions.  The sentence $\Phi$ says that during the first 100 units of time,
  whenever an event satisfies $r$ (i.e., $\cM$ sets $r$ to true, or $r$ was true and is 
  not updated by $\cM$) then
  some event less than 2 time units in the strict past satisfies $q$ (written as $\YP_{<2} q$), 
  and some other event at least 5 time units in the past satisfies $q$ and since then all 
  events satisfy $p$ (written as $p\S_{\geq 5}q$).
  
  We give in \Cref{fig:example1} a network of timed automata which implements the
  specification $\Phi$.  Component $\cA_1$ records with clock $x$ the time elapsed
  since the latest $q$ event occurred.  Clock $x$ is initialized to $+\infty$ and remains
  so until the first $q$ event is seen.  Then it is reset to $0$ whenever a $q$ event
  occurs.
  
  Component $\cA_2$ is in location $\ell_1$ iff $p$ since $q$ holds (there is a $q$ event
  in the past such that since then all events satisfy $p$).  When $\cA_2$ is in location
  $\ell_{1}$, clock $y$ records the time elapsed since the earliest \emph{witness} of $p\S
  q$, i.e., the earliest $q$ event such that since then all events satisfy $p$.
  
  Finally, $\cA_3$ is in location $\ell_{2}$ if the sequence of events seen so far 
  violates $\Phi$: within the first 100 time units, there is an $r$ event such that the 
  closest $q$ event is more than 2 time units in the strict past ($\pre{\cA_1}{x}\geq 2$) 
  or $p\S q$ does not hold ($\post{\cA_2}\state=\ell_{0}$) or the earliest witness of 
  $p\S q$ is less than 5 time units in the past.
  Component $\cA_3$ is in location $\ell_{1}$ if $\Phi$ is definitely satisfied, more 
  than 100 time units has elapsed without violating the property. Finally, $\cA_3$ 
  stays in location $\ell_{0}$ if the sequence of events seen so far satisfies $\Phi$ and 
  at most 100 time units has elapsed.
  \qed
\end{example}

\subsection{The model checking problem} 
In this paper, we are interested in the \emph{model checking} problem: given a model $\cM$
owning a set $\Prop$ of atomic propositions, and a specification $\Phi$ on $\Prop$, do all
\emph{behaviours} of $\cM$ satisfy $\Phi$.  Behaviours of a model are formalized as
\emph{timed words} over sets of atomic propositions.  The alphabet $\Sigma$ that we will
consider is $2^\Prop$, the set of subsets of $\Prop$.  A timed word over $\Sigma$, denoted
by $w\in(\Sigma\times \mathbb{R})^{\infty}$, is an \emph{infinite} sequence of
pairs of letter and timestamp of the form $(a_0,\tau_0)(a_1,\tau_1)(a_2,\tau_2)\dots$,
where each $a_i$ is a subset of $\Prop$ and $\tau_i$ is the timestamp of the $i$-th
letter; timestamps are assumed to be non-decreasing, i.e., $0\leq \tau_{0}\leq
\tau_{1}\leq \tau_{2} \cdots$.

A run $\rho$ of the model $\cM$ is an alternating sequence of delay and discrete
transitions, starting from an initial configuration.  Formally, a run is an infinite
sequence
\[
\rho:= C_0 \xrightarrow{\delta_0} C'_0 \xrightarrow{\overline{t}_0} C_1 
\xrightarrow{\delta_1} C'_1 \xrightarrow{\overline{t}_1} C_2 \cdots
\]
where $C_0$ is an initial configuration, each $\delta_i \geq 0$ is a time delay, and each
$\overline{t}_i$ is a tuple of transitions (one from each component of the model) as
defined above, such that each step respects the semantics of time elapse and discrete
transitions.  The timed word $w_\rho = (a_0, \tau_0) (a_1, \tau_1) \cdots$ associated to
the above run $\rho$ is given as follows:
for all $i \ge 0$
\begin{itemize}
  \item $\tau_i = \delta_0 + \delta_1 + \cdots + \delta_i$ (the timestamp of occurrence of
  $\overline{t}_i$) and

  \item $a_i$ is the set of atomic propositions $p$ such that $C_{i+1}(p) = 1$ (the
  propositions that are true \emph{after} executing $\overline{t}_i$).
\end{itemize}
The timed word $w_\rho$ is a behaviour of the model $\cM$. 

The \emph{model checking problem} asks whether all (infinite, non-zeno) behaviours 
generated by all possible runs of the model $\cM$ satisfy a given specification $\Phi$,
written as $\cM \models \Phi$.  In this work, we are interested in specifications given by
fragments of $\MTL$, which we will define in the later sections.  To solve the model
checking problem, we will construct from the specification $\Phi$ a network of timed
automata $\cA=(\cA_1,\ldots,\cA_n)$.
Then, we will consider the \emph{closed} network $(\cM,\cA)$ and check whether 
some (bad) cycle may be iterated infinitely often.

For instance, to check whether a given model $\cM$ satisfies the specification $\Phi$ from
\Cref{ex:network-formula}, we consider the \emph{closed} network
$\cN=(\cM,\cA_1,\cA_2,\cA_3)$ where the automata $\cA_1,\cA_2,\cA_3$ are given in
\Cref{fig:example1}.  Then, $\cM\not\models\Phi$ iff a configuration $C$ of $\cN$ where
$\cA_3$ is in location $\ell_2$ is visited infinitely often.

We may also check whether a specification is \emph{satisfiable} by considering a 
\emph{universal} model allowing all sequences of events. This is achieved by considering, 
for each atomic proposition $p\in\Prop$, an automaton $\cM_p$ as shown in \Cref{fig:example1}.
Then, the model is the network $\cM=(\cM_p)_{p\in\Prop}$.

\section{The $\detMTL$ fragment}
\label{sec:detMTL}
We present a fragment of $\MTL$ for which we can construct a network $(\cA_1, \cA_2,
\dots, \cA_n)$ such that each $\cA_i$ is deterministic.  For the formula in
\Cref{ex:network-formula}, the automata $\cA_1, \cA_2$ and $\cA_3$ shown in
\Cref{fig:example1} are deterministic.  When $(\cA_1, \dots, \cA_n)$ are deterministic,
the product with the model $(\cM, \cA_1, \dots, \cA_n)$ does not further ``multiply'' the
runs of $\cM$: for every run $\sigma$ of $\cM$, there is at most one run $\rho$ of $(\cM,
\cA_1, \dots, \cA_n)$ whose projection leads to $\sigma$.  On the other hand, if some
$\cA_i$ is non-deterministic, the run $\sigma$ can branch out into multiple runs in the
product, leading to an explosion in the state-space (of the zone graph that is computed).

The set of $\MTL$ formulae over the atomic propositions $\Prop$ is defined as
$$
\varphi := p \mid \varphi \wedge \varphi  \mid  \neg \varphi  \mid \Y_{I} \varphi \mid 
\Next_{I} \varphi  \mid \varphi \U_{I} \varphi \mid \varphi \S_{I} \varphi
$$ 
where $p \in \Prop$, and $I$ is an interval with end-points from
$\overline{\Nat}=\Nat\cup\{\infty\}$.  Note that, our definition of $\MTL$ is with both
past and future modalities.  The \emph{pointwise semantics} of $\MTL$ formulae is defined
inductively as follows.  A timed word $w = (a_0,\tau_0)(a_1,\tau_1)(a_2,\tau_2) \cdots $
is said to satisfy the $\MTL$ formula $\varphi$ at position $i\geq0$, denoted as $(w,i)
\models \varphi$ if (omitting the classical boolean connectives)
\begin{itemize}[left=0.1em]
  \item $(w,i) \models p$ if $p \in a_i$
  \item $(w,i) \models \Y_{I} \varphi$ if $i>0$, $(w,i-1) \models \varphi$ and $\tau_{i} - \tau_{i-1} \in I$.  
  \item $(w,i) \models \Next_{I} \varphi$ if $(w,i+1) \models \varphi$ and $\tau_{i+1} - \tau_{i} \in I$.  

  \item $(w,i) \models \varphi_1 \U_{I} \varphi_2$ if there exists $j\geq i$  s.t.\ 
  $(w,j) \models \varphi_2$, $\tau_{j} - \tau_{i} \in I$
  and $(w,k) \models \varphi_1$ for all $i \leq k < j$.
  
  \item $(w,i) \models \varphi_1 \S_{I} \varphi_2$ if there exists $0\leq j \leq i$
  s.t.\ $(w,j)\!\models\!\varphi_2$, $\tau_{i} - \tau_{j} \in I$ and
  $(w,k)\!\models\!\varphi_1$ for all $j< k\leq i$.
\end{itemize}  
A word $w$ \emph{initially} satisfies $\varphi$, written as $w \models \varphi$ if $(w, 0)
\models \varphi$.  Finally, a model $\cM$ satisfies $\MTL$ formula $\varphi$, denoted as
$\cM \models \varphi$ if $w_\rho \models \varphi$ for every run $\rho$ of $\cM$.  We
rewrite the derived operators $\Always, \Eventually, \Past, \History$ in terms of $\U$ or
$\S$, in the standard manner:
$\F_{I}\varphi=\true\U_{I}\varphi$, 
$\G_{I}\varphi=\neg\F_{I}\neg\varphi$,
$\Past_{I}\varphi=\true\S_{I}\varphi$ and
$\History_{I}\varphi=\neg\Past_{I}\neg\varphi$.

\begin{definition}[$\detMTL$]\label{def:fast-mtl}
  The fragment $\detMTL$ consists of formulas where the outermost temporal operators
  are future operators $\X_{I}$, $\U_{I}$
  and may use arbitrary intervals $I$, while all inner temporal
  operators are restricted to past operators.  Inner $\mathsf{Since}$ operators must use
  non-punctual intervals (either $[0,0]$ or non-singleton intervals).  Formally, formulas
  are built as follows:
  \begin{align*}
    (\text{Sentences})\quad 
    \Phi & := p \mid \Phi \wedge \Phi \mid \neg \Phi \mid \Next_{I} \varphi \mid \varphi \U_{I} \varphi 
    \\
    (\text{Point-formulas}) \quad 
    \varphi & := p \mid \varphi \wedge \varphi \mid \neg \varphi \mid \Y_I \varphi \mid \varphi \S_{I'} \varphi
  \end{align*}
  where $I,I'$ are integer-bounded intervals, and $I'$ is either $[0,0]$ or a non-singleton
  interval.
\end{definition}

Sentences are evaluated at the first position $(w, 0)$ of a word $w$, whereas
point-formulas can be checked anywhere in the word.  The formula in
\Cref{ex:network-formula} is a sentence in the $\detMTL$ fragment: the derived
$\G_I$ operator induces an outer $\U_I$ formula, and all the inner operators are past.
The formula $\YP_{<2} q$ can be rewritten as
$\Past_{(0,2)}q \vee \Y_{[0,0]}\Past_{[0,0]}q$.

\section{From $\detMTL$ to network of deterministic timed automata}
\label{sec:fastmtl-to-ta}

For a $\detMTL$ sentence $\Phi$ we build a deterministic network $\cN_\Phi$ and specify a
set $\bad_{\Phi}$ of bad configurations.  Our objective is to achieve a construction that
entails the following theorem for $\detMTL$ sentences.

\begin{theorem}\label{thm:fast-mtl-theorem}
  Let $\Phi$ be a $\detMTL$ sentence and $\cM$ a timed model. Then:
  $\cM \not \models \Phi$ iff there exists an infinite run of the closed network $(\cM,
  \cN_\Phi)$ which eventually remains in bad configurations.
\end{theorem}

Given a sentence $\Phi$, we construct the following deterministic automata:
\begin{align*}
  \Aisatp{p} \qquad \AisatX \qquad \AisatU 
\end{align*}
for every atomic proposition $p \in \Prop$ that is not under the scope of a temporal
operator, and every subformula $\X_I$ and $\U_I$ appearing in the sentence.  The automaton
$\Aisatp{p}$ takes the atomic proposition $p$ as input.  The automaton $\AisatX$ takes as
input the boolean variable $\argone$ corresponding to the argument of $\X_I$.  Automaton
$\AisatU$ takes two boolean variables $\argone,\argtwo$ as inputs corresponding to the
first and second arguments of $\U_I$.  Notice that these operators appear only at the
outermost level of $\Phi$ and will be evaluated at the first position in a word.  The
automata for these outermost operations will be called \emph{initial satisfiability
automata}, and are superscripted with an $isat$.  Apart from these, we build an automaton
$\Ainit$ which resets a clock $\xinit$ at the first action.  All the $isat$ automata make
use of this single clock $\xinit$ in their tests.  The $isat$ automata do not own any
other clocks.  Finally, for every temporal point-formula $\varphi$ appearing as
a subformula of $\Phi$, we build a deterministic automaton $\cA_\varphi$.  For these
point-formula automata we also allow for sharing of clocks, which we explain later.

\paragraph*{Bad configurations.} The collection of all the $isat$ automata and the
point-formula automata for $\Phi$ form the network $\cN_\Phi$.  States of each of the
$isat$ automata will be marked either $\true$ or $\false$.  A configuration $C$ of
$\cN_\Phi$ is bad if $\Phi$ evaluates to $\false$ when we replace every outermost $p$, and
operators $\X_I$, $\U_I$ of $\Phi$ by the truth value of the state the corresponding
$isat$ automata are in.  For example if $\Phi = p \land \X_I q$, a configuration $C_1$
where $\Aisatp{p} = \true$ and $\mathcal{A}^{isat}_{\X_I q} =\true$ is \emph{not bad},
whereas $C_2$ with $\Aisatp{p} = \false$ will be bad.  The set of bad configurations of
$\Phi$ will be called $\bad_\Phi$.  This notion naturally extends to the product $(\cM,
\cN_\Phi)$: a configuration of the product is bad if its projection to $\cN_\Phi$
evaluates to $\false$.

\begin{wrapfigure}{r}{0.4\textwidth}
  \centering
  \includegraphics[page=5]{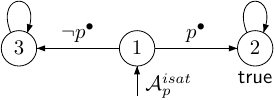}
  \caption{Automaton $\Aisatp{p}$.}
  \label{fig:isat-prop}
\end{wrapfigure}
We describe the initial satisfiability automata in \Cref{sec:init-sat} and show
some local correctness properties.  The automata for past operators are presented in
\Cref{sec:past-automata,sec:since-one-sided,sec:since-two-sided}.  Finally, in
\Cref{sec:proof-of-fast-mtl}, we present the proof of
\Cref{thm:fast-mtl-theorem} which shows the global correctness of the construction,
by making use of the intermediate local correctness lemmas shown in the previous sections.

\subsection{Initial satisfiability of sentences} \label{sec:init-sat}

\Cref{fig:isat-prop} depicts the automaton $\Aisatp{p}$ which takes as input the atomic
proposition $p$.  The property that it satisfies is stated in the following lemma, whose
proof is straightforward.

\begin{restatable}{lemma}{lemisatp}
    \label{lem:isat-p}
  For every timed word $w = (a_0, \tau_0) (a_1, \tau_1) \cdots$ over $\Prop$ with
  $p\in\Prop$, there is a unique run $C_0 \xrightarrow{\delta_0, \overline{t}_0} C_1
  \xra{\delta_1, \overline{t}_1} C_2 \cdots$ of $\Aisatp{p}$ over $w$. Moreover,
    $w \models p$ iff for all $i \ge 1$, configuration $C_i$ satisfies
    $\Aisatp{p}.\state = 2$.
\end{restatable}

\begin{figure}
  \centering
  \raisebox{3mm}{\includegraphics[page=6,scale=1]{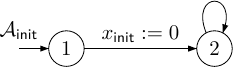}}\hfil
  \includegraphics[page=7,scale=1]{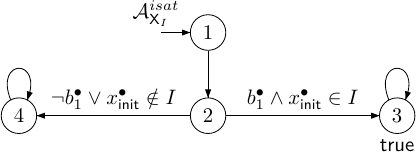}
  \caption{Left: Clock $\xinit$ is reset on the first event and measures time elapsed
  for all the initial satisfiability automata ($\AisatX$ and $\AisatU$).  \\
  Right: Automaton $\AisatX$ for initial satisfiability for $\X_{I}\argone$.
  The boolean variable $\argone$ is not owned by $\AisatX$ and stands for the argument of 
  $\X_I$. In $\AisatX$, we simply write $\xinit$ instead of $\pre\Ainit\xinit$.
  The automaton is deterministic and complete.  If a run reaches state 3 (resp.\ 4) then
  $\X_{I}\argone$ was initially true (resp.\ false).
  }
  \label{fig:init-sat}
\end{figure}

\Cref{fig:init-sat} describes $\Ainit$ and $\AisatX$.  The invariant satisfied by
$\AisatX$ is explained in the caption of \Cref{fig:init-sat} and formalized below.

\begin{restatable}{lemma}{lemisatX}
    \label{lem:isat-X}
  Let $\Prop' = \Prop \uplus \{\argone\}$ be an augmented set of atomic propositions.  For
  every word $w = (a_0, \tau_0) (a_1, \tau_1) \cdots$ over $\Prop'$, there is a unique run
  $C_0 \xrightarrow{\delta_0, \overline{t}_0} C_1 \xra{\delta_1, \overline{t}_1} C_2
  \cdots$ of $(\Ainit, \AisatX)$ over $w$.  Moreover,
    $w \models \X_I \argone$ iff for all $i \ge 2$, configuration $C_i$ satisfies
    $\AisatX.\state = 3$.
\end{restatable}

\begin{figure}
  \centering
  \includegraphics[page=8,scale=0.8]{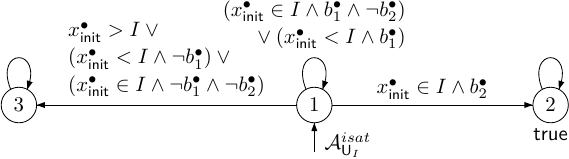}
  \caption{
  Automaton $\AisatU$ for the initial satisfiability for $\argone \U_{I} \argtwo$ 
  (with $I\neq\emptyset$).
  The boolean variables $\argone,\argtwo$ are not owned by $\AisatU$ and stand for the 
  left and right arguments of $\U_I$. 
  The automaton is deterministic and complete. 
  We write $\xinit<I$ (resp.\ $\xinit>I$) to state that $\xinit$ is smaller (resp.\ larger) than all
  values in $I$.  For instance, if $I=[b,c)$ with $b<c$ then $\xinit<I$ means $\xinit<b$ and
  $\xinit>I$ means $\xinit\geq c$.  
  If a run reaches state 2 (resp.\ 3) then we know for sure that $\argone \U_{I} \argtwo$ was initially 
  true (resp.\ false). 
  State 1 means we still do not know whether $\argone \U_{I} \argtwo$ was initially true. 
  If the run stays forever in state 1, $\argone \U_{I} \argtwo$ was initially false.
  }
  \label{fig:initsat-until}
\end{figure}

Finally, \Cref{fig:initsat-until} gives the deterministic automaton for the initial
satisfiability of $\U_I$.  Correctness of the automaton is formalized below.

\begin{lemma}\label{lem:isat-U}
  Let $\Prop' = \Prop \uplus \{\argone, \argtwo\}$ be an augmented set of
  propositions.  For every word $w = (a_0, \tau_0) (a_1, \tau_1) \cdots$ over $\Prop'$,
  there is a unique run $C_0 \xrightarrow{\delta_0, \overline{t}_0} C_1 \xra{\delta_1,
  \overline{t}_1} C_2 \cdots$  of $(\Ainit, \AisatU)$ over $w$. Moreover,
    $w \models \argone \U_I \argtwo$ iff there exists an $i \ge 1$ s.t.\ for all 
    $k\geq i$, configuration $C_{k}$ satisfies $\AisatU.\state = 2$,
\end{lemma}

\begin{proof}
  Suppose $w \models \argone \U_I \argtwo$.  Then, (1) there exists an $i \ge 0$ such that
  $\argtwo \in a_i$, (2) for all $0 \le j < i$, we have $\argone \in a_j$, and (3) $\tau_i
  - \tau_0 \in I$.

  We assume that $i\geq0$ is minimal with the above property.
  Hence, in the run of $w$, 
  \begin{itemize}
    \item  transitions $\overline{t}_j$ for $0 \le j < i$ are self-loops on state $1$.
    Indeed, $\postM{\argone}$ is true at $C_{j+1}$ and either $\tau_j-\tau_0<I$ or
    $\tau_j-\tau_0\in I$ and $\postM{\argtwo}$ is false at $C_{j+1}$ due to the
    minimality of $i$.
  
    \item  transition $\overline{t}_i$ goes to state 2 since $\postM{\argtwo}$
    is true at $C_{i+1}$ and $\tau_i - \tau_0 \in I$.
  \end{itemize}  
  Hence configuration $C_{i+1}$ satisfies $\AisatU.\state = 2$.

  Conversely, suppose the automaton moves to state $2$ for the first time after reading
  the prefix $(a_0, \tau_0) (a_1, \tau_1) \cdots (a_i, \tau_i)$.  Since $\AisatU$
  starts at state $1$, the run loops around $1$ on all events $a_0$ up to $a_{i-1}$ and
  then a transition is made to state $2$.  In the loop $\argone$ is true, and in the
  transition that moves to $2$, we have $\argtwo$ to be true.  The transition to $2$ also
  checks for $\xinit \in I$.  Since $\xinit$ is reset at $a_0$, we get $\tau_i - \tau_0
  \in I$.  Hence, the run shows the three properties listed above and $w\models \argone
  \U_I \argtwo$.
\end{proof}

\subsection{Past operators}\label{sec:past-automata}

\subsubsection{Yesterday operator.}

\begin{figure}[tb]
  \centering
  \includegraphics[page=10,scale=1]{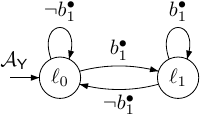}
  \hfil
  \raisebox{2mm}{\includegraphics[page=9,scale=1]{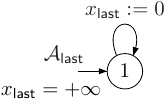}}  
  \caption{Left: 
  Automaton $\AY$ for the untimed Yesterday operator $\Y\argone$.
  The boolean variable $\argone$ is not owned by $\AY$ and stands for the argument of $\Y$.  
  \\
  Right:
  Sharer automaton $\Alast$ with clock $\xlast$ used by all Yesterday operators.  
  }
  \label{fig:Yesterday-automaton}
\end{figure}

The automaton for the Yesterday operator $\Y$, depicted in
\Cref{fig:Yesterday-automaton}~(left), has two states $\ell_{0}$ and $\ell_{1}$.  The
states remember whether the argument $\argone$ of $\Y$ was true or not at the previous
letter: all transitions reading a $\argone$ enter $\ell_{1}$, and the ones that read
$\neg\argone$ go to $\ell_{0}$.  To incorporate the timing constraint, we use a single
state automaton $\Alast$, depicted in \Cref{fig:Yesterday-automaton}~(right), which resets
a clock $\xlast$ in every step.  Now, for an arbitrary interval constraint $I$, the
current position satisfies $\Y_{I}\argone$ if the previous letter had a $\argone$ and the
value of clock $\xlast$ (which was reset at the previous position) when the current
position is read falls in the interval $I$.

\begin{restatable}{lemma}{lemYesterday}
    \label{lem:Yesterday}
  Let $\Prop' = \Prop \uplus \{\argone\}$.  For every word $w = (a_0, \tau_0) (a_1,
  \tau_1) \cdots$ over $\Prop'$, we have a unique run 
  $C_0 \xra{\delta_0} C'_0 \xra{\overline{t}_0} C_1 \xra{\delta_1} C'_1
  \xra{\overline{t}_1} C_2 \cdots$ of $(\Alast,\AY)$.
  Moreover, 
  \begin{itemize}
    \item for all $i\ge 0$, $(w,i) \models \Y\argone$ iff
    $C_{i}(\AY.\state)=C'_{i}(\AY.\state)=\ell_1$,
    
    \item for all $i\ge 1$, $\tau_{i}-\tau_{i-1}\in I$ iff $C'_{i}(\Alast.\xlast)\in I$.
  \end{itemize}
  Let $\outv_{\Y_{I}}=(\AY.\state=\ell_1 \wedge \Alast.\xlast\in I)$.  Then, for all
  $i\geq0$, we may check if $(w,i)\models\Y_{I}\argone$ with the test $\outv_{\Y_{I}}$
  during transition $\overline{t}_{i}$ from $C'_{i}$ to $C_{i+1}$.
\end{restatable}

\subsubsection{Since operator with a one-sided interval.}
\label{sec:since-one-sided}

\begin{figure}[b]
  \centering
    \raisebox{6.5mm}{\includegraphics[page=11]{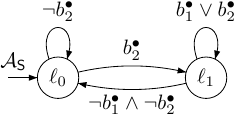}}
  \hfill
  \raisebox{9mm}{\includegraphics[page=12]{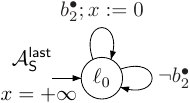}}
  \hfill
  \includegraphics[page=13]{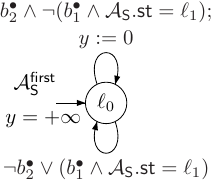}
  \caption{Left: Automaton for the untimed Since operator $\argone \S \argtwo$. \\ 
  Middle: Automaton recording with clock $x$ the time since the last (most recent) occurrence of $\argtwo$. \\
  Right: Automaton recording with clock $y$ time since the first (earliest) $\argtwo$ which 
  may serve as a witness of $\argone \S_{I} \argtwo$.
  }
  \label{fig:Since-1-sided-automata}
\end{figure}

The formula $\argone \S_{I} \argtwo$ essentially asks that there exists a position
$j$ in the past, such that $\argtwo$ is true at that position, and at all positions
between $j$ and the current position, $\argone$ is true, and the time elapsed between
the position $j$ and the current position lies in the interval $I$.  
Note that there could be multiple candidates for the position $j$.
Let us first consider the untimed version of the Since operator given in
\Cref{fig:Since-1-sided-automata}~(left).  The transition from $\ell_{0}$ to $\ell_{1}$ on
$\argtwo$ is taken with the first $\argtwo$ that could be a witness.  When the automaton
remains in $\ell_{1}$, each time a $\argtwo$ is seen, the automaton encounters a fresh
witness.

When we move to the timed setting, 
some of these witnesses may satisfy the timing constraint 
while others may not, and so this information must be tracked carefully.
We will first consider the simpler case where the intervals are
one-sided, i.e., of the form $[0,c)$, $[0,c]$, $[b,+\infty)$ or $(b,+\infty)$.

When the interval is one-sided, it is sufficient to track the two \emph{extreme} $\argtwo$
witnesses: the earliest witness and the last possible witness at each point.
The idea is simply that if the earliest witness does not satisfy a lower-bound constraint,
then no later witness could satisfy it.  Analogously, if the last witness does not satisfy
an upper-bound, no earlier witness could have satisfied it.

In the automaton $\ASlast$ of \Cref{fig:Since-1-sided-automata}~(middle),
clock $x$ maintains the time since the latest $\argtwo$ (which may serve as a witness
for $\argone\S\argtwo$ when $\AS$ is in state $\ell_{1}$). 
For an \emph{upper-bound} interval $\S_{I}$ where $I=[0,c)$ or $I=[0,c]$, we check the
value of $x$, i.e., the time since the last $\argtwo$.  The output is $1$ if the target
state of $\AS$ is $\ell_{1}$ and $x\in I$.

\begin{restatable}{lemma}{lemsinceub}
    \label{lem:since-upper-bound}
  Let $\Prop' = \Prop \uplus \{\argone, \argtwo\}$.  For every word $w = (a_0, \tau_0)
  (a_1, \tau_1) \cdots$ over $\Prop'$ there is a unique run 
  $C_0 \xra{\delta_0} C'_0 \xra{\overline{t}_0} C_1 \xra{\delta_1} C'_1
  \xra{\overline{t}_1} C_2 \cdots$ of $(\AS,\ASlast)$.
  Moreover, for all $i \ge 0$:
  \begin{enumerate}[left=0.1em]
    \item $(w,i) \models \argone\S\argtwo$ iff $C_{i+1}(\AS.\state)=\ell_1$,
    
    \item When $C_{i+1}(\AS.\state)=\ell_1$ we have $C_{i+1}(\ASlast.x)=\tau_{i}-\tau_{j}$
    where $j\leq i$, $\argtwo\in a_{j}$ and $\argtwo\notin a_{k}$ for all $j<k\leq i$.

    \item Let $I$ be an upper-bound interval and $\outv_{\S_{I}}=(\post{\AS}\state=\ell_1
    \wedge \post\ASlast{x}\in I)$.  Then, for all $i\geq0$, we may check if
    $(w,i)\models\argone\S_{I}\argtwo$ with the test $\outv_{\S_{I}}$ during transition
    $\overline{t}_{i}$ from $C'_{i}$ to $C_{i+1}$.
  \end{enumerate}
\end{restatable}

In the automaton $\ASfirst$ of \Cref{fig:Since-1-sided-automata}~(right),
clock $y$ maintains the time since the earliest $\argtwo$ which may serve as a witness
for $\argone \S \argtwo$:
clock $y$ is reset when reading $\argtwo$ (the current position is a witness of 
$\argone\S\argtwo$) and there are no earlier witnesses (we do not have 
$\argone\wedge\Y(\argone\S\argtwo)$).
For a \emph{lower-bound} interval $\S_{I}$ where $I=[b,+\infty)$ or $I=(b,+\infty)$ with
$b\in\mathbb{N}$, we check the value of $y$, i.e., the time since the earliest $q$.  The
output is $1$ if the target state of $\AS$ is $\ell_{1}$ and $y\in I$.  

\begin{restatable}{lemma}{lemsincelb}
    \label{lem:since-lower-bound}
  Let $\Prop' = \Prop \uplus \{\argone, \argtwo\}$.  For every word $w = (a_0, \tau_0)
  (a_1, \tau_1) \cdots$ over $\Prop'$ there is a unique run 
  $C_0 \xra{\delta_0} C'_0 \xra{\overline{t}_0} C_1 \xra{\delta_1} C'_1
  \xra{\overline{t}_1} C_2 \cdots$ of $(\AS,\ASfirst)$.
  Moreover, for all $i \ge 0$:
  \begin{enumerate}
    \item $(w,i) \models \argone\S\argtwo$ iff $C_{i+1}(\AS.\state)=\ell_1$,

    \item When $C_{i+1}(\AS.\state)=\ell_1$ we have 
    $C_{i+1}(\ASfirst.y)=\tau_{i}-\tau_{j}$ where $j$ is the earliest witness at $i$ of 
    $\argone\S\argtwo$: $j\leq i$, $\argtwo\in a_{j}$ and 
    $\argone\in a_{k}$ for all $j<k\leq i$, and for all $0\leq j'<j$ either 
    $\argtwo\notin a_{j'}$ or $\argone\notin a_{k'}$ for some $j'<k'\leq i$.

    \item Let $I$ be a lower-bound interval and $\outv_{\S_{I}}=(\post{\AS}\state=\ell_1
    \wedge \post\ASfirst{y}\in I)$.  Then, for all $i\geq0$, we may check if
    $(w,i)\models\argone\S_{I}\argtwo$ with the test $\outv_{\S_{I}}$ during transition
    $\overline{t}_{i}$ from $C'_{i}$ to $C_{i+1}$.
  \end{enumerate}
\end{restatable}

\subsubsection{Since with a two-sided interval.}
\label{sec:since-two-sided}

The general case with a nonempty non-singleton interval $I$ is much harder.  We may have a
position satisfying $\argone \S_{I} \argtwo$ but neither the earliest nor the last witness
of $\argtwo$ satisfy both the upper and lower constraints of $I$.  In such a case, there
could still be another witness in between which does satisfy the constraint.  We need to
change our perspective and look at each $\argtwo$ event as a potential witness for
$\argone \S_I \argtwo$.  However, we cannot simply start a new clock at each $\argtwo$
event, as we only have finitely many clocks at our disposal.  We discuss below our
solution for this general case.

Let $I$ be a non-singleton interval with lower bound $b$ and upper bound $c$.  We assume
$0\notin I$ and $c\neq+\infty$, otherwise it is a simple upper-bound or lower-bound
constraint which was already solved in \Cref{sec:since-one-sided}.  Since both $b,c$ are
integers and $I$ is non-singleton, we have $c-b\geq 1$.

\begin{wrapfigure}{r}{0.5\textwidth}
  \centering
  \resizebox{0.5\textwidth}{!}{
  \begin{tikzpicture}
    \draw [thin, gray] (0,0) to (8,0); \fill[red] (0,0) circle
    (1.5pt); \draw (0,0) to (0,0.2); \node at (0, 0.4) {\scriptsize
      $t$}; \draw (1.5, 0) to (1.5,0.2); \node at (1.5, 0.4)
    {\scriptsize $t + (c - b)$}; \draw (1.2, 0) to (1.2, -0.2); \node
    at (1.2, -0.4) {\scriptsize $t'$}; \fill[green] (1.2, 0) circle
    (1.5pt);

    \begin{scope}[xshift=4cm]
      \draw[red] (0,0.1) to (1.5, 0.1); \draw (0,0) to (0,0.2); \node
      at (0, 0.4) {\scriptsize $t+b$}; \draw (1.5, 0) to (1.5,0.2);
      \node at (1.5, 0.4) {\scriptsize $t + c$}; \draw (1.2, 0) to
      (1.2, -0.2); \node at (1.2, -0.4) {\scriptsize $t'+ b$};
      \draw[green] (1.2,-0.1) to (2.7,-0.1); \draw (2.7, 0) to (2.7,
      -0.2); \node at (2.7, -0.4) {\scriptsize $t'+ c$};
    \end{scope}
    
  \end{tikzpicture}
  }  
  \caption{Illustrating the idea for $S_I$}
  \label{fig:app:since-idea}
\end{wrapfigure}
The idea of our approach is illustrated in \Cref{fig:app:since-idea}.  
Assume there is a $\argtwo$ at $t$, represented by the red dot.
Assuming $I=[b,c]$, this $\argtwo$ is a witness for the events happening in the interval $[t+b,
t+c]$ (marked by the red line), assuming all the intermediate points contain $\argone$.
Now, consider the last $\argtwo$-event in the interval $[t, t + (c - b)]$.
This is marked by the green dot in the figure, at time $t'$. This
$\argtwo$ event acts as a witness for the events in $[t'+b, t'+c]$, marked
by the green line in the figure. For any point to the right of $t'+c$,
a valid $\argtwo$ witness should necessarily be to the right of $t + (c - b)$.
This observation leads to the following construction: when the first $\argtwo$ is seen (the red
dot), as long as there is a continuous sequence of $\argone$'s following it, break the timeline
into blocks of $c - b$ and store the earliest and latest $\argtwo$ within these blocks.  When
the values become sufficiently large, the earliest block becomes irrelevant, allowing to
reuse clocks.  Coupled with $c-b\geq 1$, we are able to implement this idea with a
bounded number of clocks.  

Let $I$ be a non-singleton interval with lower and upper bounds $b,c\in\mathbb{N}$ and 
$0\notin I$. Since $I$ is non-singleton, we have $c-b\geq 1$.
Let $k=2+\Big\lfloor\dfrac{b}{c-b}\Big\rfloor$.\footnote{If the interval $I$ is not of the
  form $(b,c)$ (i.e., open on both sides) then we can construct an automaton $\ASgenI$ with
  $k=1+\Big\lceil\dfrac{b}{c-b}\Big\rceil$ by changing $x_j<c-b$ on Line 10 to $x_j\leq
  c-b$ and $c-b\leq x_j$ on Line 11 to $c-b<x_j$.}
The automaton $\ASgenI$ for $\argone\S_{I}\argtwo$ has a single location $\ell_{0}$ and
uses clocks $x_{1},y_{1},\ldots,x_{k},y_{k}$.
The initial condition at $\ell_{0}$ is $x_1=y_1=\cdots=x_k=y_k=+\infty$ (all clocks are 
initially inactive). 
Automaton $\ASgenI$ is synchronized with $\AS$ (untimed since) from 
\Cref{fig:Since-1-sided-automata}~(left).
Its transitions are defined in \Cref{algo:Since-I-general}
with ${\leqlt}={\leq}$ if $b\in I$ and ${\leqlt}={<}$ otherwise.

\begin{algorithm}[htb]
  \small
  \caption{Transitions of automaton $\ASgenI$.}
  \label{algo:Since-I-general}
  \begin{algorithmic}[1]

    \If{$\post\AS\state=\ell_{0}$}
      \Comment $\argone\S\argtwo$ does not hold
      \State $x_1,y_1,\ldots,x_k,y_k:=+\infty$
      
    \ElsIf{$\neg\postM\argone \vee \pre\AS\state=\ell_{0}$}
      \Comment $\argone\S\argtwo$ holds ($\post\AS\state=\ell_{1}$), new earliest witness
      \State $x_1,y_1:=0$; $x_2,y_2,\ldots,x_k,y_k:=+\infty$
      
    \Else
      \Comment $\argone\S\argtwo$ holds ($\post\AS\state=\ell_{1}$) and there are past witnesses ($\postM\argone \wedge \pre\AS\state=\ell_{1}$)
      
      \If{$b\leqlt x_n\neq+\infty \wedge (n=k \vee \neg(b\leqlt x_{n+1} \neq+\infty))$}
      \Comment Clock shift with $1<n\leq k$
      \State $[x_1,y_1,\ldots,x_k,y_k]:=[x_n,y_n,\ldots,x_k,y_k,+\infty,\ldots,+\infty]$
      \EndIf
    
      \If{$\postM\argtwo$}
        \Comment new witness (not earliest)
        \IfThen{$x_j<c-b$}{$y_j:=0$}
        \Comment $1\leq j\leq k$
        \EndIfThen
      
        \IfThen{$c-b\leq x_j\neq+\infty \wedge x_{j+1}=+\infty$}{$x_{j+1},y_{j+1}:=0$}
        \Comment $1\leq j<k$
        \EndIfThen
      \Else
        \Comment not a new witness ($\neg\postM\argtwo$)
        \State\textsf{Skip}
      \EndIf
    \EndIf
  \end{algorithmic}
\end{algorithm}

Notice that the automaton is deterministic and complete.  Consequently, every word will
have exactly one run in the automaton.  We give some intuitions before the formal
statement in \Cref{lem:since-general}.
\begin{itemize}[left=0.1em]
  \item The test of Line 1
  is satisfied when the current position does not satisfy
  $\argone\S\argtwo$.  In this case, the new clocks are not needed and we set them to
  $+\infty$, meaning inactive.
  
  \item  When the test of Line 3 succeeds, the current position satisfies 
  $\argone\S\argtwo$ and there are no other witness in the past. So we set $x_1,y_1$ to 
  $0$ and the other clocks are inactive.
  
  \item If we have $b\leqlt x_n$ for some active clock with $n>1$ then $x_n$ satisfies the
  lower bound constraint of $I$ and it may serve as a witness instead of the earlier
  clocks $x_1,y_1,\ldots,x_{n-1},y_{n-1}$ (notice that $x_n<y_{n-1}\leq
  x_{n-1}<\cdots<y_1\leq x_1$ and if one of these clocks satisfies the upper bound
  constraint of $I$ then so does $x_n$).  In this case, clocks
  $x_1,y_1,\ldots,x_{n-1},y_{n-1}$ are not needed anymore and we reduce the set of active
  clocks with the clock shift in Lines (6--8).
  
  \item  When the test of Line 9 succeeds,
  the existing witnesses stay and the current position is a new witness.  Let
  $x_1,y_1,\ldots,x_j,y_j$ be the active clocks.
  
  If $x_j<c-b$ (Line 10) then we simply reset $y_j$ to $0$ since the 
  current position is now the last witness in the block of length $c-b$ which started 
  when $x_j$ was reset. Otherwise $x_j\geq c-b$ (Line 11) and the current position starts 
  a new block by setting $x_{j+1},y_{j+1}$ to $0$ (new pair of active clocks). 
  When this occurs, we have $j<k$ thanks to the clock shift of Lines (6--8).

  \item The last possibility (Lines 12,13), the current position satisfies
  $\argone\S\argtwo$ but is not a new witness since $\argtwo$ doest not hold.  In that
  case, we do not update the clocks since the set of witnesses does not change.
\end{itemize}

\begin{restatable}
  {lemma}{lemsincegeneral}
    \label{lem:since-general}
  Let $I$ be a non-singleton interval with lower and upper bounds $b,c\in\mathbb{N}$ and 
  $0\notin I$.  Let ${\leqlt}={\leq}$ if $b\in I$ and ${\leqlt}={<}$ otherwise.
  Let $\Prop' = \Prop \uplus \{\argone, \argtwo\}$.  For every word $w =
  (a_0, \tau_0) (a_1, \tau_1) \cdots$ over $\Prop'$ there is a unique run $C_0
  \xra{\delta_0} C'_0 \xra{\overline{t}_0} C_1 \xra{\delta_1} C'_1 \xra{\overline{t}_1}
  C_2 \cdots$ of $(\AS,\ASgenI)$.
  
  Moreover, let $\outv_{\S_{I}}=(\post{\AS}\state=\ell_1 \wedge (\post\ASgenI{x_{1}}\in I
  \vee \post\ASgenI{y_{1}}\in I))$.  Then, for all $i \ge 0$, we may check if
  $(w,i)\models\argone\S_{I}\argtwo$ with the test $\outv_{\S_{I}}$ during transition
  $\overline{t}_{i}$ from $C'_{i}$ to $C_{i+1}$.
\end{restatable}

\subsection{A full example}\label{sec:full-example}

We explain the construction of a network $\cN_{\Phi}$ associated with a sentence $\Phi$ 
on an example. We use the template automata defined in the previous sections. We 
instantiate the formal arguments $\argone,\argtwo$ of these template automata with actual 
test formulas. For instance, $\AY[\postM\argone/(\neg\postM{p}\wedge\postM{q})]$ denotes 
a copy of the automaton $\AY$ from \Cref{fig:Yesterday-automaton}~(left) in which we 
substitute the test $\neg\postM{p}\wedge\postM{q}$ for the formal argument $\postM\argone$.

  We consider the sentence $\Phi=(\Phi_{1}\wedge\Phi_{2})\rightarrow\Phi_{3}$ where 
  $\Phi_{1}=r$ and
  \begin{align*}
    \Phi_{2} & = \F_{I_{6}}(p\wedge\varphi_{4}) = \true\U_{I_{6}}(p\wedge\varphi_{4}) 
    &
    \Phi_{3} & = \F_{I_{7}}(\varphi_{3}\wedge\varphi_{5}) = 
    \true\U_{I_{7}}(\varphi_{3}\wedge\varphi_{5}) 
    \\
    \varphi_{4} &= (p\vee q)\S_{I_{4}} \varphi_{3} 
    &
    \varphi_{5} &= (\varphi_{1}\vee\varphi_{2}) \S_{I_{5}} (q\vee\varphi_{3})
    \\
    \varphi_{1} &= \Y_{I_{1}} (\neg p\wedge q) 
    \qquad
    \varphi_{2} = \Y_{I_{2}} (\neg p\wedge q) 
    &
    \varphi_{3} &= \Y_{I_{3}} r 
  \end{align*}
  with $I_{4}=[0,c_{4})$ an upper bound constraint
  and $I_{5}=[c_{5},+\infty)$ a lower bound constraint.
  
  The network $\cN_{\Phi}$ does not contain automata for the boolean connectives.
  Instead, we use the power of tests in automata, which allow arbitrary boolean
  combinations of atomic tests. 
  A \emph{temporal} subformula is either an atomic proposition or a formula where the
  outermost connective is a temporal modality (Yesterday, Since, neXt or Until).  The
  temporal subformulas of the sentence $\Phi$ are the atomic propositions $p,q,r$, the
  point formulas $\varphi_{1},\ldots,\varphi_{5}$ and the sentences $\Phi_{2},\Phi_{3}$.
  
  The network $\cN_{\Phi}$ does not need automata for the atomic propositions since they
  are owned by the model $\cM$.  For each temporal subformula, we add to the network the
  automata required to determine the truth value of the subformula.  These automata are
  described in the previous sections.  For instance, the temporal subformula $\varphi_{4}$
  requires a copy of the automata $\AS$ and $\ASlast$ from
  \Cref{fig:Since-1-sided-automata} (with suitable substitutions of the formal arguments).
  \Cref{lem:since-upper-bound} provides the test expression which allows to
  determine the truth value of $\varphi_{4}$.
  
  For the sentence $\Phi=(\Phi_{1}\wedge\Phi_{2})\rightarrow\Phi_{3}$, we construct the
  network
  $$
  \cN_{\Phi}=(\cB_{0},\cB_{1},\cB_{2},\cB_{3},\cB_{4},\cB_{5},
  \cB_{6},\cB_{7},\cA_{1},\cA_{2},\cA_{3})
  $$
  where
  \begin{align*}
    \cB_{1} &:= \Alast &
    \cB_{2} &:= \AY[\postM\argone/(\neg\postM{p}\wedge\postM{q})]
    \\
    \outv_{1} &= (\pre{\cB_{2}}\state=\ell_{1} \wedge \pre{\cB_{1}}\xlast\in I_{1}) &
    \outv_{2} &= (\pre{\cB_{2}}\state=\ell_{1} \wedge \pre{\cB_{1}}\xlast\in I_{2})
  \end{align*}
  \begin{align*}
    \cB_{3} &:= \AY[\postM\argone/\postM{r}] &
    \outv_{3} &= (\pre{\cB_{3}}\state=\ell_{1} \wedge \pre{\cB_{1}}\xlast\in I_{3})
    \\
    \cB_{4} &:= \AS[\postM\argone/(\postM{p}\vee\postM{q}),\postM\argtwo/\outv_{3}] &
    \cB_{5} &:= \ASlast[\postM\argtwo/\outv_{3}] 
    \\
    \cB_{6} &:= \AS[\postM\argone/(\outv_{1}\vee\outv_{2}),
    \postM\argtwo/(\postM{q}\vee\outv_{3})] &
    \cB_{7} &:= \ASfirst[\postM\argone/(\outv_{1}\vee\outv_{2}),
    \postM\argtwo/(\postM{q}\vee\outv_{3})]
    \\
    \outv_{4} &= (\post{\cB_{4}}\state=\ell_{1} \wedge \post{\cB_{5}}{x}\in I_{4}) &
    \outv_{5} &= (\post{\cB_{6}}\state=\ell_{1} \wedge \post{\cB_{7}}{y}\in I_{5})
    \\
    \cB_{0} &:= \Ainit &
    \cA_{1} &:= \Aisatp{r}
    \\
    \cA_{2} &:= \cA^{isat}_{\U_{I_{6}}}[\postM\argone/\true,
    \postM\argtwo/(\postM{p}\wedge\outv_{4})] &
    \cA_{3} &:= \cA^{isat}_{\U_{I_{7}}}[\postM\argone/\true,
    \postM\argtwo/(\outv_{3}\wedge\outv_{5})] 
  \end{align*}
  In the automata above, we use ghost variables $(\outv_{i})_{1\leq i\leq 5}$ to make the 
  definitions easier to read.  When we expand these definitions, we get for instance
  \begin{align*}
    \cB_{4} &= \AS[\postM\argone/(\postM{p}\vee\postM{q}),
    \postM\argtwo/(\pre{\cB_{3}}\state=\ell_{1} \wedge \pre{\cB_{1}}\xlast\in I_{3})] \\
    \cA_{2} &:= \cA^{isat}_{\U_{I_{6}}}[\postM\argone/\true,
    \postM\argtwo/(\postM{p}\wedge(\post{\cB_{4}}\state=\ell_{1} \wedge \post{\cB_{5}}{x}\in I_{4}))] 
  \end{align*}
  Notice that the automaton $\cB_{2}$ is shared: it is used in 
  $\outv_{1},\outv_{2}$. Similarly, $\cB_0$ and $\cB_1$ are shared.
  
  It remains to define the set $\bad_{\Phi}$ of \emph{bad} configurations for the 
  sentence $\Phi$. A configuration $C$ is bad if $\Phi$ evaluates to false at $C$. Hence, 
  $\bad_{\Phi}$ is the set of configurations $C$ which satisfy
  $\pre{\cA_{1}}\state=2 \wedge \pre{\cA_{2}}\state=2 \wedge \neg(\pre{\cA_{3}}\state=2)$ 
  (i.e., $\Phi_{1}\wedge\Phi_{2}\wedge\neg\Phi_{3}$).

\subsection{Proof of \Cref{thm:fast-mtl-theorem}}\label{sec:proof-of-fast-mtl}

We give now the general construction for an arbitrary sentence $\Phi$.
Let $\varphi_{1},\ldots,\varphi_{n}$ be the temporal point subformulas of $\Phi$ other 
than the atomic propositions. We assume that if $\varphi_{j}$ is a subformula of 
$\varphi_{i}$ then $j\leq i$. Let $\Phi_{1},\ldots,\Phi_{m}$ be the temporal subsentences 
of $\Phi$. We define the network
$$
\cN_{\Phi}=(\cB_{0},\cB_{1},\cB_{2},\ldots,\cB_{2n-1},\cB_{2n},\cA_{1},\ldots,\cA_{m})
$$
together with the ghost variables $\outv_{1},\ldots,\outv_{n}$.
For each $1\leq i\leq n$, we define $\cB_{2i-1},\cB_{2i}$ and $\outv_{i}$ as follows:
\begin{itemize}[left=0.1em]
  \item  If $\varphi_{i}=\Y_{I_{i}}\varphi'_{i}$ then we let $\cB_{2i}=\Alast$. 
  The automaton $\cB_{2i-1}$ is a copy of the untimed $\AY$ with a suitable substitution 
  of its formal argument $\argone$ that we define now. The formula $\varphi'_{i}$ is a 
  boolean combination of temporal subformulas. Let $\varphi''_{i}$ be the same boolean 
  combination in which an \emph{outermost} temporal subformula $\varphi_{j}$ is replaced with 
  $\outv_{j}$ and an \emph{outermost} atomic proposition $p$ is replaced with $\postM{p}$.
  For instance, if $\varphi'_{i}=(q\wedge\varphi_{2})\vee\varphi_{4}$ then 
  $\varphi''_{i}=(\postM{q}\wedge\outv_{2})\vee\outv_{4}$. Then, 
  $\cB_{2i-1}=\AY[\postM\argone/\varphi''_{i}]$. Finally,
  $\outv_{i}=(\pre{\cB_{2i-1}}\state=\ell_{1} \wedge \pre{\cB_{2i}}\xlast\in I_{i})$.

  \item  If $\varphi_{i}=\varphi'_{i}\S_{I_{i}}\psi'_{i}$ and $I$ is an upper-bound 
  interval. Then we let
  $\cB_{2i-1}=\AS[\postM\argone/\varphi''_{i},\postM\argtwo/\psi''_{i}]$,
  $\cB_{2i}=\ASlast[\postM\argtwo/\psi''_{i}]$ and
  $\outv_{i}=(\post{\cB_{2i-1}}\state=\ell_{1} \wedge \post{\cB_{2i}}{x}\in I_{i})$.

  \item If $\varphi_{i}=\varphi'_{i}\S_{I_{i}}\psi'_{i}$ and $I$ is a lower-bound
  interval.  Then we let
  $\cB_{2i-1}=\AS[\postM\argone/\varphi''_{i},\postM\argtwo/\psi''_{i}]$,
  $\cB_{2i}=\ASfirst[\postM\argone/\varphi''_{i},\postM\argtwo/\psi''_{i}]$ and
  $\outv_{i}=(\post{\cB_{2i-1}}\state=\ell_{1} \wedge \post{\cB_{2i}}{y}\in I_{i})$.

  \item If $\varphi_{i}=\varphi'_{i}\S_{I_{i}}\psi'_{i}$ and $I$ is a non-singleton 
  interval which is neither an upper-bound nor a lower-bound.  Then we let
  $\cB_{2i-1}=\AS[\postM\argone/\varphi''_{i},\postM\argtwo/\psi''_{i}]$,
  $\cB_{2i}=\cA_{\S_{I_i}}^{\mathsf{gen}}[\postM\argone/\varphi''_{i}, 
  \postM\argtwo/\psi''_{i}, \pre\AS\state/\pre{\cB_{2i-1}}\state]$ and
  $\outv_{i}=(\post{\cB_{2i-1}}\state=\ell_{1} \wedge (\post{\cB_{2i}}{x_{1}}\in I_{i} 
  \vee \post{\cB_{2i}}{y_{1}}\in I_{i}))$.
\end{itemize}
Next, we let $\cB_{0}=\Ainit$ and for $1\leq i\leq m$ we define
\begin{itemize}
  \item  If $\Phi_{i}=p$ is an atomic proposition, then we let $\cA_{i}=\Aisatp{p}$.

  \item  If $\Phi_{i}=\X_{J_{i}}\varphi'_{i}$ then we let 
  $\cA_{i}=\cA^{isat}_{\X_{J_{i}}}[\postM\argone/\varphi''_{i}]$.

  \item  If $\Phi_{i}=\varphi'_{i}\U_{J_{i}}\psi'_{i}$ then we let 
  $\cA_{i}=\cA^{isat}_{\U_{J_{i}}}[\postM\argone/\varphi''_{i},\postM\argtwo/\psi''_{i}]$.
\end{itemize}
Finally, as explained in \Cref{sec:fastmtl-to-ta}, the set $\bad_\Phi$ of bad 
configurations consists of all configurations $C$ such that $\Phi$ evaluates to false 
when we replace each temporal subsentence $\Phi_{i}$ of $\Phi$ by 
$\pre{\cA_{i}}\state=\ell_{i}$ where $\ell_{i}$ is the state labelled $\true$ in 
$\cA_{i}$ ($\ell_{i}=3$ if $\Phi_{i}$ is a next subformula, and $\ell_{i}=2$ otherwise).

Now that we have constructed the \emph{deterministic} network $\cN_{\Phi}$, we prove its
correctness.

\begin{lemma}\label{lem:N-Phi-correctness}
  Let $\Phi$ be a sentence in $\detMTL$.
  The network $\cN_{\Phi}$ constructed above is deterministic and complete.
  Hence, for every word $w = (a_0, \tau_0) (a_1, \tau_1) \cdots$ over $\Prop$ there is a
  unique run $C_0 \xra{\delta_0} C'_0 \xra{\overline{t}_0} C_1 \xra{\delta_1} C'_1
  \xra{\overline{t}_1} C_2 \cdots$ of $\cN_\Phi$. 
  Moreover,
  \begin{enumerate}
    \item for all temporal point subformula $\varphi_{j}$ of $\Phi$ (with $1\leq j\leq n$)
    and for all position $i\geq0$ in $w$, we have $(w,i)\models\varphi_{j}$ iff evaluating
    the test $\outv_{j}$ during transition $\overline{t}_{i}$ from $C'_{i}$ to $C_{i+1}$
    yields $\true$.
    
    \item  $w\not\models\Phi$ iff there is $k\geq0$ such that $C_{i}\in\bad_{\Phi}$ for 
    all $i\geq k$.
  \end{enumerate}
\end{lemma}

\begin{proof}
  All automata used in the network $\cN_{\Phi}$ are deterministic and complete. Hence, so 
  is the network and each timed word has indeed a unique run.
  \begin{enumerate}[left=0.1em]
    \item  The proof is by induction on $j$. It follows easily from the correctness 
    lemmas in the previous sections: 
    \Cref{lem:Yesterday} for yesterday with an arbitrary interval, 
    \Cref{lem:since-upper-bound} for since with an upper-bounded interval, 
    \Cref{lem:since-lower-bound} for since with a lower-bounded interval and
    \Cref{lem:since-general} for since with a non-singleton interval which is 
    neither upper-bounded nor lower-bounded.
  
    \item  Using (1) and the lemmas for initial satisfiability (\Cref{lem:isat-p}, 
    \Cref{lem:isat-X} and \Cref{lem:isat-U}), we deduce that for each temporal 
    subsentences $\Phi_{j}$ of $\Phi$ (with $1\leq j\leq m$), we have $w\models\Phi_{j}$ 
    iff the automaton $\cA_{j}$ eventually stabilises in a state labelled $\true$, i.e., 
    state $2$ for an atomic proposition or for an until modality, and state $3$ for next.
    Notice that each automaton $\cA_{j}$ must eventually stabilise in some state since 
    each cycle in these automata are actually self-loops. Now, using the definition of 
    $\bad_{\Phi}$, we deduce that $w\not\models\Phi$ iff the run for $w$ eventually stays
    in bad configurations.
    \qedhere
  \end{enumerate}
\end{proof}

\begin{remark}
  The number of automata in the network $\cN_{\Phi}$ is $\mathcal{O}(|\Phi|)$.  The
  automata $\cB_{2i}$ have a single state and own clocks, the automata $\cB_{2i-1}$ have
  two states and do not own clocks, and the automata $\cA_{i}$ have 3 or 4 states and do
  not own clocks.  The number of clocks owned by $\cB_{2i}$ is 1 unless it comes from some
  $\S_{I}$ where $I$ is a two-sided interval with lower and upper bounds
  $b,c\in\mathbb{N}$ such that $c-b\geq1$.  In that case, the number of clocks is $2k$
  where $k=2+\Big\lfloor\dfrac{b}{c-b}\Big\rfloor\leq2+b$.  Hence, the total number of
  clocks in the network is $\mathcal{O}(\Phi)$ if $\Phi$ does not use since with two-sided
  intervals, or if $k=\mathcal{O}(1)$ for all since with two-sided intervals, or if the
  constants are given in unary (actually, only the lower constants of two-sided
  intervals).
\end{remark}

\begin{remark}
  We reduce further the size of the network $\cN_{\Phi}$ by removing duplicates.  For
  instance, if $\Phi$ has several occurrences of $\Y$, possibly with different arguments
  and intervals, the automaton $\Alast$ occurs as many times in the network constructed
  above.  In practice, we use a single occurrence of $\Alast$ adapting accordingly the
  atomic tests $\xlast\in I_{i}$ in the corresponding ghost variables $\outv_{i}$. 
  We also remove duplicate occurrences of $\AY$ automata and of since 
  automata ($\AS$, $\ASlast$, $\ASfirst$).
\end{remark}

\section{Integrating the Future}\label{sec:future}

Our goal in this section is to include future modalities in the logic and obtain an
efficient model checking procedure for a general fragment of $\MTL$ with past and future.
To do so, we consider the logic $\detMTL$ from \Cref{def:fast-mtl} and allow Point
formulas to also have $\Next_I$ and $\U_{I'}$ operators.  This results in a restriction of
$\MTL$ in which, except in the topmost level, the intervals in Since and Until are
restricted to be $[0,0]$ or non-singleton.  We call this fragment $\MTLfp$.  This fragment
subsumes the well-known $\MITL$ fragment and even goes beyond extending it with past.

We recall that in the pointwise semantics, $\MTL$ model checking is undecidable over
infinite words, and over finite words it is decidable but the complexity is
Ackermanian~\cite{Survey-AlurH91,MTL-OW07J,MTL-OW07C,OuaknineW06}.  Hence, the literature
often considers the $\MITL$ fragment in the pointwise
semantics, with multiple known approaches for building efficient model checking
algorithms.  Of these we highlight two recent ones, both of which only work with future
modalities.  The first from~\cite{MightyL}, works by translating $\MITL$ into a network of
timed automata (via 1-clock alternating automata) and was implemented in a
state-of-the-art tool \MightyL. The second from~\cite{AGGS-CONCUR24}, converts $\MITL$
into generalized timed transducers, which use future clocks defined in~\cite{AGGJS-CAV23}.
However, this approach has not yet been implemented.

In the rest of this section, we combine ideas from~\cite{AGGS-CONCUR24}
and from~\cite{MightyL}, with our ideas for initial satisfiability and
past modalities presented in the previous sections, to develop a new construction from
formulas in $\MTLfp$ directly to network of generalized timed automata with shared
variables, without the explicit use of transducers.  As we are dealing with future
modalities, non-determinism cannot be avoided, but we show two algorithmic improvements,
both designed to {\em reduce non-deterministic branching}.  Our experimental results later
showcase improved performance (compared to either of the two approaches above) in
number/type of benchmarks solved and time taken, thus highlighting the significant impact
of these improvements.

\subsection{Including Future clocks in the Model}

First, we lift the definitions to include future clocks.  Thus the set of clocks $X$ is
now partitioned into history clocks $\Xh$, which are just clocks as we saw till now and
$\Xf$, future clocks, that take non-positive values.  A future clock takes values in the
interval $\overline{\mathbb{R}}_{-}=(-\infty,0]\cup\{-\infty\}$, where value $-\infty$
signifies that this clock is undefined or inactive.  Thus, a {\em generalized timed
automaton with shared variables (GTA)} is a tuple $\cA=(Q,X,V,\Delta,\init)$, where
$X=\Xh\overset{\cdot}{\cup} \Xf$ is the set of clocks owned by $\cA$.

Now, an atomic update for a future clock variable $y\in \Xf$ owned by $\cA$ is
non-deterministic (unlike to normal/history clocks) and has the form $\now\cA{y}:\in I$ 
(or simply $y:\in I$), for some interval $I$ with end-points in 
$\overline{\mathbb{Z}}_{-}=\{0,-1,-2,\ldots,-\infty\}$.  Further, an
atomic test for $y\in\Xf$ owned by $\cA$ has the form $\now\cA{y}\in I$ (or simply $y\in
I$) where $I\subseteq\overline{\mathbb{R}}_{-}$ is an integer bounded interval.
For simplicity, in examples we often write $-y\in I$ with endpoints in
$\overline{\mathbb{N}}$ instead.
As before $\cA$ may also test a future clock variable owned
by another component with $\pre\cN{y}\in I$ or $\post\cN{y}\in I$ where $\cN$ is a name.
As before a {\em
network of GTA with shared variables} is just a tuple
$\overline{\cA}=(\cA_{1},\ldots,\cA_n)$ where each $\cA_i$ is a GTA with shared variables.

The semantics of a network of GTA is the same as for network of TA as before with two key
differences: First, in the definition of a \emph{configuration} $C$ for any $1\leq i\leq
n$, we have in addition $C(\now{\cA_i}y)\in\overline{\mathbb{R}}_-$ for $y\in \Xf_{i}$.
Second, a \emph{time elapse} transition $C\xra{\delta}C'$ can be taken only if (i) $C'$
coincides with $C$ on all variables other than clocks, (ii) the value of each clock
(regardless of whether it is future or history) is advanced by the same quantity $\delta$
and (iii) $C(\now{\cA_i}y)+\delta\leq 0$ for all $1\leq i\leq n$ and $y\in \Xf_i$ (with 
$-\infty+\delta=-\infty$).  The
third condition is new and means that an elapse can occur only if it does not force a
future clock to take a positive value (same requirement as
in~\cite{AGGJS-CAV23,AGGS-CONCUR24} for future clocks).

Finally, we recall that our TA (and GTA) are always strongly non-Zeno, i.e., every
accepting run is non-Zeno (a timed run is Zeno if {\small $\sum_{i\geq 0}\delta_i$} is
bounded).  As remarked in Lemma~5 of \cite{AGGS-CONCUR24}, this assumption (not a severe
restriction, since every GTA can be converted easily to a strongly non-Zeno one) has
interesting consequence for future clocks: if they are not ultimately  $-\infty$,
they should be released infinitely often.  If not, there is a last point where a future
clock is released to a finite value,
and the entire suffix of the run lies within this finite time, contradicting non-Zenoness.

\subsection{Model checking the future}

\begin{figure}[tbp]
  \centering
  \includegraphics[page=15,scale=1]{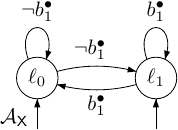}
  \hfil
  \raisebox{6mm}{\includegraphics[page=14,scale=1]{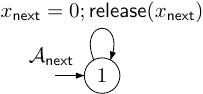}}
  \caption{\small
  Left: Automaton $\AX$ for the untimed Next operator $\X\argone$.
  The boolean variable $\argone$ is not owned by $\AX$ and stands for the argument of $\X$.  
  \\
  Right: Automaton $\Anext$ with clock $\xnext$ used (shared) by all Next operators.  
  }
  \label{fig:Next-automaton}
\end{figure}

For the model checking procedure, the idea at the high level is as before.  Given a
formula in $\MTLfp$, we translate it to a network of GTA with shared variables and then
build a closed network with a model and check existence of a {\em bad cycle} of
configurations.  As explained in~\cite{AGGS-CONCUR24}, the definition of a bad cycle now
also has to take into account future clocks.  The goal of this section is to prove the
following theorem.

\begin{theorem}\label{thm:full-mtl-theorem}
  Let $\Phi$ be a $\MTLfp$ sentence and $\cM$ a timed model with history and future
  clocks.  Then, we can construct a network of GTA with shared variables $\cN_\Phi$ of
  size linear in size of $\Phi$ (assuming unary encoding of constants in the
  formula)\footnote{In fact, we can also do the same with binary encoding of constants
    with one caveat: the lower bounds on intervals used in
    Since and Until operators should be bounded by a constant.} 
  such that $\cM \not \models \Phi$ iff there exists a reachable cycle of the closed
  network $(\cM, \cN_\Phi)$, which eventually remains in bad configurations and where
  every future clock $x$ of the network is either released during the cycle or is
  undefined $-\infty$.
\end{theorem}

\subsubsection{Next with sharing.}
As for Yesterday, for any occurrence of the Next operator $\X_{I}$, we have an automaton
depicted in \Cref{fig:Next-automaton} (left), which has two states $\ell_{0}$ and
$\ell_{1}$ and captures the {\em untimed semantics} of Next. 
The states guess
whether the argument $\argone$ of $\X$ will be true or not at the next letter without
worrying about the timing constraints.
Then, we use a single state automaton $\Anext$, depicted in
\Cref{fig:Next-automaton} (right), which releases clock $\xnext$ whenever it reaches $0$.
Thus, it always predicts the time to the next event. It will be shared across all
occurrences of Next in the entire network, to check all the timing constraints.
Now, for an arbitrary interval constraint $I$, the current position satisfies
$\X_{I}\argone$ if the next letter has a $\argone$ and the value of clock $\xnext$
(released here) falls in the interval $I$.  Formally,
\begin{lemma}
  Let $\Prop' = \Prop \uplus \{\argone\}$.  For every word $w = (a_0, \tau_0) (a_1,
  \tau_1) \cdots$ over $\Prop'$, we have a unique run 
  $C_0 \xra{\delta_0} C'_0 \xra{\overline{t}_0} C_1 \xra{\delta_1} C'_1
  \xra{\overline{t}_1} C_2 \cdots$ of $(\AX,\Anext)$. Moreover, for all $i\ge 0$, 
  \begin{enumerate}
    \item $(w,i) \models \X\argone$ iff
    $C_{i+1}(\AX.\state)=\ell_1$,
    
    \item $\tau_{i+1}-\tau_{i}\in I$ iff $C_{i+1}(\Anext.\xnext)\in I$.
  \end{enumerate}
  Let $\outv_{\X_{I}}=(\post{\AX}\state=\ell_1 \wedge -\post\Anext\xnext\in I)$.
  Then, for all $i\geq0$, we may check if $(w,i)\models\X_{I}\argone$ with the test
  $\outv_{\X_{I}}$ during transition $\overline{t}_{i}$ from $C'_{i}$ to $C_{i+1}$.
\end{lemma}

\subsubsection{Until with sharing (one-sided).}

The formula $\argone \U_{I} \argtwo$ asks that there exists a position $j$ in the future
(there could be many of them), such that $\argtwo$ is true at that position, and $\argone$
is true at all positions between the current position and $j$, and the time elapsed
between the current position and position $j$ lies in the interval $I$.
As for Since, we first consider the simpler case where intervals are one-sided, i.e., of
the form $[0,c)$, $[0,c]$, $[b,+\infty)$ or $(b,+\infty)$.
Again it is sufficient to track the two \emph{extreme} $\argtwo$ witnesses: the earliest
witness and the last possible witness at each point.  

Let us formalize the notion of a
witness for $\U$ formulas.  Assume the current position to be $i$.  A position
$j \ge i$ is a witness for $\argone \U \argtwo$ at position $i$ if $j$ satisfies $\argtwo$
and for all positions $i \le k < j$, position $k$ satisfies $\argone$.  The smallest $j$
among all witnesses is called the earliest witness and the largest $j$ (if it exists) is
said to be the latest witness.  If $\argone\U\argtwo$ holds but there is no latest
witness, it means that there are infinitely many witnesses for $\argone \U \argtwo$ at
position $i$.  Our goal now is to track the time to the earliest and latest witnesses.  We
make use of this observation: position $j$ is a latest witness for $\argone \U \argtwo$
at position $i$ if
\begin{itemize}[nosep]
  \item firstly, it is a witness for $\argone \U \argtwo$, and

  \item secondly, it is not the case that $j$ satisfies $\argone$ and $j+1$ satisfies
  $\argone \U \argtwo$: in other words, at position $j$, $\neg (\argone \wedge \X (\argone
  \U \argtwo) )$ holds.  Otherwise, notice that $j$ is not the latest witness.
\end{itemize}

We now introduce three automata as shown in \Cref{fig:Until-1-sided-automata} with two
future clocks $x$ and $y$.  (i) $\AUfirst$ which owns and uses future clock $x$ to track
the time to the first $\argtwo$, (ii) $\AU$ which captures the semantics of the untimed
Until operator, using clock $x$ in the invariant of state $\ell_{1}$, and (iii) $\AUlast$
which owns and uses future clock $y$ and tracks the time to a point which satisfies
$\argtwo \wedge \neg (\argone \wedge \X (\argone \U \argtwo) )$.  The next lemma
formalizes how the three automata put together help deduce time to the earliest and latest
witnesses for $\argone \U \argtwo$, using which we can design a suitable output
function.

\begin{figure}
  \centering
  \raisebox{0.5mm}{\includegraphics[page=17,scale=1]{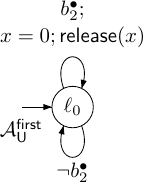}}
  \hfil
  \raisebox{4.5mm}{\includegraphics[page=16,scale=1]{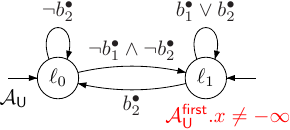}}
  \hfil
  \includegraphics[page=18,scale=1]{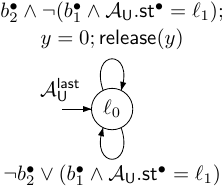}
  \caption{Left: $\AUfirst$ guessing time to the earliest $\argtwo$ with future clock $x$. 
  $\release{x}$ stands for the non-deterministic update $x:\in\overline{\mathbb{R}}_{-}$, 
  allowing $x$ to be set to any non-positive value, including $-\infty$.
  \\ 
  Middle: $\AU$ for untimed until but which uses future clock $x$ in the invariant of 
  state $\ell_{1}$.  \\
  Right: $\AUlast$ guessing with future clock $y$ time to the latest witness of
  $\argone\U\argtwo$.
  }
  \label{fig:Until-1-sided-automata}
\end{figure}

\begin{restatable}{lemma}{lemuntilsidedbound}
  \label{lem:until-1sided-bound}
  Let $\Prop' = \Prop \uplus \{\argone, \argtwo\}$.  For every word $w = (a_0, \tau_0)
  (a_1, \tau_1) \cdots$ over $\Prop'$ there is a unique run 
  $C_0 \xra{\delta_0} C'_0 \xra{\overline{t}_0} C_1 \xra{\delta_1} C'_1
  \xra{\overline{t}_1} C_2 \cdots$ of $(\AUfirst,\AU,\AUlast)$.
  Moreover, for all $i \ge 0$:
  \begin{enumerate}
    \item $(w,i) \models \argone\U\argtwo$ iff
    $C_{i}(\pre\AU\state)=C'_{i}(\pre\AU\state)=\ell_1$,
    
    \item If $I$ is an upper-bound interval, we set $\outv_{\U_{I}}=(\pre\AU\state=\ell_1
    \wedge \Absolut{\pre\AUfirst{x}}\in I)$.  Then, for all $i\geq0$, we may check if
    $(w,i)\models\argone\U_{I}\argtwo$ with the test $\outv_{\U_{I}}$ during
    transition $\overline{t}_{i}$ from $C'_{i}$ to $C_{i+1}$.

    \item If $I$ is a lower-bound interval, we set
    $\outv_{\U_{I}}=(\pre\AU\state=\ell_1 \wedge
    \Absolut{\pre\AUlast{y}}\in I\cup\{+\infty\})$.  
    Then, for all $i\geq0$, we may check if $(w,i)\models\argone\U_{I}\argtwo$ with
    the test $\outv_{\U_{I}}$ during transition $\overline{t}_{i}$ from $C'_{i}$ to
    $C_{i+1}$.
  \end{enumerate}
\end{restatable}

As we are dealing with future, all three automata are non-deterministic.  However, the
only non-determinism in $\AUfirst$ and $\AUlast$ lies in the non-deterministic updates
(release).  The zone automata\footnote{Zone automata are finite abstractions of
(generalized) timed automata, used for efficient algorithms checking
reachability~\cite{Daws,AGGJS-CAV23}}
built from them are actually deterministic since the release operation results in a unique
zone.  Unlike in the previous section, we cannot hope to get a completely deterministic
network since with future modalities we must have non-determinism. 
But we can try to limit the non-deterministic branching for better performance in practice.
Towards this, our {\em networks of GTA with shared variables} for each future modality
allow sharing of clocks and predictions (compared e.g., to using combined transducers as
done in~\cite{AGGS-CONCUR24}).  The advantage of using different automata is that we can
now share the prediction of the earliest witness without insisting to share the latest
witness at a point.  This results in increased sharing and thus further decrease in
non-deterministic branching.  Our experimental results in the next section validate these
theoretical insights.

\begin{remark}
  The most interesting automaton is $\AU$ in \Cref{fig:Until-1-sided-automata}~(middle),
  which uses clock $x$ from $\AUfirst$, but does not own any clock.  Notice that it is a
  two-state automaton and does not require any B\"uchi acceptance condition.  In fact, in
  the untimed setting (i.e., in the absence of clocks), Until requires some B\"uchi
  acceptance condition and 3 states if the acceptance condition is state-based.
  Lifting this to the timed setting gives a 3-state transducer with a B\"uchi acceptance 
  condition in~\cite{AGGS-CONCUR24}.  However, we show in this paper that, surprisingly, 
  2 states are sufficient and no acceptance condition is required, by exploiting clocks
  and non-zenoness.  The main idea is that in state $\ell_1$, we employ an
  invariant\footnote{Note that, though we did not
    introduce state invariant in the model, it can easily be simulated with guards on all
    incoming transitions --for succinctness we chose to present it this way.}
  implying that $\argtwo$ will eventually occur.
  Under the non-zenoness assumption, stating that clock $x$
  has a finite (not $-\infty$) value, forces clock $x$ to be eventually released.  Hence,
  $\AUfirst$ must eventually take the loop labeled $\argtwo$.
  We highlight that going from 3 states to 2 states represents an exponential improvement
  in non-deterministic branching and runtime as we go from $3^{\mathcal{O}(|\Phi|)}$ to
  $2^{\mathcal{O}(|\Phi|)}$ in the number of zones that we will have to explore during
  reachability checking.  Also, reducing the number of B\"uchi acceptance conditions makes
  the liveness algorithm more efficient.  Further, we also note that if we remove the
  invariant, and use of the clock, we obtain the automaton for untimed {\em Weak-Until}.
\end{remark}

\subsubsection{Until with sharing (two-sided non-singular).}

We consider now the most difficult case of the General Until operator with a two-sided
non-singular interval, i.e., of the form $[b,c)$, $[b,c]$, $(b,c)$, or $(b,c]$ with $b\neq
c$, $c\neq+\infty$ and $b\neq0$ when the interval is left-closed.

To handle this case, we define the general until automaton $\AUgenI$ which will be 
synchronized with the automata $\AUfirst$, $\AU$ and $\AUlast$.
The set of states of $\AUgenI$ is $\{0,1,\ldots,N\}\times\{1,2\}$ where
$N=2+\left\lfloor\dfrac{b}{c-b}\right\rfloor$.
The initial state is $(0,1)$.
All states are accepting. 
The transitions are defined in \Cref{algo:Until-I-general}.
State $(0,2)$ is not reachable.

The graph of the automaton $\AUgenI$ when $N=3$ is depicted in 
\Cref{fig:Until-I-gen-automaton}.

\begin{figure}[tb]
  \centering
  \includegraphics[page=23,scale=.9]{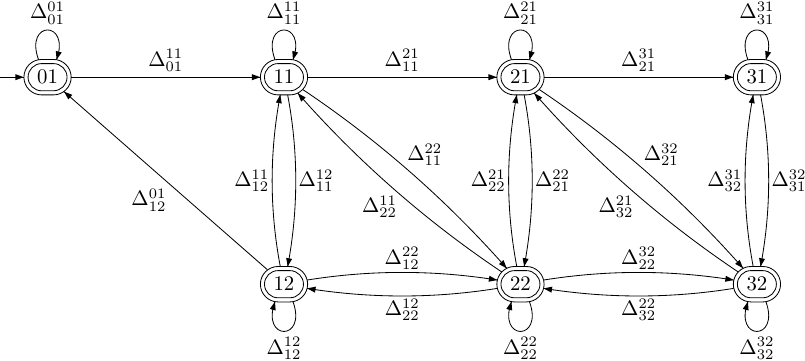}
  \caption{\small
  Automaton $\AUgenI$ for the general until operator $\argone\U_{I}\argtwo$ when $N=3$. \\
  The labels of the transitions are given in \Cref{tbl:labels-trans-1,tbl:labels-trans-2}.}
  \label{fig:Until-I-gen-automaton}
\end{figure}

Using the automata $(\AUfirst,\AU,\AUlast)$, there are a few easy cases (Lines 2,3,4 of
\Cref{algo:Until-I-general}) when $\AUgenI$ is in its initial state:
\begin{description}
  \item[Line 2]  Since $0\notin I$, if the current position satisfies 
  $\argone\U_{I}\argtwo$ then it must satisfy $\argone\wedge\X(\argone\U\argtwo)$. 
  When the latter is not true, we output $0$ and stay in state $(0,1)$.
  
  Note that, for all transitions in $\AUgenI$, except Lines 2,9,10,11, the current
  position satisfies $\argone\wedge\X(\argone\U\argtwo)$.

  \item[Line 3] If $\neg(\Absolut{\post\AUfirst{x}}<I)$ then\footnote{We use $z<I$ (resp.\
  $z>I$) to mean that $z$ is less than (resp.\ greater than) all values in $I$.  For
  instance, if $I=[b,c)$ then $z<I$ (resp.\ $z>I$) translates to $z<b$ (resp.\ $z\geq c$).}, 
  either $\Absolut{\post\AUfirst{x}}\in I$ and the earliest $\argtwo$ (in the strict
  future) is a witness of $\argone\U_{I}\argtwo$, or $\Absolut{\post\AUfirst{x}}>I$ and
  $\argone\U_{I}\argtwo$ does not hold at the current position (even the earliest
  $\argtwo$ is too far away to be a witness for $\argone\U_{I}\argtwo$).
  
  \item[Line 4] Similarly, if $\neg(\Absolut{\post\AUlast{y}}>I)$ then, either
  $\Absolut{\post\AUlast{y}}\in I$ and the last witness of $\argone\U\argtwo$ is a witness
  of $\argone\U_{I}\argtwo$, or $\Absolut{\post\AUlast{y}}<I$ and $\argone\U_{I}\argtwo$
  does not hold at the current position.
\end{description}

If neither of the above cases hold, then we need to guess a potential
witness within $I$ and verify it. This requires substantial
book-keeping which we will now explain. Assume we are given a timed
word $w = (a_0, \tau_0) (a_1, \tau_1) \cdots$.
Let us a call $j\geq0$ a \emph{difficult point} if 
$w,j\models\argone\wedge\X(\argone\U\argtwo)$ and
$\Absolut{\post\AUfirst{x}}<I$ (the earliest $\argtwo$ in the strict future is too close)
and $\Absolut{\post\AUlast{y}}>I$ (the last witness of $\argone\U\argtwo$, if any, is too
far away).

This leaves the possibility for a witness within $I$.  So, for difficult points,
we need to make a guess whether we have a witness within $I$, and check it.  We
cannot keep making such predictions for every difficult point, as we have only finitely
many clocks.  Therefore, we will guess some special witnesses.  First we state a useful
property.

\begin{lemma}
  Let $j$ be a difficult point.  Let $j_\mathrm{last}>j$ be the position of the last
  witness for $\argone \U \argtwo$ at position $j$ (we let $j_\mathrm{last}=+\infty$ if
  there is no last witness,
  i.e., if there are infinitely many witnesses).  Then, for all $k$ such that 
  $j\le k<j_\mathrm{last}$, we have $w,k\models\argone\wedge\X(\argone\U\argtwo)$.
\end{lemma}

Therefore, all positions between $i$ and $j_\mathrm{last}$ (including $i$ and
$j_\mathrm{last}$) satisfy $\argone \U \argtwo$, and the automaton $\AU$ stays in state
$\ell_1$, while reading $j$ up to $j_\mathrm{last}$.

\begin{figure}[tbh]
  \centering
  \begin{tikzpicture}
    \draw [thin, gray] (0,0) -- (8.5,0); \draw (0,0) to (0, 0.2);
    \node at (0, 0.4) {\small $\tau_j$};
    \draw (7,0) to (7, 0.2); 
    \node at (7, 0.4) {\small $\tau_j + c$}; 
    \draw (6,0) to (6, -0.2); 
    \fill[red] (6, 0) circle (1.5pt); 
    \node at (6, -0.4) {\small $\tau_{j'}$}; 
    \fill [blue] (8.4,0) circle (1.5pt); 
    \node at (8.4, -0.4) {\small $\tau_{j''}$};
    \draw (8.4, 0) to (8.4, -0.2);

    \draw [blue] (0, 0.1) to (3.4, 0.1); 
    \draw (3.4,0) to (3.4, 0.2);
    \node at (3.4, 0.4) {\small $\tau_{j''} - b$};
  \end{tikzpicture}
  \caption{Illustration of a difficult point $j$. 
  The point $j'$ is the last $\argtwo$-witness such that $\neg(\tau_{j'}-\tau_j>I)$. 
  The point $j''$ is the first $\argtwo$-witness such that $\tau_{j''}-\tau_j>I$.}
  \label{fig:Until-idea}
\end{figure}

We will now come back to the idea of choosing special witnesses.  This is illustrated in
\Cref{fig:Until-idea}.  For a difficult point $j$, we let $j'>j$ be the greatest position
containing $\argtwo$ such that $\neg(\tau_{j'}-\tau_j>I)$.  Let $j''>j$ be the least
position containing $\argtwo$ such that $\tau_{j''}-\tau_j>I$.  So $j'$ is the last point
either within $I$ or ``before entering'' $I$ that is a witness, and $j''$ is the first
point ``after'' the interval which is a witness. These positions exist since $j$ is a
difficult point.

\begin{description}
  \item[Lines 5--7] If the conditions of Lines 2,3,4 do not pass, it means the automaton
  is reading a difficult point.
  While reading a difficult point $j$, let us make use of fresh future
  clocks $x_1$ and $y_1$ to predict the time to the two witnesses $j'$ and $j''$:
  $x_1=\tau_j-\tau_{j'}$ and $y_1=\tau_j-\tau_{j''}$ (Line~5).
  We will have to check that positions $j',j''$ contain $\argtwo$ and that no position
  $j'<k<j''$ contains $\argtwo$.  To this end, we go to state $(1,1)$ (Line 7).
  
  Notice that, if $w,j\models\argone\U_{I}\argtwo$ then there is a position $k>j$ which 
  contains $\argtwo$ and such that $\tau_k-\tau_j\in I$. From the properties of $j',j''$ 
  we deduce that $k\leq j'$ and then that $\tau_{j'}-\tau_j\in I$. This explains the 
  output in Line 7: $w,j\models\argone\U_{I}\argtwo$ iff $\Absolut{x_{1}}\in I$.
\end{description}

We explain now the transitions starting from state $(k,\ell)$ with $k\geq 1$.
We will once again take the help of \Cref{fig:Until-idea}. 
In state $(k,1)$, we are waiting for the witness $j'$ guessed while reading the difficult
point $j$, whereas in state $(k,2)$, we have already seen the witness $j'$ and we are
waiting for the witness $j''$ guessed while reading the difficult point $j$. 
In these states $(k, 1)$ and $(k, 2)$, the automaton encounters new positions: for each of
these positions, an easy check for whether $\argone \U_I \argtwo$ is true or not can be
done similar in spirit to Lines 2,3,4.  If not, it will be a difficult point again and
hence we are left with this question of whether we need to guess the two witnesses as in
\Cref{fig:Until-idea} for this newly encountered difficult point.  However, as we explain
below, we need not make these guesses for every difficult point.  We can restrict to
difficult points that are at a certain distance apart.

The algorithm first determines the output of the transition possibly making new guesses
(Lines 13--23), then it checks whether one (or two) of the expected witnesses is seen and
makes the necessary changes (Lines 24--36).
As in Lines 3,4 there are two cases where it is easy to determine the output: this is when
the first $\argtwo$ in the strict future is not too close (Line 14) and when the last
witness of $\argone\U\argtwo$ in the strict future is not too far away (Line 16).
Otherwise, when we reach Line 17, we have a new difficult point.  We show next that for
some of these, we can determine the output without making new guesses.

Consider a point $i$ with $j<i<j''$ and $\neg(\tau_{j''}-\tau_i<I)$ (corresponding to the 
test on Line 17 of \Cref{algo:Until-I-general} and the blue line in
\Cref{fig:Until-idea}).  We claim that $w,i\models\argone\U_I\argtwo$ iff 
$\tau_{j'}-\tau_i\in I$ or $\tau_{j''}-\tau_i\in I$. The right to left direction is 
clear. Conversely, assume that $w,i\models\argone\U_I\argtwo$ and 
$\neg(\tau_{j''}-\tau_i\in I)$. Using $\neg(\tau_{j''}-\tau_i<I)$ we get 
$\tau_{j''}-\tau_i>I$. Since all points $k$ with $j'<k<j''$ do not contain $\argtwo$, we 
deduce that the witness $i'$ for $i$ must satisfy $i'\leq j'$. Therefore, 
$\tau_{j'}-\tau_i\geq\tau_{i'}-\tau_i\in I$  and
$\tau_{j'}-\tau_i\leq\tau_{j'}-\tau_j\not> I$. We deduce that
$\tau_{j'}-\tau_i\in I$ and $j'$ is a witness for $i$.
This explains the output defined on Line 18 of \Cref{algo:Until-I-general}.

Assume now that we reach Line 19 at some point $i$ with $j<i\leq j''$: $i$ is a difficult 
point which satisfies $\tau_{j''}-\tau_i<I$.
Then, any possible witness of $w,i\models\argone\U_I\argtwo$ must be after $j''$ and
we need to make new guesses using fresh clocks, say $x_2, y_2$.  The explanations for
Lines 20--22 is similar to the description of Lines 5--7 given above.

We will call the difficult points where we start new guesses as \emph{special difficult points}.  
We show that the distance between two special difficult points is at least $c-b$ 
(which is $\ge 1$, as we consider non-singular intervals with bounds in $\mathbb{N}$).
Suppose $j=j_1<j_2<\cdots<j_k<j''$ is the sequence of special points between $j$ and $j''$.
We have $\tau_{j''}-\tau_{j_{1}}>I$ and $\tau_{j''}-\tau_{j_{2}}<I$.
We deduce that $\tau_{j_{2}}-\tau_{j_{1}}\geq c-b$.
The proof is similar for all pairs of consecutive special difficult points.  
Using $\tau_{j''}-\tau_{j_{2}}<I$, we deduce that
$(k-2)(c-b)\leq\tau_{j_k}-\tau_{j_{2}}\leq b$. 
This entails $k\leq N=2+\left\lfloor\dfrac{b}{c-b}\right\rfloor$%
\footnote{When $I$ is not both left-open and right-open, the distance between two special
difficult points is $>c-b$.  In this case, we may use the bound
$N=1+\left\lceil\dfrac{b}{c-b}\right\rceil$.}.
By the time we reach $j''$, we need to have opened at most $N$ special difficult points,
and hence we can work with the extra clocks $x_1, y_1, x_2, y_2, \dots, x_N, y_N$.

So far, we have explained Lines 13--23 of \Cref{algo:Until-I-general}.
To be more concrete, the description was for $(k,\ell)$ with $k=1$ but it is similar for 
$(k,\ell)$ with $k\neq0$. Lines 13-23 determine the output of the transition and whether 
or not we see a new special difficult point which requires new guesses (Line 21).
Notice that $k$ has been incremented on Line 20 if we are at a new special difficult 
point.

We turn now to Lines 24--36 which check whether or not the current point is a witness 
guessed in the past, make the necessary updates and determine the target state of the 
transition. 
A state $(k,\ell)$ with $k\neq0$ indicates that there are $k$ active special difficult
points currently $j=j_1<j_2<\cdots<j_k<j''$ where we have predicted
$x_1, y_1, \ldots, x_k, y_k$ respectively.  

When $\ell=1$, we are waiting for the $j'$ witness of the oldest active point and we have
the invariant:
\begin{equation}
\post\AUlast{y}\leq y_k < x_k \le y_{k-1} < x_{k-1} \le \cdots \le y_2 < x_2 \le y_1 < x_1 \leq 0 \,.
  \label{eq:inv-ell1}
\end{equation}
Notice that we may have $x_i=y_{i-1}$: the ``first'' witness ($j_i'$) of the $i^{th}$
special point could coincide with the ``second'' witness ($j_{i-1}''$) of the $(i-1)^{th}$
point.  This leads to certain subtleties, which we will come to in \Cref{ex:general-until}.

When $\ell=2$, the $j'$ witness has been seen, and we
are waiting for its $j''$ witness (the space between $\tau_{j'}$ and $\tau_{j''}$ in
\Cref{fig:Until-idea}). The invariant is
\begin{equation}
  \post\AUlast{y}\leq y_k < x_k \le y_{k-1} < x_{k-1} \le \cdots \le y_2 < x_2 \le y_1 \leq 0
  \quad\text{and}\quad x_{1}=-\infty \,.
  \label{eq:inv-ell2}
\end{equation}
\begin{description}
  \item[Line 26]  If the current point does not contain $\argtwo$, then it is not a 
  witness and the target state is simply $(k,\ell)$.

  \item[Lines 27--32] $\ell=2$ and the current point contains $\argtwo$.  Then it must be
  the witness $j''_1$ of the oldest active special point (recall that no positions between
  $j'_1$ and $j''_1$ contain $\argtwo$).  We check with $y_{1}=0$ that the time from $j_1$
  to $j''_1$ was predicted correctly.
  When this happens, $x_1,y_1$ are no longer useful.  If $k=1$, there are no other active
  special points and we simply go back to the initial state $(0,1)$.  If $k>1$, all the
  higher clocks are shifted using the permutation \textsf{shift} defined as
  $(x_1,y_1,\dots,x_k,y_k):=(x_2,y_2,\dots,x_k,y_k,x_1,y_1)$ and keeps the other clocks
  unchanged. Notice that $k$ is decremented in this case.
  
  \item[Lines 33--35] The current point contains $\argtwo$ and either $\ell=1$ (waiting 
  for witness $j'_1$) or $\ell=2$ and we have executed Lines 27--32. In the second case, 
  the current point is $j''_1$ as explained above, and we are now 
  waiting for the witness $j'_2$. Notice that we may have $j'_2=j''_1$. 
  Notice that the witness that was called $j'_2$ before the shift, is called $j'_1$ after 
  the shift.
  Hence, in both cases $\ell\in\{1,2\}$, we are waiting for the witness $j'_1$.
  
  If $x_1\neq0$ then the current point is \emph{not} the witness $j'_1$ we are waiting 
  for, in which case we simply go to state $(k,1)$ (Line 34).

  Observe that $x_1=0$ does not imply that the current point is $j'_1$, since there could
  be a sequence of positions containing $\argtwo$ occurring at the same time, and $j'_1$
  is the last of them.  Therefore, when $x_1=0$, we make a non-deterministic choice
  whether the current point is $j'_1$ (Line 35) or not (Line 34).
  
  \medskip

  \item[Lines 9--11] Finally, let us explain these lines.  
  Recall that in states $(k,\ell)$ with $k\geq1$, $\argone\U\argtwo$ holds and we are 
  waiting for witnesses.
  When the test $\neg(\postM\argone\wedge\post\AU\state=\ell_{1})$ of Line 9 succeeds, the
  current point is the last possible witness ($j_\mathrm{last}$) for $\argone\U\argtwo$.
  Notice that, if the predictions were correct, this may only happen when we see $\argtwo$ 
  and the automaton is in state $(1,2)$ (Line 10).
  Indeed, in state $(k,1)$ with $k\geq1$ the invariant~\eqref{eq:inv-ell1} implies 
  $\post\AUlast{y}\leq y_1<x_1\leq0$ which is not possible if we are at the last possible
  witness.  Similarly, in state $(k,\ell)$ with $k\geq2$ the invariant
  $\post\AUlast{y}\leq y_2<x_2\leq0$ cannot hold at the last witness.  Therefore, the
  current point is both the last witness and the second witness of $j_1$:
  $j_\mathrm{last}=j''_1$.  We check that indeed $y_1=0$ and we go back to the initial
  state (Line 11).  The output is false since $0\notin I$.
\end{description}

\begin{algorithm}[tbp]
  \small
\caption{Automaton $\AUgenI$ when $I$ is a two-sided interval, in particular, $0\notin I$.} 
\label{algo:Until-I-general}
\begin{algorithmic}[1]

  \Case{State $(0,1)$} \Comment initial state of $\AUgenI$
    \IfThen{$\neg(\postM\argone\wedge\post\AU\state=\ell_{1})$}{\Output 0; \goto (0,1)}
    \Comment $\neg(\argone\wedge\X(\argone\U\argtwo))$\EndIfThen
    \IfThen{$\neg(\Absolut{\post\AUfirst{x}}<I)$}{\Output $\Absolut{\post\AUfirst{x}}\in I$; \goto (0,1)}\EndIfThen
    \IfThen{$\neg(\Absolut{\post\AUlast{y}}>I)$}{\Output $\Absolut{\post\AUlast{y}}\in I$; \goto (0,1)}\EndIfThen
    \State $x_{1}:\in\{t\in(-\infty,0]\mid\neg(-t>I)\}$; $y_{1}:\in\{t\in(-\infty,0]\mid -t>I\}$
    \Comment Special difficult point
    \State Check $\post\AUlast{y}\leq y_1 \wedge  x_1\leq\post\AUfirst{x}$
    \State \Output $(\Absolut{x_{1}}\in I)$; \goto (1,1)  \Comment boolean value 
  \EndCase
  
  \Case{State $(k,\ell)$ with $k>0$ and $\ell\in\{1,2\}$} \Comment waiting witness
  predicted by $x_{1}$ if $\ell=1$ or $y_{1}$ if $\ell=2$
    \If{$\neg(\postM\argone\wedge\post\AU\state=\ell_{1})$}
      \Comment Last witness: $\neg(\argone\wedge\X(\argone\U\argtwo))$
      \State Check $(k,\ell)=(1,2)$ \Comment Only possible when $(k,\ell)=(1,2)$
      \State Check $\postM\argtwo\wedge y_{1}=0$; $y_{1}:=-\infty$; \Output 0; \goto (0,1)
    \Else \Comment Not the last witness
      \If{$\neg(\Absolut{\post\AUfirst{x}}<I)$}
        \State \Output $\Absolut{\post\AUfirst{x}}\in I$ \Comment boolean value
      \ElsIf{$\neg(\Absolut{\post\AUlast{y}}>I)$} 
        \State \Output $(\Absolut{\post\AUlast{y}}\in I)$ \Comment boolean value
      \ElsIf{$\neg(\Absolut{y_{k}}<I)$}
        \State \Output $(\Absolut{x_{k}}\in I) \vee (\Absolut{y_{k}}\in I)$ \Comment boolean value 
      \Else  \Comment new special difficult point 
        \State $k\gets k+1$; 
        \State $x_{k}:\in\{t\in(-\infty,0]\mid\neg(-t>I)\}$; $y_{k}:\in\{t\in(-\infty,0]\mid -t>I\}$
        \State Check $\post\AUlast{y}\leq y_k \wedge x_k\leq y_{k-1} \wedge x_k\leq\post\AUfirst{x}$;
        \Output $(\Absolut{x_{k}}\in I)$ \Comment boolean value 
      \EndIf
      \If{$\neg\postM\argtwo$} \Comment not a witness 
        \State \goto $(k,\ell)$
      \Else
        \If{$\ell=2$}
          \State Check $y_{1}=0$; $y_{1}:=-\infty$
          \Comment witness predicted by $y_{1}$ 
          \IfThen{$k=1$}{\goto $(0,1)$} 
          \EndIfThen
          \State $(x_1,y_1,\dots,x_k,y_k):=(x_2,y_2,\dots,x_k,y_k,x_1,y_1)$
          \Comment permutation shift
          \State $k\gets k-1$; 
        \EndIf
        
        \Choose{}
          \When{True}{\goto $(k,1)$}   
            \Comment not the witness predicted by (the possibly new) $x_{1}$ 
          \EndWhen
          \When{$(x_{1}=0)$}{$x_{1}:=-\infty$; \goto $(k,2)$} 
            \Comment witness predicted by $x_{1}$ 
          \EndWhen
        \EndChoose
      \EndIf
    \EndIf
  \EndCase

\end{algorithmic}
\end{algorithm}

All this discussion leads to the correctness of the construction for the general Until case, as stated by the following lemma.

\begin{lemma}\label{lem:until-general}
  Let $\Prop' = \Prop \uplus \{\argone, \argtwo\}$.  For every word $w = (a_0, \tau_0)
  (a_1, \tau_1) \cdots$ over $\Prop'$ there is a unique run 
  $C_0 \xra{\delta_0} C'_0 \xra{\overline{t}_0} C_1 \xra{\delta_1} C'_1
  \xra{\overline{t}_1} C_2 \cdots$ of $(\AUfirst,\AU,\AUlast,\AUgenI)$.
  Moreover, for all $i \ge 0$: $(w,i) \models \argone\U_I \argtwo$ iff the output of $\AUgenI$ in transition $\overline{t_i}$ is $1$. 
\end{lemma}

Finally, combining all lemmas for each of the individual operators, we may prove
\Cref{thm:full-mtl-theorem} following the same lines
of the proof of \Cref{thm:fast-mtl-theorem} in \Cref{sec:proof-of-fast-mtl}.

\begin{remark}
  \label{rmk:finite-words}
  Note that \Cref{thm:fast-mtl-theorem} and \Cref{thm:full-mtl-theorem} can
  be extended to the case of finite words, where we want to reach a bad configuration
  instead of eventually remaining in bad configurations.  The only modification required
  concerns the future modalities operator.  We can use the same automata defined above,
  but the definition of a bad configuration $C$ for finite words asks, in addition to the
  existing conditions, that (1) for each occurrence of $\AX$ the final state is
  $\ell_{0}$ ($C(\pre{\AX}\state)=\ell_{0}$) and (2) for each occurrence
  of $\AU$ the final state is $\ell_{0}$ ($C(\pre\AU\state)=\ell_{0}$).
\end{remark}

\begin{table}[tbp]
  \hspace{-20mm}
$\small
\begin{array}{r@{\colon}l} 
  \Delta_{01}^{01} & 
  \left\{\!\!\!\begin{array}{l@{\;\mid\;}ll}
    \neg C & 0 & \text{\color{red} L2} \\
    C \wedge -x\geq b & (-x\in I) & \text{\color{red} L3} \\
    C \wedge -x<b \wedge -y<c & (-y\in I) & \text{\color{red} L4} 
  \end{array}\right.
  \\[5ex]
  \Delta_{01}^{11} &
    C \wedge -x<b \wedge -y\geq c; 
    [x_1,y_1]; y\leq y_1\leq-c<x_1\leq x
    \mid (-x_1\in I)
    \quad\text{\color{red} L5--7}
  \\[4ex]
  \Delta_{11}^{11} & 
  \left\{\!\!\!\begin{array}{l@{\;\mid\;}ll}
    C \wedge -x\geq b & (-x\in I) & \text{\color{red} L14} \\
    C \wedge -x<b \wedge -y<c & (-y\in I) & \text{\color{red} L16} \\
    C \wedge -x<b \wedge -y\geq c \wedge -y_1\geq b & (-x_1\in I) \vee (-y_1\in I) & \text{\color{red} L18}
  \end{array}\right\} \text{\color{red} + (25 or 34)}
  \\[5ex]
  \Delta_{11}^{12} & 
  \left\{\!\!\!\begin{array}{l@{\;\mid\;}ll}
    C \wedge \postM\argtwo \wedge x_1=0 \wedge -x\geq b; x_1:=-\infty & (-x\in I) & \text{\color{red} L14} \\
    C \wedge \postM\argtwo \wedge x_1=0 \wedge -x<b \wedge -y<c; x_1:=-\infty & (-y\in I) & \text{\color{red} L16} \\
    C \wedge \postM\argtwo \wedge x_1=0 \wedge -x<b \wedge -y\geq c \wedge -y_1\geq b; x_1:=-\infty & (-y_1\in I) & \text{\color{red} L18} 
  \end{array}\right\} \text{\color{red} + 35}
  \\[5ex]
  \Delta_{11}^{21} &
    C \wedge -x<b \wedge -y\geq c \wedge -y_1<b; 
    [x_2,y_2]; y\leq y_2\leq-c<x_2\leq y_1
    \mid (-x_2\in I)
    \quad\text{\color{red} L20--22 + (25 or 34)}
  \\[1ex]
  \Delta_{11}^{22} & 
    C \wedge \postM\argtwo \wedge x_1=0 \wedge -x<b \wedge -y\geq c \wedge -y_1<b; 
    x_1:=-\infty; 
    [x_2,y_2]; y\leq y_2\leq-c<x_2\leq y_1
    \mid (-x_2\in I) 
    \quad\text{\color{red} L20--22 + 35}
  \\[4ex]
  \Delta_{21}^{21} & 
  \left\{\!\!\!\begin{array}{l@{\;\mid\;}ll}
    C \wedge -x\geq b & (-x\in I) & \text{\color{red} L14} \\
    C \wedge -x<b \wedge -y<c & (-y\in I) & \text{\color{red} L16} \\
    C \wedge -x<b \wedge -y\geq c \wedge -y_2\geq b & (-x_2\in I) \vee (-y_2\in I) & \text{\color{red} L18} 
  \end{array}\right\} \text{\color{red} + (25 or 34)}
  \\[5ex]
  \Delta_{21}^{22} & 
  \left\{\!\!\!\begin{array}{l@{\;\mid\;}ll}
    C \wedge \postM\argtwo \wedge x_1=0 \wedge -x\geq b; x_1:=-\infty & (-x\in I) & \text{\color{red} L14} \\
    C \wedge \postM\argtwo \wedge x_1=0 \wedge -x<b \wedge -y<c; x_1:=-\infty & (-y\in I) & \text{\color{red} L16} \\
    C \wedge \postM\argtwo \wedge x_1=0 \wedge -x<b \wedge -y\geq c \wedge -y_2\geq b; x_1:=-\infty & (-x_2\in I) \vee (-y_2\in I) & \text{\color{red} L18} 
  \end{array}\right\} \text{\color{red} + 35}
  \\[5ex]
  \Delta_{21}^{31} &
    C \wedge -x<b \wedge -y\geq c \wedge -y_2<b; 
    [x_3,y_3]; y\leq y_3\leq-c<x_3\leq y_2
    \mid (-x_3\in I)
    \quad\text{\color{red} L20--22 + (25 or 34)}
  \\[1ex]
  \Delta_{21}^{32} & 
    C \wedge \postM\argtwo \wedge x_1=0 \wedge -x<b \wedge -y\geq c \wedge -y_2<b; 
    x_1:=-\infty; 
    [x_3,y_3]; y\leq y_3\leq-c<x_3\leq y_2
    \mid (-x_3\in I) 
    \quad\text{\color{red} L20--22 + 35}
  \\[4ex]
  \Delta_{31}^{31} & 
  \left\{\!\!\!\begin{array}{l@{\;\mid\;}ll}
    C \wedge -x\geq b & (-x\in I) & \text{\color{red} L14} \\
    C \wedge -x<b \wedge -y<c & (-y\in I) & \text{\color{red} L16} \\
    C \wedge -x<b \wedge -y\geq c \wedge -y_3\geq b & (-x_3\in I) \vee (-y_3\in I) & \text{\color{red} L18} 
  \end{array}\right\} \text{\color{red} + (25 or 34)}
  \\[5ex]
  \Delta_{31}^{32} & 
  \left\{\!\!\!\begin{array}{l@{\;\mid\;}ll}
    C \wedge \postM\argtwo \wedge x_1=0 \wedge -x\geq b; x_1:=-\infty & (-x\in I) & \text{\color{red} L14} \\
    C \wedge \postM\argtwo \wedge x_1=0 \wedge -x<b \wedge -y<c; x_1:=-\infty & (-y\in I) & \text{\color{red} L16} \\
    C \wedge \postM\argtwo \wedge x_1=0 \wedge -x<b \wedge -y\geq c \wedge -y_3\geq b; 
    x_1:=-\infty & (-x_3\in I) \vee (-y_3\in I) & \text{\color{red} L18} 
  \end{array}\right\} \text{\color{red} + 35}
\end{array}$
  
  \bigskip
  \caption{Labels of transitions from states $(k,1)$ with $0\leq k\leq 3$. \\
  We assume that $I=[b,c)$ with $0<b<c<+\infty$ and $N=3$
  (where $N=1+\left\lceil\dfrac{b}{c-b}\right\rceil$). \\
  To lighten the notation, we simply write $x$ for $\post\AUfirst{x}$ and $y$ for 
  $\post\AUlast{y}$. \\
  We also let $C=\postM\argone\wedge\post\AU\state=\ell_{1}$ corresponding
  to $\argone\wedge\X(\argone\U\argtwo)$. \\
  The release of a future clock is written $[z]$ and stands for $z:\in(-\infty,0]$. \\
  Actually, the invariants imply that $C$ holds in all states $(k,1)$ with
  $k\geq1$ or $(k,2)$ with $k\geq2$.  \\
  Hence we can remove the $C$ check from all transitions starting from those states. \\
  Notice also that $\argone\U\argtwo$ holds in all states $(k,\ell)\neq(0,1)$.  
  Hence in those states, $\neg\argtwo$ implies $C$ and we may also remove $C$ from
  transitions checking $\neg\postM\argtwo$ starting from state $(1,2)$.
  }
  \label{tbl:labels-trans-1}
\end{table}

\begin{table}[tbp]
  \hspace{-10mm}
$\small
\begin{array}{r@{\colon}l} 
  \Delta_{12}^{01} &
  \left\{\!\!\!\begin{array}{l@{\;\mid\;}ll}
    \neg C \wedge \postM\argtwo \wedge y_1=0; y_1:=-\infty & 0 & \text{\color{red} L11} \\
    C \wedge \postM\argtwo \wedge y_1=0 \wedge -x\geq b; y_1:=-\infty & (-x\in I) & \text{\color{red} L14 + 28,29} \\
    C \wedge \postM\argtwo \wedge y_1=0 \wedge -x<b \wedge -y<c; y_1:=-\infty & (-y\in I) & \text{\color{red} L16 + 28,29} 
  \end{array}\right.
  \\[5ex]
  \Delta_{12}^{11} &
    C \wedge \postM\argtwo \wedge y_1=0 \wedge -x<b \wedge -y\geq c; 
    [x_1,y_1]; y\leq y_1\leq-c< x_1\leq x
    \mid (-x_1\in I)
    \quad\text{\color{red} L20--22 + 28,30,31 + 34}
  \\[1ex]
  \Delta_{12}^{12} & 
  \left\{\!\!\!\begin{array}{l@{\;\mid\;}ll}
    C \wedge \neg\postM\argtwo \wedge -x\geq b & (-x\in I) & \text{\color{red} L14+25} \\
    C \wedge \neg\postM\argtwo \wedge -x<b \wedge -y<c & (-y\in I) & \text{\color{red} L16+25} \\
    C \wedge \neg\postM\argtwo \wedge -x<b \wedge -y\geq c \wedge -y_1\geq b & (-y_1\in I) & \text{\color{red} L18+25} \\
    C \wedge \postM\argtwo \wedge -x<b \wedge -y\geq c \wedge y_1=0; 
    [y_1]; y\leq y_1\leq-c
    & 0 & \text{\color{red} L20--22 + 28,30,31 + 35} 
  \end{array}\right.
  \\[6ex]
  \Delta_{12}^{22} &
    C \wedge \neg\postM\argtwo \wedge -x<b \wedge -y\geq c \wedge -y_1<b; 
    [x_2,y_2]; y\leq y_2\leq-c< x_2\leq y_1
    \mid (-x_2\in I)
    \quad\text{\color{red} L20--22 + 25}
  \\[3ex]
  \Delta_{22}^{22} & 
  \left\{\!\!\!\begin{array}{l@{~}ll}
    C \wedge \neg\postM\argtwo \wedge -x\geq b &\mid (-x\in I) & \text{\color{red} L14 + 25} \\
    C \wedge \neg\postM\argtwo \wedge -x<b \wedge -y<c &\mid (-y\in I) & \text{\color{red} L16 + 25} \\
    C \wedge \neg\postM\argtwo \wedge -x<b \wedge -y\geq c \wedge -y_2\geq b &\mid (-y_2\in I) & \text{\color{red} L18 + 25} \\
    C \wedge \postM\argtwo \wedge -x<b \wedge -y\geq c \wedge -y_2<b \wedge y_1=0 \wedge x_2=0; \\
    \hspace{20mm} y_1:=y_2; 
    [x_2,y_2]; y\leq y_2\leq-c< x_2\leq y_1
    &\mid (-x_2\in I)
    & \text{\color{red} L20--22 + 28,30,31 + 35} 
  \end{array}\right.
  \\[7ex]
  \Delta_{22}^{32} &
    C \wedge \neg\postM\argtwo \wedge -x<b \wedge -y\geq c \wedge -y_2<b; 
    [x_3,y_3]; y\leq y_3\leq-c< x_3\leq y_2
    \mid (-x_3\in I)
    \quad\text{\color{red} L20--22 + 25}
  \\[1ex]
  \Delta_{22}^{11} &
  \left\{\!\!\!\begin{array}{l@{~}ll}
    C \wedge \postM\argtwo \wedge y_1=0 \wedge -x\geq b; \\
    \hspace{20mm} (x_1,y_1,x_2,y_2):=(x_2,y_2,-\infty,-\infty) &\mid (-x\in I) 
    & \text{\color{red} L14 + 28,30,31 + 34} \\
    C \wedge \postM\argtwo \wedge y_1=0 \wedge -x<b \wedge -y<c; \\
    \hspace{20mm} (x_1,y_1,x_2,y_2):=(x_2,y_2,-\infty,-\infty) &\mid (-y\in I) 
    & \text{\color{red} L16 + 28,30,31 + 34} \\
    C \wedge \postM\argtwo \wedge y_1=0 \wedge -x<b \wedge -y\geq c \wedge -y_2\geq b; \\
    \hspace{20mm} (x_1,y_1,x_2,y_2):=(x_2,y_2,-\infty,-\infty) &\mid (-x_1\in I)\vee(-y_1\in I)
    & \text{\color{red} L18 + 28,30,31 + 34} 
  \end{array}\right.
  \\[8ex]
  \Delta_{22}^{21} & 
  \left\{\!\!\!\begin{array}{l@{~}ll}
    C \wedge \postM\argtwo \wedge y_1=0 \wedge -x<b \wedge -y\geq c \wedge -y_2<b; \\
    \hspace{20mm} (x_1,y_1):=(x_2,y_2); 
    [x_2,y_2]; y\leq y_2\leq-c< x_2\leq y_1
    &\mid (-x_2\in I)
    & \text{\color{red} L20--22 + 28,30,31 + 34}
  \end{array}\right.
  \\[3ex]
  \Delta_{22}^{12} &
  \left\{\!\!\!\begin{array}{l@{~}ll}
    C \wedge \postM\argtwo \wedge y_1=0 \wedge x_2=0 \wedge -x\geq b; \\
    \hspace{20mm} (y_1,x_2,y_2):=(y_2,-\infty,-\infty) &\mid (-x\in I) 
    & \text{\color{red} L14 + 28,30,31 + 35} \\
    C \wedge \postM\argtwo \wedge y_1=0 \wedge x_2=0 \wedge -x<b \wedge -y<c; \\
    \hspace{20mm} (y_1,x_2,y_2):=(y_2,-\infty,-\infty) &\mid (-y\in I) 
    & \text{\color{red} L16 + 28,30,31 + 35} \\
    C \wedge \postM\argtwo \wedge y_1=0 \wedge x_2=0 \wedge -x<b \wedge -y\geq c \wedge -y_2\geq b; \\
    \hspace{20mm} (y_1,x_2,y_2):=(y_2,-\infty,-\infty) &\mid (-y_1\in I) 
    & \text{\color{red} L18 + 28,30,31 + 35} 
  \end{array}\right.
  \\[10ex]
  \Delta_{32}^{32} & 
  \left\{\!\!\!\begin{array}{l@{~}ll}
    C \wedge \neg\postM\argtwo \wedge -x\geq b &\mid (-x\in I) & \text{\color{red} L14 + 25} \\
    C \wedge \neg\postM\argtwo \wedge -x<b \wedge -y<c &\mid (-y\in I) & \text{\color{red} L16 + 25} \\
    C \wedge \neg\postM\argtwo \wedge -x<b \wedge -y\geq c \wedge -y_3\geq b &\mid (-y_3\in I) & \text{\color{red} L18 + 25} \\
    C \wedge \postM\argtwo \wedge -x<b \wedge -y\geq c \wedge -y_3<b \wedge y_1=0 \wedge x_2=0; \\
    \hspace{10mm} (y_1,x_2,y_2):=(y_2,x_3,y_3); 
    [x_3,y_3]; y\leq y_3\leq-c< x_3\leq y_2
    &\mid (-x_3\in I)
    & \text{\color{red} L20--22 + 28,30,31 + 35} 
  \end{array}\right.
  \\[7ex]
  \Delta_{32}^{21} &
  \left\{\!\!\!\begin{array}{l@{~}ll}
    C \wedge \postM\argtwo \wedge y_1=0 \wedge -x\geq b; \\
    \hspace{10mm} (x_1,y_1,x_2,y_2,x_3,y_3):=(x_2,y_2,x_3,y_3,-\infty,-\infty)
    &\mid (-x\in I) & \text{\color{red} L14 + 28,30,31 + 34} \\
    C \wedge \postM\argtwo \wedge y_1=0 \wedge -x<b \wedge -y<c; \\
    \hspace{10mm} (x_1,y_1,x_2,y_2,x_3,y_3):=(x_2,y_2,x_3,y_3,-\infty,-\infty)
    &\mid (-y\in I) & \text{\color{red} L16 + 28,30,31 + 34} \\
    C \wedge \postM\argtwo \wedge y_1=0 \wedge -x<b \wedge -y\geq c \wedge -y_3\geq b; \\
    \hspace{10mm} (x_1,y_1,x_2,y_2,x_3,y_3):=(x_2,y_2,x_3,y_3,-\infty,-\infty)
    &\mid (-x_2\in I)\vee(-y_2\in I) & \text{\color{red} L18 + 28,30,31 + 34} 
  \end{array}\right.
  \\[8ex]
  \Delta_{32}^{31} & 
  \left\{\!\!\!\begin{array}{l@{~}ll}
    C \wedge \postM\argtwo \wedge y_1=0 \wedge -x<b \wedge -y\geq c \wedge -y_3<b; \\
    \hspace{10mm} (x_1,y_1,x_2,y_2):=(x_2,y_2,x_3,y_3); 
    [x_3,y_3]; y\leq y_3\leq-c< x_3\leq y_2
    &\mid (-x_3\in I) & \text{\color{red} L20--22 + 28,30,31 + 34}
  \end{array}\right.
  \\[3ex]
  \Delta_{32}^{22} &
  \left\{\!\!\!\begin{array}{l@{~}ll}
    C \wedge \postM\argtwo \wedge y_1=0 \wedge x_2=0 \wedge -x\geq b; \\
    \hspace{10mm} (y_1,x_2,y_2,x_3,y_3):=(y_2,x_3,y_3,-\infty,-\infty) 
    &\mid (-x\in I) & \text{\color{red} L14 + 28,30,31 + 35} \\
    C \wedge \postM\argtwo \wedge y_1=0 \wedge x_2=0 \wedge -x<b \wedge -y<c; \\
    \hspace{10mm} (y_1,x_2,y_2,x_3,y_3):=(y_2,x_3,y_3,-\infty,-\infty) 
    &\mid (-y\in I) & \text{\color{red} L16 + 28,30,31 + 35} \\
    C \wedge \postM\argtwo \wedge y_1=0 \wedge x_2=0 \wedge -x<b \wedge -y\geq c \wedge -y_3\geq b; \\
    \hspace{10mm} (y_1,x_2,y_2,x_3,y_3):=(y_2,x_3,y_3,-\infty,-\infty) 
    &\mid (-y_2\in I) & \text{\color{red} L18 + 28,30,31 + 35} 
  \end{array}\right.
\end{array}$
  
  \bigskip
  \caption{Labels of transitions from states $(k,2)$ with $1\leq k\leq 3$. \\
  }
  \label{tbl:labels-trans-2}
\end{table}

\begin{example}\label{ex:general-until}
  To understand the General Until construction better, let us explain how
  \Cref{algo:Until-I-general} works for the case of $N=3$ and produces the
  automaton depicted in \Cref{fig:Until-I-gen-automaton}.  The transitions of the
  automaton are labeled $\Delta_{ij}^{k\ell}$ to denote a transition going from state
  $(i,j)$ (denoted $ij$ in the automaton) to state $(k,\ell)$ and are detailed in
  \Cref{tbl:labels-trans-1,tbl:labels-trans-2}.  Our goal in doing so is not
  to argue correctness of the algorithm (which has been done earlier), but to show how the
  algorithm produces the resulting automaton.  \Cref{tbl:labels-trans-1} considers
  all transitions from states of the form $(k,1)$ for different values of
  $k\in\{0,1,2,3\}$ (i.e., the top 4 states in \Cref{fig:Until-I-gen-automaton}),
  while \Cref{tbl:labels-trans-2} considers the transitions from states $(k,2)$ (the
  bottom 3 states).

  First, note that we consider an interval that is closed on one side and open on the
  other side, i.e., $I=[b,c)$ with $0<b<c<+\infty$, so that we can demonstrate how the
  rather abstract guards in the algorithm (e.g., $x \in I, x < I, x \geq I$) are
  instantiated in the case of both open and closed intervals.
  Second, note that the number of states in the automaton $\AUgenI$ is $2N+1=7$ and the
  number of clocks is $2N=6$, which matches the bounds given in
  \Cref{thm:full-mtl-theorem} for this case.

  Next, we explain the states from which there are some \emph{special} transitions, namely
  the states $(0,1)$ and $(1,2)$.
  These are the states from which you could potentially see a position that satisfies
  $\neg C$, where $C$ denotes $\postM\argone\wedge\post\AU\state=\ell_{1}$ corresponding
  to $\argone\wedge\X(\argone\U\argtwo)$.  $\neg C$ means that at the current position,
  there are no witnesses of $\argone\U\argtwo$ in the strict future.

  \noindent {\em From state $(0,1)$.} 
  The automaton starts in state $(0,1)$ and stays in state $(0,1)$ according to transition
  $\Delta_{01}^{01}$ on the self-loop corresponding simply to lines \Red{2,3,4}.  These
  are the cases where, either the formula $\argone\U_I\argtwo$ is not satisfied, or the
  satisfaction of the formula can be determined by only using the first witness or the
  last witness.

  If the satisfaction of the current position using only the first and last witnesses
  cannot be determined, we need to guess a new set of potential witnesses and verify it,
  towards which the automaton takes the transition $\Delta_{01}^{11}$ to $(1,1)$.  As can
  be seen the guard and output corresponds to conjunction of lines \Red{5-7}.

  \noindent {\em From states $(k,1)$.} 
  Now, from each state $(k,1)$, we check if one of the expected witnesses is seen; the
  first $x$, the last $y$ or the $k^{th}$ set of witnesses are seen corresponding to lines
  \Red{14,16,18}.  We also make the necessary changes: either stay in the same state
  (lines \Red{25} or \Red{34}) or go to $(k,2)$ according to line \Red{35}.  This
  describes transitions $\Delta_{k1}^{k1},\Delta_{k1}^{k2}$ for different values of $k$.
  Finally, starting new clocks $x_k,y_k$ is done when going from $(k,i)$ to $(k+1,i)$ or
  $(k+1,i+1)$ as in lines \Red{21,22}.

  \noindent{\em From states $(k,2)$.} 
  In the second table, we describe the remaining transitions in the automaton.  Again there
  are a few patterns of interest worth noting, in addition to what was described above.
  In transition $\Delta_{12}^{01}$ we go back to state $(0,1)$ and hence use the guard as in
  lines \Red{9,11}.
  In transitions such as $\Delta_{22}^{12}$ and indeed anywhere where lines
  \Red{28,30,31} are used, the idea is that $x_1,y_1$ (respectively $x_k$, $y_k$) are not
  useful and hence we shift all the larger indexed clocks and thus decrement $k$ (or keep
  $k$ same if two more clocks are being introduced as in $\Delta_{22}^{22})$.

  We would next like to point out a subtle point in the construction.  In
  transitions such as $\Delta_{22}^{12}$, $\Delta_{22}^{21}$, $\Delta_{32}^{21}$ and
  $\Delta_{32}^{31}$, we are combining multiple lines of the algorithm together to produce
  a single transition in the automaton.  In doing so, we are simplifying the transitions
  by combining some operations and only displaying the final results of the
  simplification.

  For instance, consider the last transition in $\Delta_{22}^{22}$.
  Since we are in state $(2,2)$, we have two pairs of active of clocks - $(x_1,y_1)$ and
  $(x_2,y_2)$ - that are tracking the potential witnesses.
  Lines \Red{20--22} will release a new pair of clocks $(x_3,y_3)$ to track a new set of
  potential witnesses.
  However, since we are subsequently executing the lines \Red{28,30,31}, we are also
  shifting the existing pairs of clocks.
  Therefore, on this transition, we do not refer to the clocks $x_3$ and $y_3$, as they
  have been shifted to $x_2$ and $y_2$ respectively, by the end of the transition.
  Further, notice that the output of this transition is $-x_3 \in I$, as indicated in line
  \Red{22}.  However, since $x_3$ has been shifted to $x_2$ by the end of the transition,
  this is now represented as $-x_2 \in I$.
  
  This is a subtle part of the construction.
  The high level idea is that when we combine lines of the algorithm, the
  transition that is resultant in the automaton is simpler.  
  To illustrate this, we consider again the last transition in $\Delta_{22}^{22}$ that is
  obtained by executing the lines \text{\color{red} 20--22 + 28,30,31 + 35} of
  \Cref{algo:Until-I-general}.  We give both the complete form and the simplified form.
  The full form of this transition is
  \[
  \begin{alignedat}{3}
    & C \wedge -x<b \wedge -y\geq c \wedge -y_2<b; 
    [x_3,y_3];\ y\leq y_3\leq -c < x_3\leq y_2
    &\quad& \mid (-x_3\in I)
    &\quad& \text{\color{red} L20--22} \\[4pt]
    & \postM\argtwo\wedge y_1=0;\ y_1 := -\infty;\ 
    (x_1,y_1,x_2,y_2,x_3,y_3)
    := (x_2,y_2,x_3,y_3,-\infty,-\infty)
    &\quad&
    &\quad& \text{\color{red} L28,30,31} \\[4pt]
    & x_1=0;\ x_1 := -\infty
    &\quad&
    &\quad& \text{\color{red} L35}
  \end{alignedat}
  \]
  The simplified form of this transition is 
  \[
  \begin{alignedat}{3}
    & C \wedge \postM\argtwo  \wedge -x<b  \wedge -y\geq c  \wedge -y_2<b  \wedge y_1=0  \wedge x_2=0; \\[2pt]
    & y_1 := y_2;\   [x_2,y_2];\   y\leq y_2\leq -c < x_2\leq y_1&\quad& \mid (-x_2\in I)&\quad& \text{\color{red} L20--22 + 28,30,31 + 35}
  \end{alignedat}
  \]
  
\begin{remark}
  Note that the full form contains the operation $x_1 := -\infty$, which sets the clock
  $x_1$ to $-\infty$ and hence deactivates it, while the simplified form does not contain
  this operation.  This is because the transition starts from state $(2,2)$ and by the
  invariant~\eqref{eq:inv-ell2} we already have $x_1=-\infty$ at the beginning of the
  transition.
\end{remark}  

\end{example}

\section{Improved simulation check and liveness for \GTA} 
\label{sec:gta-liveness}

In this section, we describe our technical contributions that have gone into the
development of the $\GTAp$ tool. We start with an overview of these contributions.

A liveness algorithm for timed automata or GTA computes a symbolic abstraction called the
\emph{zone graph}.  Zone graphs are known to be infinite and termination mechanisms for
zone computations have seen a long line of work~(see~\cite{DBLP:conf/formats/BouyerGHSS22}
for a survey).  A node in the zone graph is of the form $(q, Z)$ where $q$ is a control
state of the network of automata and $Z$ is a constraint system on clocks that represents
the set of configurations abstracted by the current node.  For reachability, termination
is guaranteed by using \emph{finite simulation relations}, which are binary relations on
nodes $(q, Z) \preceq (q, Z')$.  For liveness, to ensure soundness, simulation needs to be
replaced with a \emph{mutual simulation}: $(q, Z) \preceq (q, Z')$ and $(q, Z') \preceq
(q, Z)$.  The procedure for checking $(q, Z) \preceq (q, Z')$ is the most used operation
in the zone graph computation and it is extremely crucial to have efficient algorithms for
the same.  When the automaton contains no \emph{diagonal constraints} (guards of the form
$x - y \leqlt c)$, the check can be done in $\mathcal{O}(|X|^2)$ where $X$ is the total
number of clocks.  When diagonal constraints are present, the check is known to be
NP-hard, and it has a running time of $\mathcal{O}(|X|^2 \cdot 2^d)$ where $d$ is the
number of diagonals present.  Essentially, the check calls the diagonal-free module at
most $2^d$ times.

\smallskip

In the $\GTAp$ tool, we have implemented an improved simulation check for GTAs with future
clocks, compared to the simulation check of~\cite{AGGJS-CAV23}.  This results in coarser
simulations and smaller parts of the zone graph being explored.  There are cases where 
the gain is exponential, as for instance with GTAs of the form of 
\Cref{fig:G-sim-example}.
We also use a new direct algorithm for the mutual simulation check, that bypasses the need
to call the basic simulation check twice.  This reduces the computation time by half.

\begin{figure}[htbp]
  \centering\noindent
  \includegraphics[page=22,scale=1]{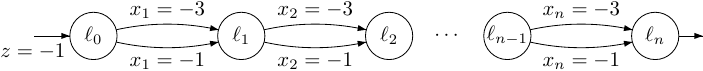}
  \caption{A \GTA\ that shows the exponential improvement produced smaller $\Gg$-sets.}
  \label{fig:G-sim-example}
\end{figure}

Finally, another important contribution is that we lift and adapt Couvreur's algorithm
\cite{Couvreur} for B\"uchi liveness to the setting of generalized timed automata.

We present below our procedure for the liveness problem of safe \GTAs. 
Decidability of liveness for safe \GTAs\ was recently shown in~\cite[Theorem~11]{AGGS-CONCUR24}.
A \GTA\ is live, i.e., it admits an infinite accepting non-Zeno run, iff in its zone 
graph there is a reachable accepting cycle, i.e., a reachable zone $(q,Z)$ and a path
$(q,Z)\xra{t_{1}}(q_{1},Z_{1})\xra{t_{2}}\cdots(q_{k},Z_{k})\xra{t_{k+1}}(q,Z')$ such that
$(q,Z)\preceq(q,Z')$ and $(q,Z')\preceq(q,Z)$, 
each B\"uchi condition is witnessed along 
this cycle, and each future clock $x\in X_{F}$ is either released along the cycle 
or could take value $-\infty$ (there is a valuation $v_{x}\in Z$ with $v_{x}(x)=-\infty$).

There are several ways to implement this liveness check.  Our choice is to extend
Couvreur's SCC algorithm~\cite{Couvreur,Couvreur-1,Couvreur-2} (the state-of-the-art
approach for checking liveness of timed automata) to the setting of \GTAs.  Our procedure
can be thought of as an adaptation of the idea for checking liveness of \GTA\ proposed
in~\cite{AGGS-CONCUR24} to the well-studied framework of Couvreur's SCC algorithm.  On top
of the usual liveness algorithm, we need to perform extra checks to determine if future
clocks are released in SCCs.

The idea for checking liveness in~\cite{AGGS-CONCUR24} ultimately boils down to
identifying cycles in zone graphs that satisfy certain special properties.  Couvreur's SCC
algorithm provides an optimal method for detecting cycles.  \Cref{alg:liveness},
given in Appendix~\ref{sec:couvreur-algorithm}, is an SCC-based liveness algorithm to
check if a GTA $\gtasystem$ has an accepting cycle with respect to some generalized
B\"uchi acceptance condition. It leverages the various optimizations of Couvreur's SCC
algorithm~\cite{Couvreur} that makes liveness checking more efficient for timed automata.
Notice that, to the best of our knowledge, Couvreur's algorithm has not been previously
extended to event-clock automata or timed automata with diagonal constraints.
Hence, \Cref{alg:liveness} is the first extension of Couvreur's liveness
procedure for these classes.

\section{Implementation and Experimental evaluation}\label{sec:experiments}

In this section, we present an experimental evaluation of our prototype implementation
tool \tool\ for satisfiability and model checking for $\MTLfp$ formulas.  The tool
contains two modules:
\begin{itemize}
    \item \rm{MITL2GTA} that implements our optimized translation from $\detMTL$ to
    networks of deterministic timed automata (from \Cref{sec:fastmtl-to-ta}) and the
    optimized translation from $\MTLfp$ to networks of $\GTA$ (from \Cref{sec:future}),
    using the improvements (sharing and two-state until) for future modalities.

    \item \GTAp\ where we have implemented the new SCC-based liveness algorithm (discussed in \Cref{sec:gta-liveness}) and an improved reachability algorithm for $\GTA$.
    This is an enhanced version of \GTA\ from~\cite{AGGJS-CAV23}, built on top of the open-source timed automata library \Tchecker~\cite{TChecker}.
\end{itemize}
By combining these two modules, we obtain a new pipeline \tool (\rm{MITL2GTA}+\GTAp) for satisfiability and model checking for $\MTLfp$ formulas (with both past and future modalities) over timed words under pointwise semantics.
Our tool, along with the benchmarks used in this paper, is available and can be downloaded from \href{https://github.com/EQuaVe/TEMPORA}{https://github.com/EQuaVe/TEMPORA}. 

We organize our evaluation into three parts as follows:
\begin{itemize}[left=0.1em]
  \item First, we evaluate the performance of the $\detMTL$ to deterministic timed
  automata translation on a set of benchmarks, which include both standalone past
  modalities and combinations with future modalities.  To the best of our knowledge, in
  the pointwise semantics, there are no other tools that tackle this fragment.
  
  \item Next, we compare our $\MTLfp$ to $\GTA$ translation against two competing
  approaches in the pointwise semantics, namely:
  \begin{enumerate}[left=0.1em]
    \item \textbf{\MightyL~\cite{MightyL}}: a state-of-the-art tool that translates
    $\MITL$ to OCATA and then to networks of timed automata.  The resulting automata are
    analyzed using $\UPPAAL$~\cite{UPPAAL-1} (for finite words) and \Opaal~\cite{OPAAL}
    (for infinite words).

    \item \textbf{Baseline GTA Translation} an implementation of the (theoretical)
    translation from $\MITL$ to $\GTA$ proposed in~\cite{AGGS-CONCUR24}.
  \end{enumerate}
  
  \item Finally, to demonstrate the full scope of our implementation, we present our full
  model-checking pipeline on two classical examples, namely \textrm{Fischer} and
  \textrm{Dining Philosophers} with $\MTLfp$ formulae capturing some standard properties.
\end{itemize}

\medskip\noindent \textbf{Experimental setup}

\noindent Our tool \tool\ takes as input a $\MTLfp$ formula and applies the appropriate
translation described in \Cref{sec:fastmtl-to-ta,sec:future}, producing a network of
$\GTA$ (which is deterministic, when the formula is in $\detMTL$).  Then, checking
satisfiability of a formula over finite timed words in the pointwise semantics reduces to
checking emptiness of the corresponding network of $\GTA$.  We also adapted our translation to work
with finite timed words as explained in \Cref{rmk:finite-words}.  We use an improved
version of $\GTA$ reachability algorithm from~\cite{AGGJS-CAV23}, which is built on the
open source tool \Tchecker~\cite{TChecker}.  Over infinite words, checking satisfiability
reduces to checking B\"uchi acceptance in $\GTA$.  For this, we implement the liveness
algorithm from~\cite{AGGS-CONCUR24}, again on top of TChecker.
Finally, we implemented a prototype model-checking algorithm for $\MTLfp$ formulae w.r.t.
a network of $\GTA$ with the pointwise semantics.

All experiments were run on an Ubuntu 18.04 machine with an Intel i7 3.40GHz processor and
32 GB RAM. In this section, we describe the benchmarks and present an analysis of the
results.  Details of our implementation, including the various engineering optimisations
that we employ, are discussed in Appendix~\ref{sec:app:experiments}.

\medskip\noindent \textbf{Benchmarks.}
We conduct our experiments on a suite of 72 benchmarks, which we have grouped into three
groups: (1) benchmarks designed to demonstrate past operators (alongside future
modalities), (2) benchmarks from \MightyL~\cite{MightyL} that are themselves based on
examples from~\cite{GastinO01,Plaku2015}, and (3) benchmarks from Acacia synthesis tool's
benchmark suite~\cite{Acacia}.

\noindent \textit{Group 1.} This group consists of 22 $\MITL$ formulae designed to
highlight various strengths of our translation.  These include formulae featuring
conjunctions of temporal operators, the simultaneous use of past and future modalities
(for which, to the best of our knowledge, no existing tool supports), and other
representative constructs.  We rewrite the derived operators $\Always, \Eventually,
\Release$ in terms of $\U$, in the standard way.  We use the following shorthand for
common patterns:
$\rho(n) = \Next(n) \wedge \Always (p_1 \vee p_2)$, where $\Next(n)  =  \Next_{[2,3]}
\big( (p_1  \implies  \Next(n-1)) \wedge (p_2  \implies  \neg\Next(n-1)) \big)$ with
base case of $\Next(1) = p_1$,
$\eta(n) = \Always (\bigwedge_{i=1}^{n} (p_i   \implies   (a_i \Since_{[0, 2]} b_i)))
\wedge \Always(\bigvee_{i=1}^{n} p_i)$,
$\Reqgrant = \Always_{[2,\infty)} \left((R \implies  \Eventually_{[4,5]} G) \right)$, and
$\Spotex = \Always_{[0, 10]} \Eventually_{[1, 2]} (a  \iff  \Next b)$.

\noindent \textit{Group 2.} This group contains all 33 benchmarks from~\cite{MightyL},
accompanying the tool \MightyL. These are based on temporal patterns from~\cite{GastinO01,Plaku2015}, and include families of parameterized formulae:
$\Eventually(k,I) = \bigwedge_{i=1}^{k} \Eventually_{I} p_i$, $\Always(k,I) = \bigwedge_{i=1}^{k} \Always_{I} p_i$,
$\U(k,I) = ( \cdots (p_1 \U_I p_2) \cdots ) \U_I p_k$, and
$\Release(k,I) = ( \cdots (p_1 \Release_I p_2) \cdots ) \Release_I p_k$, $\mu(k,I) =
\bigwedge_{i=1}^{k} \Eventually_{[3i-1,3i]} t_i \wedge \Always \neg p$ and
$\theta(k,I) = \neg((\bigwedge_{i=1}^{k} \Always \Eventually p_i) \implies \Always (q \implies \Eventually_I r))$.

\noindent \textit{Group 3.} This group includes 17 benchmarks drawn from the Acacia
synthesis tool's benchmark suite~\cite{Acacia}, which are also featured in the SyntComp
competition suite~\cite{Syntcomp-Jacobs2015}, with added timing intervals.  These include
formulae such as
\[
\begin{aligned}
\Acacia_1 &= \left(\Eventually_{[0,2]} \Always p\right) \iff \left(\Always \Eventually q\right), \\
\Acacia_2 &= \left(\Always\left(\Eventually_{[0,2]} \left(p \implies \Next_{[1,2]} \Next_{[1,2]} \Next_{[1,2]} q\right)\right)\right) \iff \left(\Always \Eventually r\right), \\
\Acacia_3 &= \left(((p \U_{[0,2]} q) \U_{[0,2]} r) \wedge ((q \U_{[0,2]} r) \U_{[0,2]} p) \wedge ((r \U_{[0,2]} p) \U_{[0,2]} q)\right) \iff \Always \Eventually s, \\
\Acacia_4 &= \left(\Always(p \implies \Eventually_{[3,4]} q) \wedge \Always(\neg p \implies \Eventually_{[3,4]} \neg q)\right) \iff \Always \Eventually r.
\end{aligned}
\]
The complete list of all 72 benchmark formulae are given in Tables~A and~B of Appendix~\ref{sec:app:experiments}.

\subsection*{Our results}
We present the results of our experiments in
\Cref{table:past,table:mitl-gta-inf-exp-main-results,table:model-checking}, which show the
number of nodes stored and the time taken by each tool to check satisfiability of the
formulae over finite and infinite words, respectively.  We also present a log-scale plot
of the time taken by \MightyL\ and \tool\ in \Cref{fig:plot}, which shows the performance
of our tool on the various benchmarks.  We discuss our experimental results in detail
below.  In all the tables, we highlight in green the rows corresponding to formulae
belonging to the $\detMTL$ fragment.

\begin{table}[!bp]
  \centering
  \resizebox{0.7\textwidth}{!}{
    \begin{tabular}{|l|r|r|}
      \hline
      \multicolumn{1}{|c|}{Formula}
      & \multicolumn{2}{c|}{\tool}\\
      \cline{2-3}
      & Stored nodes & Time in ms.\\
    \hline
\rowcolor{green!20}
    $\Eventually_{[0,20]} (\Y_{[2,3]} p_1 \vee \Y_{[4,5]} p_2) \vee \Y_{[6,7]} p_3)$ & 287 & 1.5\\ 
    \hline
\rowcolor{green!20}
    $\Always_{[0,20]} (\Y_{[2,3]} p_1 \vee \Y_{[4,5]} p_2) \vee \Y_{[6,7]} p_3)$ &  262 & 3.1\\ 
    \hline
\rowcolor{green!20}
    $\Eventually (p \S_{[1, 2]} (p \S_{[1, 2]} (p \S_{[1, 2]} q)))$ &  1,303 & 29.9\\ 
    \hline
\rowcolor{green!20}
    $\Eventually (p \S_{[1, \infty)} (p \S_{[1, \infty)} (p \S_{[1, \infty)} q)))$ & 66 & 0.5\\ 
    \hline
\rowcolor{green!20}
    $\Eventually (p \S_{[1, 2]} q \wedge  p \S_{[2, 3]} q \wedge p \S_{[3, 4]} q \wedge p \S_{[4, 5]} q)$ &  1,571 & 50.9\\ 
    \hline
\rowcolor{green!20}
    $\Eventually (p \S_{[1, \infty)} q \wedge  p \S_{[2, \infty)} q \wedge p \S_{[3, \infty)} q \wedge p \S_{[4, \infty)} q)$ & 68 & 0.8\\ 
    \hline
\rowcolor{green!20}
    $\Eventually (p \S_{[1, 2]} q \wedge  p \S_{[2, 3]} q \wedge p \S_{[3, 4]} q \wedge p \S_{[4, 5]} q \wedge p \S_{[5, 6]} q)$ & 27,737 & 1,762.1\\ 
    \hline
\rowcolor{green!20}
    $\Eventually (p \S_{[1, 2]} q \wedge  p \S_{[2, 3]} q \wedge p \S_{[3, 4]} q \wedge p \S_{[4, 5]} q \wedge p \S_{[5, 6]} q \wedge p \S_{[6, 7]} q)$ & 167,077 & 21,490.5\\ 
    \hline
\rowcolor{green!20}
    $\Always(\neg p_1 \vee (a_1 \S_{[0, 2]} b_1)) \wedge (\neg p_2 \vee (a_2 \S_{[0, 2]} b_2)) \wedge$ & &  \\
\rowcolor{green!20}
    $(\neg p_3 \vee (a_3 \S_{[0, 2]} b_3)) \wedge (\neg p_4 \vee (a_4 \S_{[0, 2]} b_4)) \wedge \Always(p_1 \vee p_2 \vee p_3 \vee p_4)$ & 186 & 1.67\\ 
    \hline
  \end{tabular}
  }
  \caption{Experimental results obtained by running our prototype implementation \tool\ for satisfiability of $\detMTL$ formulae over finite timed words. Timeout is set to 300 seconds.
  }
  \label{table:past}
\end{table}

\medskip\noindent {\bf Satisfiability for the past fragment of $\detMTL$.}
The first set of experiments, presented in \Cref{table:past}, evaluates the performance of
our tool \tool\ on $\detMTL$ formulae involving past operators.  We find that \tool\
successfully handles all benchmarks in this category.  We note that \MightyL\ could not be
run on these, as it does not support past modalities (listed as part of future work
in~\cite{MightyL}).
The last but two formulae are specifically designed to stress-test the handling of
multiple nested past operators.  \tool\ processes each of them efficiently, storing only a
small number of nodes and completing each run in under 30 seconds.  The final formula in
this set is more complex, combining several past operators with a conjunction of future
modalities.  Here too, \tool\ performs efficiently, requiring only 186 nodes and less than
2 ms to check satisfiability.
These results demonstrate the effectiveness of our deterministic translation, which avoids
the exponential blowup typically associated with non-deterministic approaches and enables
scalable analysis even for complicated formulae.
To see this, consider, for instance, the formula $\Eventually (p \U_{[1, 2]} q \wedge p
\U_{[2, 3]} q \wedge p \U_{[3, 4]} q \wedge p \U_{[4, 5]} q)$, obtained by replacing the
past operator $\S$ with the future operator $\U$ in the $5^{th}$ entry \tool\ solves this
by generating 28,030 nodes in 784 ms.  Another example is the nested formula $\Eventually
(p \U_{[1, 2]} (p \U_{[1, 2]} (p \U_{[1, 2]} q)))$ which results in 299,130 nodes and
takes 9 seconds, compared to the past version $\Eventually (p \S_{[1, 2]} (p \S_{[1, 2]}
(p \S_{[1, 2]} q)))$ which generates only 1,303 nodes in 29.9 milliseconds.

\medskip\noindent {\bf Satisfiability for full $\MTLfp$ and comparisons.}
For the second set of experiments, we evaluate the performance of \tool\ on the
satisfiability problem for $\MTLfp$ formulae over both finite and infinite timed words.
We compare our tool against \MightyL\ and the baseline implementation of the translation
proposed in~\cite{AGGS-CONCUR24}.  For finite timed words, we compare against \MightyL\
using $\UPPAAL$'s reachability analysis; for infinite timed words, we use \MightyL\ with
\Opaal's liveness checking algorithm.

\begin{figure}[tbp]
  \centering
  \begin{minipage}{0.53\textwidth}
      \centering
      \includegraphics[scale=0.15]{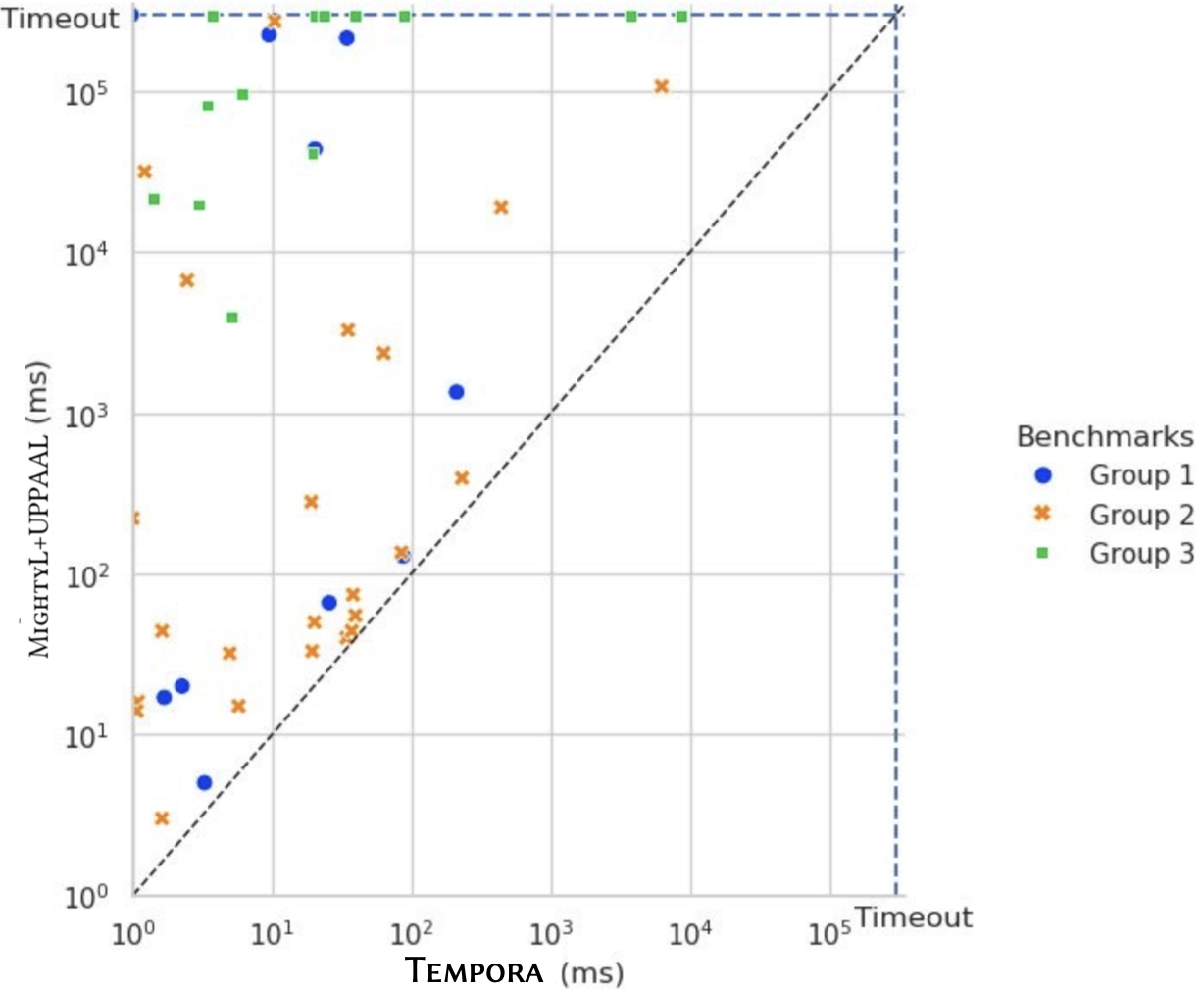}  
  \end{minipage}
  \hspace{0.1in}
  \begin{minipage}{0.4\textwidth}
      \centering
      \includegraphics[scale=0.15]{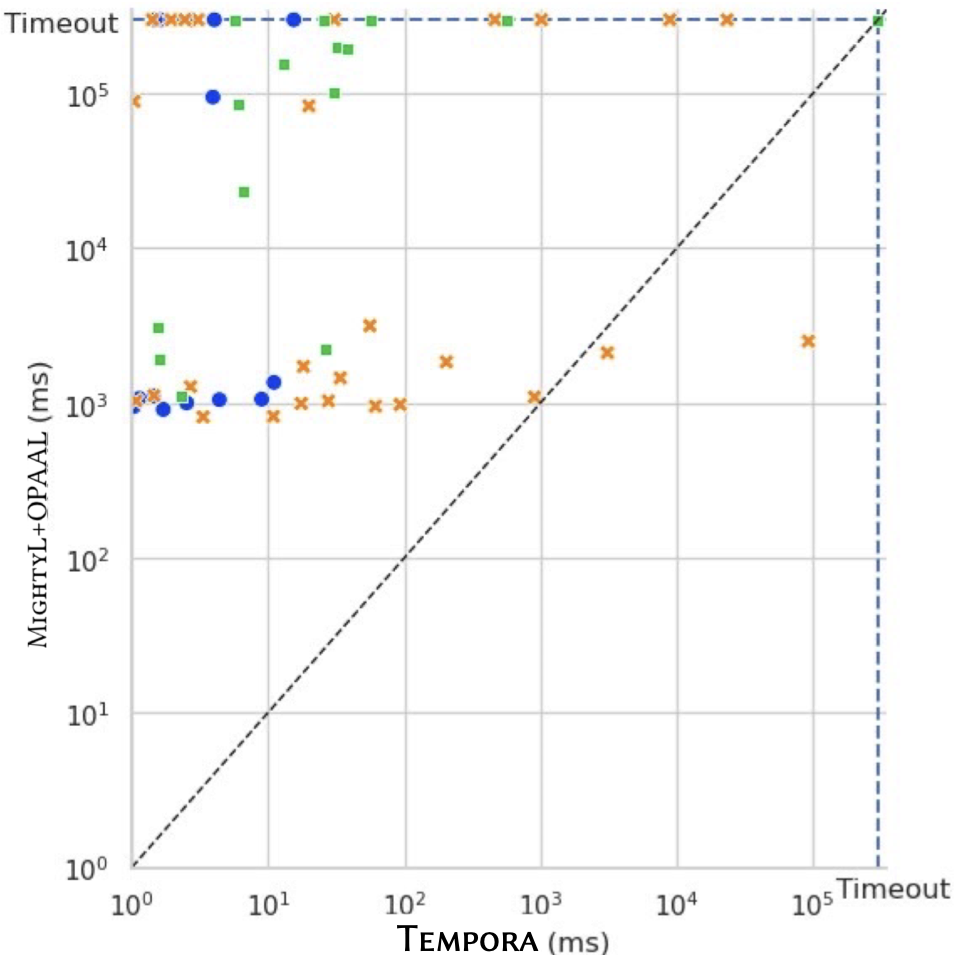}
  \end{minipage}
  \caption{Log-scale plot $\MightyL$ vs $\tool$ comparison for time taken to check
  satisfiability on finite (left) and infinite (right) timed words, after the translation.
  Group 1 refers to our benchmarks, while Group 2 are \MightyL\ benchmarks~\cite{MightyL}
  are Group 3 benchmarks derived from Acacia~\cite{Acacia}.  Timeout set to $300s$.}
  \label{fig:plot}
\end{figure}
\Cref{fig:plot}(left) compares the performance of our tool and the $\MightyL$ pipeline on
all 73 benchmarks for finite timed words, and \Cref{fig:plot}(right) presents the same
comparison for infinite words.  The time reported is the time after the translation (to
$\TA$ in case of $\MightyL$ and $\GTA$ for $\tool$), that is, the time for zone graph
exploration to find a reachable node (for finite timed words) or a live cycle (for
infinite timed words).  This is because, the time for the translation itself is rather
small in most cases (of the order of 3-4ms), except for the cases shown in
\Cref{table:translation-time} in Appendix~\ref{sec:app:experiments}.  For $\Opaal$, we
report the sum of time taken by \Opaal\ to generate C++ code, the \texttt{g++} compiler to
compile the generated code and \LTSmin\ to check liveness.  We note that \Opaal\ is a
multi-core tool and uses 8 cores on our machine.  On 10 of the benchmarks, $\Opaal$ gave a
segmentation fault (which we mark as timeout) and $\UPPAAL$ reported an overflow error on
the output generated by $\MightyL$ for 2 benchmarks.

\Cref{table:mitl-gta-inf-exp-main-results} presents selected results reporting the number
of nodes stored and time taken by each tool to check satisfiability over infinite words,
while the full set of results appears in Appendix~\ref{sec:app:experiments}

\begin{table}[!tbp]
  \centering
  \resizebox{0.73\textwidth}{!}{
    \begin{tabular}{|l|r|r|r|c|r|r|}
      \hline
      \multicolumn{1}{|c|}{Formula}
      & \multicolumn{2}{c}{Baseline}
      & \multicolumn{2}{|c|}{\MightyL + OPAAL}
      & \multicolumn{2}{c|}{\tool}\\
      \cline{2-7}
      & Stored & Time in
      & Stored & Time in
      & Stored & Time in\\
      & nodes & ms.
      & nodes & ms.
      & nodes & ms.\\
    \hline
    \rowcolor{green!20} 
    $\Eventually(5,[2,\infty))$ & - & TO & 3069 & 373 + 820 + 85  & 200 & 2.38\\ 
    \hline
    \rowcolor{green!20} 
    $\Eventually(5,[0,2])$ & 952 & 5.99 & 2693 & 214 + 702 + 82 & 68 & 0.49\\ 
    \hline
    $\U(5,[2,\infty))$ & 567,192 & 5,404 & 9,383 & 382 + 960 + 114 & 4,075 & 32.70\\ 
    \hline
    $\U(5,[1,2])$ & - & TO & - & TO & 465,185 & 8,750\\ 
    \hline
    \rowcolor{green!20} 
    $\mu(2)$ & 953,492 & 24,977 & 964 & 68,462 + 20,584+ 103 & 158 & 0.90\\
    \hline
    \rowcolor{green!20} 
    $\mu(3)$ & - & TO & - & TO & 315 & 3.34\\
    \hline
    $\theta(3)$ & - & TO & 42,804 & 836 + 1,057 + 225 & 330,044 & 2977\\
    \hline
    \rowcolor{green!20} 
    $\Always(5,[0,2])$ & 4,793 & 34.6 & 5,913 & 257 + 760+ 101 & 224 & 1.45\\ 
    \hline
    \rowcolor{green!20} 
    $\Always(5,[1,2])$ & - & TO & - & TO & 225 & 2.45\\ 
    \hline
    $\Release(5,[2,\infty))$ & 29,685 & 278.3 & 1,059 & 231 + 726 + 77 & 3,119 & 27.75\\ 
    \hline
    $\Release(5,[1,2])$ & - & TO & -	 & TO & 1,093,187 &	23,474\\ 
    \hline
    \hline
    \rowcolor{green!20} 
    $p \U_{[11, 12]} q$ & 1,484 & 53 & - & TO & 220 & 1.37\\
    \hline
    $\Reqgrant$ & - & TO & - & TO & 444 & 4.00\\
    \hline
    $\Spotex$ & 117,319 & 1,141 & 7,921 & 74,585 + 20,184 + 238 & 504	& 3.50\\
    \hline
    \rowcolor{green!20} 
    $\eta(4)$ & - & TO & - & NA & 249 & 1.49\\
    \hline
    $\rho(4)$ & - & TO & - & TO & 1933 & 15.93\\
    \hline
    \hline
    $\Acacia_1$ & - & TO & 4,943 & 258 + 755 + 96 & 390 & 2.54\\ 
    \hline
    $\Acacia_2$ & - & TO & 742,427 & 426 + 949 + 1,683 & 231 & 1.94\\
    \hline
    $\Acacia_3$ & - & TO & 7,271,276 & 611 + 1,143 + 21,349 & 918 & 7.41\\ 
    \hline
    $\Acacia_4$ & - & TO & - & TO & 4,159 & 55.71\\
    \hline    
  \end{tabular}
  }
  \caption{Selected set of experimental results obtained by running our prototype
  implementation \tool, \MightyL\ and, the baseline implementation of the translation
  proposed in~\cite{AGGS-CONCUR24} for satisfiability over infinite traces.  Timeout is
  set to 300 seconds.
  The complete table is available in~Table 8 in Appendix~\ref{sec:app:experiments}.
  }
  \label{table:mitl-gta-inf-exp-main-results}
\end{table}

As can be seen from \Cref{fig:plot}, in the finite timed words case, $\tool$ is always
faster than $\MightyL+\UPPAAL$ and for benchmarks where both tools take more than 1s, we
are significantly faster, as the plot is in {\em log-scale}.  Further, $\UPPAAL$ times out
in several cases, while we never do so on the benchmarks tested.  Thus, we scale better on
larger and more complex formulae.  For the infinite case, except for 2 formulae we are
always better.  Note that the plot is only on the 63 benchmarks on which all three tools
work, and therefore, we exclude the formulae containing past operators.
Our performance is significantly better in terms of time-taken as can be seen from the
figure.  Further, out of the 72 benchmarks, $\Opaal$ timed out or seg-faulted on 18
benchmarks whereas we timeout on only 4 benchmarks (of which $\Opaal$ also times
out/segfaults on 3).
Finally, \Cref{table:mitl-gta-inf-exp-main-results} (and the extended tables in
Appendix~\ref{sec:app:experiments}) show that $\tool$ is also significantly better than
the Baseline (in fact the \MightyL\ pipeline also outperforms Baseline significantly on
most benchmarks) on all benchmarks.

{\em In summary, both for finite words and infinite words over the pointwise semantics,
$\tool$ significantly outperforms $\MightyL$ and the baseline for checking satisfiability,
both in terms of time taken and size of stored nodes.}

\medskip\noindent {\bf Model checking.}
Finally, going beyond satisfiability, we evaluate our model-checking pipeline on a set of
$\MTLfp$ formulas and two models: Fischer's mutual exclusion algorithm and the Dining
Philosophers problem.  The model-checking pipeline is implemented in \tool\ and uses
\Tchecker\ for the zone graph exploration.  To the best of our knowledge, there is no
existing tool among \UPPAAL, \MightyL, and \Opaal\ that can perform complete model
checking for pointwise semantics with a model and a $\MTLfp$ formula.  As a result, we are
unable to provide a comparison for our model-checking pipeline.  Nevertheless, as a proof
of concept, we evaluate our pipeline on two models (Fischer and Dining Philosophers)
against a set of formulas, with results in \Cref{table:model-checking}.

For the Fischer's mutual exclusion algorithm, we consider a model with 3 to 7 processes
and check various properties.  The first property checks whether the first process can
eventually enter the waiting state after making a request, while the second property
checks whether it can enter the critical section.
We also check the classical safety property for mutual exclusion protocols that ensures
that no two processes can be in the critical section at the same time.
Our experimental results show that for Fischer with 6 processes, the property $\Always
(\req_1 \rightarrow \Eventually_{[0, 5]} \wait_1)$ can be checked to hold in 4.6 ms, while
for Fischer with 7 processes, it takes 5.3 ms.
We also check the property that no two processes can be in the critical section at the
same time, which is reported to hold true for Fischer with 5 processes in 15.51 seconds.

For the Dining Philosophers problem, we consider models with 2 and 8 philosophers.  We
check properties such as whether at least one philosopher is eating at some point in time,
and whether both philosophers are guaranteed to eat within a certain window of time.
For the model with 8 philosophers, we find that the property $\Eventually_{[0, 20]}
(\bigvee_{i=1}^{8} \eat_{i})$ does not hold true within 0.01 seconds, indicating that 
no philosopher is guaranteed to eat within 20 time units.
For the model with 2 philosophers, we check whether
both philosophers are guaranteed to eat within a certain time window, which does not hold
true within 0.03 seconds.
The results indicate that our model-checking pipeline can efficiently handle the
verification of $\MTLfp$ properties on these models, demonstrating its potential for
practical applications in verifying timed systems.

\begin{table}[!tbp]
  \centering
  \resizebox{0.8\textwidth}{!}{
    \begin{tabular}{|l|l|l|r|r|}
      \hline
      Sl. 
      & \multicolumn{1}{c|}{Models}
      & \multicolumn{1}{c|}{Formula}
      & \multicolumn{2}{|c|}{\tool + \Tchecker}\\
      \cline{4-5}
      No. 
      & &
      & Stored nodes & Time (ms)\\
      \hline
      1 & Fischer (3) & $\Always (\neg \req_1 \vee \Eventually_{[0, 20]} \wait_1)$ & 32,817	& 221.3\\
      \hline
      2 & Fischer (4) & $\Always (\neg \req_1 \vee \Eventually_{[0, 20]} \wait_1)$ & 449,108	& 5,723\\
      \hline
      3 & Fischer (6) & $\Always (\neg \req_1 \vee \Eventually_{[0, 5]} \wait_1)$ & 362	& 4.6\\
      \hline
      4 & Fischer (7) & $\Always (\neg \req_1 \vee \Eventually_{[0, 5]} \wait_1)$ & 391	& 5.3\\
      \hline
      5 & Fischer (3) & $\Always (\neg \req_1 \vee \Eventually_{[0, 20]} \cs)$ & 33,918 & 209.69\\
      \hline
      6 & Fischer (4) & $\Always (\neg \req_1 \vee \Eventually_{[0, 20]} \cs)$ & 352,774	& 3,749.6\\
      \hline
      \rowcolor{green!20} 
      7 & Fischer (5) & $\Always (\neg((\cs_1 \wedge (\cs_2 \vee \cs_3 \vee \cs_4 \vee \cs_5)) \vee$ & 	& \\
      \rowcolor{green!20} 
      &  & $ (\cs_2 \wedge (\cs_1 \vee \cs_3 \vee \cs_4 \vee \cs_5)) \vee (\cs_3 \wedge (\cs_1 \vee \cs_2 \vee \cs_4 \vee \cs_5))  \vee $ &  & \\ 
      \rowcolor{green!20} 
      &  & $(\cs_4 \wedge (\cs_1 \vee \cs_2 \vee \cs_3 \vee \cs_5)) \vee (\cs_5 \wedge (\cs_1 \vee \cs_2 \vee \cs_3 \vee \cs_4))))$ & 527,829	& 15,514\\
      \hline
      \rowcolor{green!20} 
      8 & Din-Phil(8, 10, 3, 1) & 
      $\Eventually_{[0, 20]} (\eat_1 \vee \eat_2 \vee \eat_3 \vee \eat_4$ 
      & & \\
      \rowcolor{green!20} 
      & & $\vee \eat_5 \vee \eat_6 \vee \eat_7 \vee \eat_8)$
      & 371	& 5.8\\
      \hline
      9 & Din-Phil(2, 10, 3, 1) & 
      $\Always_{[10, \infty)} (\Eventually_{[0, 30]}(\eat_1) \wedge \Eventually_{[0, 30]}(\eat_2))$ & 4,904 & 31.35\\
      \hline
  \end{tabular}  
  }
  \caption{Experimental results obtained by running our prototype implementation \tool,
  and the \GTA\ liveness algorithm implementation in \Tchecker\ for checking liveness.
  The numbers in the brackets represent parameterization, for Fischer's it is \#-processes
  and for Dining Philosophers, it represents the number of philosophers, timeout before
  releasing fork, time to eat and delay time to release fork.
}
\label{table:model-checking}
\end{table}

Additional results, including ablation studies, are presented in \Cref{sec:app:ablation}
of Appendix~\ref{sec:app:experiments}.

\subsubsection*{Summary of Experimental Evaluation.}
\label{sec:exp-summary}
In summary, our experiments demonstrate the practicality and effectiveness of \tool\ in various settings. 
\tool\ handles all standard $\MITL$ benchmarks used in prior works, and crucially extends
support to the richer $\MTLfp$ fragment - including past operators and extending timed
modalities to allow punctual constraints at the top level.  This is beyond the
capabilities of state-of-the-art tools for pointwise $\mitl$ logics such as \MightyL,
which do not support past operators or the punctual timed operators.  Further, on
benchmarks where $\MightyL$ does work, $\tool$ outperforms it significantly both in the
size of zone graph computed and time taken to check satisfiability.
Finally, \tool\ can be effectively used for model checking full $\MTLfp$ specifications,
including both safety and liveness properties on standard benchmarks such as Fischer's
mutual exclusion protocol and the Dining Philosophers.

\section{Conclusion}\label{sec:conclusion}

In this paper, we developed a new synchronization-based model checking technique that
allows translating formulas in $\MITL$ with past modalities (and punctuality at the
outermost level) to networks of small (generalized) timed automata with the general theme
of making it ``more deterministic''.  While we get deterministic automata for a large
fragment, we show improvements even in the general case, especially a surprising 2-state
until automaton that takes advantage of future clocks.
All the technical developments have resulted in the tool \tool, which outperforms the
current state-of-the-art tools.  As future work, two natural directions would be towards
monitoring of $\MTLfp$ properties in the pointwise semantics, and to try to extend our
approach to the signal/continuous semantics.

\bibliography{mitl2gta.bib}

\clearpage

\appendix

\section{Appendix for missing proofs}~\label{sec:app:translation}
 
In this section, we provide the missing proofs omitted in \Cref{sec:fastmtl-to-ta} and \Cref{sec:future}.

\subsection{Missing proofs from \Cref{sec:fastmtl-to-ta}}

\lemisatp*

\begin{proof}
  In the forward direction, suppose $w \models p$.  Then, we know that $p \in a_0$, i.e.,
  the atomic proposition $p$ holds at the first position of the word.
  By construction, $\Aisatp{p}$ (\Cref{fig:isat-prop}), the automaton starts in state $1$
  and checks whether $p$ holds at the first position.  If $p \in a_0$, then $\Aisatp{p}$
  takes the transition to state $2$, and from then on, continues to remain in state $2$ on
  all subsequent inputs.
  Hence, configuration $C_i$ satisfies $\Aisatp{p}.\state = 2$ for $i \ge 1$.

  In the other direction, suppose that for all $i \ge 1$, the configuration $C_i$
  satisfies $\Aisatp{p}.\state = 2$.  Since the only way for the automaton to reach state
  $2$ from the initial state $1$ is via a transition that checks $p \in a_0$, it must be
  the case that $p \in a_0$.
  Therefore, $w \models p$.
\end{proof}

\lemisatX*

\begin{proof}
  Suppose that $w \models \X_I \argone$. 
  Then, this means that the timed word $w$ satisfies (1) $\argone\in a_1$, and (2) 
  $\delta_1=\tau_1 - \tau_0 \in I$. 
  Consider the run of the automaton $(\Ainit, \AisatX)$ over $w$.  The automaton starts in
  state $1$ at configuration $C_0$.  From state $1$, the automaton moves to state $2$
  after reading the first letter $(a_0, \tau_0)$, resetting clock $\xinit$.  Then, from
  state $2$, since $\argone$ holds at the next letter $a_1$ and the clock value of
  $\xinit$ is $\delta_1=\tau_1 - \tau_0 \in I$, the automaton moves to state $3$.
  Since state $3$ is a sink-state with a self-loop, $\AisatX$ remains in state $3$ for all $i \ge 2$.
  Therefore, for all $i \ge 2$, the configuration $C_i$ satisfies $\AisatX.\state = 3$.
  
  Conversely, suppose that for all $i \ge 2$, the configuration $C_i$ satisfies $\AisatX.\state = 3$.
  By construction of $\AisatX$, state $3$ can only be reached from 
  state $2$ if $(a_1, \tau_1)$ satisfies the following:  
  $\argone \in a_1$ and $\xinit \in I$.
  Further, since $\xinit$ was reset on reading $(a_0, \tau_0)$, this implies that $\tau_1 - \tau_0 \in I$.
  Therefore, $w \models \X_I \argone$.
\end{proof}

\lemYesterday*

  \begin{proof}
    Since the network is deterministic and complete, there is a unique run of the network
    $(\Alast,\AY)$ on the given input word $w$.

    Note that the network $(\Alast, \AY)$ maintains two components:
    \begin{itemize}
      \item $\Alast$ stores the time elapsed since the previous position using the clock $\xlast$.

      \item $\AY$ allows to check whether the formula $\Y \argone$ holds by storing the
      value of $\argone$ at the previous position.
    \end{itemize}

    By construction, the state of $\AY$ is updated as follows. 
    If $\argone$ was true at $a_{i-1}$, then after processing $(a_{i-1}, \tau_{i-1})$,
    $\AY$ enters state $\ell_1$ and remains there at both $C_i$ and $C'_i$.
    Similarly, if $\argone$ was not true at $a_{i-1}$, then after processing $(a_{i-1},
    \tau_{i-1})$, $\AY$ enters state $\ell_0$ and remains there at both $C_i$ and $C'_i$.

    Similarly, $\Alast$ resets $\xlast$ upon reading each letter of the input word.
    Hence, at $C'_i$, we have $\Alast.\xlast = \delta_i = \tau_i - \tau_{i-1}$.
    Therefore, for all $i\ge 1$, $\tau_{i}-\tau_{i-1}\in I$ iff $C'_{i}(\Alast.\xlast)\in I$.

    The formula $\Y_I \argone$ is true at $i$ iff
    \begin{itemize}
      \item $\argone$ held at position $i-1$, which translates to $\cA_{\Y}$ being in
      state $\ell_1$ at configuration $C'_i$,

      \item $\tau_i - \tau_{i-1} \in I$ (i.e., $\Alast.\xlast \in I$ at $C'_i$).
    \end{itemize}
  Therefore, the test $\AY.\state = \ell_1 \wedge \Alast.\xlast \in I$ precisely checks
  $(w, i) \models \Y_I \argone$ during the transition $\overline{t}_i$ from $C'_i$ to
  $C_{i+1}$.
  \end{proof}

\lemsinceub*

\begin{proof} ~
  \begin{enumerate}
    \item 
    From the semantics of the untimed $\Since$ operator, $(w,i) \models \argone \S
    \argtwo$ if there exists $0 \leq j \leq i$ such that $\argtwo\in a_j$, and $\argone\in
    a_k$ for all $j < k \leq i$.

    From \Cref{fig:Since-1-sided-automata}, it can be seen that $\AS$ moves to state
    $\ell_1$ upon seeing $\argtwo$, and stays in $\ell_1$ as long as $\argone\vee\argtwo$
    holds.  $\AS$ returns to $\ell_0$ if it sees a position for which
    $\argone\vee\argtwo$ does not hold.

    Thus, if $C_{i+1}(\AS.\state) = \ell_1$, there exists some $j \le i$ where $\argtwo
    \in a_j$ and if we take the largest such $j$, at all positions in $(j,i]$,
    $\argone$ holds, which implies that $(w,i) \models \argone\S\argtwo$.
    
    Conversely, suppose that $(w,i) \models \argone\S\argtwo$. 
    This means, that there was some position $j \le i$ such that $(w,j) \models \argtwo$, and at all positions between $j$ and $i$, $\argone$ holds.
    Then, from the construction, $\AS$ must be in state $\ell_1$ at $C_{i+1}$.

    \item 
    $\ASlast.x$ is reset whenever $\argtwo$ is observed, and is not updated until the next
    occurrence of $\argtwo$.  Thus, $\ASlast$ stores the total time elapsed since the most
    recent occurrence of $\argtwo$.  We have already seen from the previous item that if
    $C_{i+1}(\AS.\state)=\ell_1$, then there exists some $j \le i$ where $\argtwo\in a_j$.
    Therefore, when $C_{i+1}(\AS.\state) = \ell_1$, the last observed $\argtwo$ was at
    some position $j \le i$ with no $\argtwo$ between $j$ and $i$, then
    $C_{i+1}(\ASlast.x)=\tau_{i}-\tau_{j}$.

    \item 
    Recall that $\argone \S_I \argtwo$ holds at $(w,i)$ if and only if the following two
    conditions are satisfied
    \begin{itemize}
      \item the untimed condition for $\argone \S \argtwo$ is true, which from (i) we know
      is the case if the automaton $\AS$ is in state $\ell_1$ at configuration $C_{i+1}$,
      and

      \item the timed constraint is satisfied, i.e., $\tau_i - \tau_j \in I$, where $j$ is
      the most recent position where $\argtwo$ was true.  This holds because $I$ is an
      upper-bounded interval.  From (2), we know that when $C_{i+1}(\AS.\state) =
      \ell_1$, we have $C_{i+1}(\ASlast.x)=\tau_{i}-\tau_{j}$.
    \end{itemize}

    From the above, it is easy to see that the test
    $\outv_{\S_{I}}=(\post{\AS}\state=\ell_1 \wedge \post\ASlast{x}\in I)$ is satisfied
    during transition $\overline{t}_{i}$ from $C'_{i}$ to $C_{i+1}$ exactly when $(w,i)
    \models \argone \S_I \argtwo$.  
    \qedhere
  \end{enumerate}
\end{proof}

\lemsincelb*

\begin{proof} ~
  \begin{enumerate}
    \item 
    This is already shown in the proof of \Cref{lem:since-upper-bound} given above.
  
    \item First, notice that if $C_{i+1}(\AS.\state)=\ell_1$, the automaton $\AS$ moved
    from $\ell_{0}$ to $\ell_{1}$ in the past, i.e., with some $\overline{t}_{i'}$ with
    $i'\leq i$.  During this transition $\overline{t}_{i'}$, the automaton $\ASfirst$
    takes the loop which resets $y$.  Hence, $y$ was reset at least once in the past of
    $i$.  Let $j\leq i$ be the position where $y$ was last reset by $\ASfirst$.  We have
    $\argtwo\in a_j$ and either $\argone\notin a_{j}$ or
    $C_j(\pre\AS\state)=C'_j(\pre\AS\state)=\ell_{0}$.
      
    Now we focus on the two parts of the statement that we need to show.
    
    \begin{itemize}
      \item We already saw that $\argtwo \in a_j$.  Next, we show that $\argone\in
      a_k$ for all $j < k \leq i$. Assume that $\argone\notin a_{k}$ for some $j<k\leq i$.
      If $\argtwo\in a_{k}$ then $\ASfirst$ takes the loop which resets $y$ during 
      transition $\overline{t}_{k}$, a contradiction with the maximality of $j$.
      Hence, $\argtwo\notin a_{k}$ and transition $\overline{t}_{k}$ takes $\AS$ to state 
      $\ell_0$. So we must have some $k<i'\leq i$ such that $\overline{t}_{i'}$ takes 
      $\AS$ from $\ell_0$ to $\ell_1$. As explained above, $\ASfirst$ must take the loop 
      which resets $y$ during transition $\overline{t}_{i'}$, again a contradiction with 
      the maximality of $j$. Therefore, $\argone\in a_k$ for all $j < k \leq i$.
            
      \item For the second part, suppose towards a contradiction that there was a position
      $0\leq j'<j$ such that $\argtwo\in a_{j'}$ and $\argone\in a_{k'}$ for all
      $j'<k'\leq i$.  Since $\argtwo\in a_{j'}$ and $\argone\in a_{k'}$ for all $j'<k'
      \leq i$, automaton $\AS$ must be in state $\ell_1$ from $C_{j'+1}$ to $C_{i+1}$.
      Since $j'<j\leq i$, automaton $\AS$ is in state $\ell_1$ at configurations $C_{j}$
      and $C'_{j}$.  Further, we know that $y$ was reset by $\ASfirst$ during transition
      $\overline{t}_{j}$.  Since $C'_{j}(\pre\AS\state)=\ell_{1}$, this is only possible 
      if $\argone\not\in a_j$, a contradiction.
      Therefore, for all $0\leq j'<j$ either $\argtwo\notin a_{j'}$ or $\argone\notin
      a_{k'}$ for some $j'<k'\leq i$.
    \end{itemize}

    \item The proof proceeds as in the proof of \Cref{lem:since-upper-bound}, as we have
    shown the untimed and timed aspects separately above.  For the sake of completeness,
    we give the arguments below:
    $\argone \S_I \argtwo$ holds at $(w,i)$ if and only if the following two conditions
    are satisfied
    \begin{itemize}
      \item the untimed aspect follows from the fact that automaton $\AS$ is in state
      $\ell_1$ at configuration $C_{i+1}$ (shown in the first part), and
      
      \item the timed constraint is satisfied, i.e., $\tau_i - \tau_j \in I$, where $j$
      is the earliest position where $\argtwo$ is seen such that this $\argtwo$ can act
      as a witness for $\argone \S \argtwo$: This holds because $I$ is a lower-bounded
      interval.  From (ii), we know that when $C_{i+1}(\AS.\state) = \ell_1$, we have
      $C_{i+1}(\ASfirst.y)=\tau_{i}-\tau_{j}$.
    \end{itemize}
    
    From the above, it is easy to see that the test
    $\outv_{\S_{I}}=(\post{\AS}\state=\ell_1 \wedge \post\ASfirst{y}\in I)$ is satisfied
    during transition $\overline{t}_{i}$ from $C'_{i}$ to $C_{i+1}$ exactly when $(w,i)
    \models \argone \S_I \argtwo$.  \qedhere
  \end{enumerate}
\end{proof}

\lemsincegeneral*

The proof of \Cref{lem:since-general} will be given after the following technical lemma which establishes all the necessary invariants of the automaton $\ASgenI$.

\begin{lemma}\label{lem:since-general-invariants}
  Let $I$ be a non-singleton interval with lower and upper bounds $b,c\in\mathbb{N}$ and 
  $0\notin I$.  Let ${\leqlt}={\leq}$ if $b\in I$ and ${\leqlt}={<}$ otherwise.
  Let $\Prop' = \Prop \uplus \{\argone, \argtwo\}$.  For every word $w =
  (a_0, \tau_0) (a_1, \tau_1) \cdots$ over $\Prop'$ there is a unique run $C_0
  \xra{\delta_0} C'_0 \xra{\overline{t}_0} C_1 \xra{\delta_1} C'_1 \xra{\overline{t}_1}
  C_2 \cdots$ of $(\AS,\ASgenI)$.
  Moreover, for all $i \ge 0$ we have:
  \begin{enumerate}
    \item $(w,i) \models \argone\S\argtwo$ iff $C_{i+1}(\AS.\state)=\ell_1$.
    
    \item When $C_{i+1}(\AS.\state)=\ell_0$ then configuration 
    $C_{i+1}$ satisfies $x_{1}=y_{1}=\cdots=x_k=y_k=+\infty$.
    
    \item When $C_{i+1}(\AS.\state)=\ell_1$ then for some $1\leq j\leq k$, configuration 
    $C_{i+1}$ satisfies the following:
    \begin{itemize}
      \item  $y_j\leq x_j< y_{j-1}\leq x_{j-1} < \cdots < y_{1}\leq x_{1} <+\infty$, 
      $x_{j+1}=y_{j+1}=\cdots=x_k=y_k=+\infty$,
    
      \item  $x_{j'}-y_{j'}<c-b$ for all $1\leq j'\leq j$, and
      $x_{j'}-x_{j'+1}\geq c-b$ for all $1\leq j'<j$, 
      
      \item there are indices $i_1\leq i'_1<i_2\leq i'_2<\cdots<i_j\leq i'_j\leq i$ such 
      that 
      \begin{itemize}
        \item $\argone\in a_{i'}$ for all $i'\in(i_{1},i]$,
      
        \item $\argtwo\in a_{i'}$ for all $i'\in\{i_1,i'_1,\ldots,i_j,i'_j\}$,
      
        \item $\argtwo\notin a_{i'}$ for all 
        $i'\in(i'_1,i_2)\cup\cdots\cup(i'_{j-1},i_j)\cup(i'_j,i]$,
        
        \item $x_{1}=\tau_i-\tau_{i_1}$, $y_{1}=\tau_i-\tau_{i'_1}$, \ldots,
        $x_{j}=\tau_i-\tau_{i_j}$, $y_{j}=\tau_i-\tau_{i'_j}$,
      \end{itemize}
    
      \item  $b\leqlt x_1$ or $(w,i_{1}) \not\models \argone \wedge \Y(\argone\S\argtwo)$.
      Moreover, if $j>1$ then $b\nleqlt x_2$.
    \end{itemize}
  \end{enumerate}
\end{lemma}

\begin{proof}
  1.\ was already stated in \Cref{lem:since-upper-bound}.
  2.\ follows directly from Lines (1--2) of \Cref{algo:Since-I-general}.
  
  3. The proof is by induction on $i$. The base case is $i=0$.
  The initial state of $\AS$ is $\ell_{0}$, i.e., $C_0(\pre\AS\state)=\ell_{0}$. 
  Now, $C_{1}(\pre\AS\state)=\ell_1$ iff $\argtwo\in a_0$.
  In this case, Line 4 of \Cref{algo:Since-I-general} is taken and we easily check that
  item (3) of \Cref{lem:since-general} holds with $j=1$ and $i_1=i'_1=0$.
    
  \medskip
  
  Consider now $i>0$ and assume that item (3) hold for $i-1$.
  If $C_i(\pre\AS\state)=\ell_{0}$ then we argue exactly as in the base case above.
  The interesting case is $C_i(\pre\AS\state)=\ell_{1}$. 
  By induction, there is $1\leq j\leq k$ and indices $i_1\leq i'_1<i_2\leq
  i'_2<\cdots<i_j\leq i'_j\leq i-1$ satisfying all the conditions of item (3) for $i-1$.

  If $\argone\notin a_i$ then Line 4 of \Cref{algo:Since-I-general} is taken.
  Notice that $\argtwo\in a_i$ since the selfloop on $\ell_1$ of $\AS$ is taken.
  We easily check that item (3) of \Cref{lem:since-general} holds with
  $j=1$ and $i_1=i'_1=i$.
    
  It remains to deal with the case $C_{i+1}(\pre\AS\state)=C_i(\pre\AS\state)=\ell_{1}$
  and $\argone\in a_i$.  

  \medskip
  
  Assume first that no clock shift occurs: we have $b\nleqlt
  C'_{i}(x_{j'})$ for all $1<j'\leq j$.
    
  \medskip
  
  If $\argtwo\in a_i$ then Lines (10--11) of \Cref{algo:Since-I-general} are executed.
    
  Recall that $C'_i$ satisfies $x_{j'}-x_{j'+1}\geq c-b$ for all $1\leq j'<j$.  
  Hence, if $C'_i$ satisfies $x_{j'}<c-b$ it must be with $j'=j$.
  In this case, Line 10 resets clock $y_{j}$ to $0$. We see easily that $C_{i+1}$ 
  satisfies Item 3 with the same $j$ and the sequence of indices
  $i_1\leq i'_1<i_2\leq i'_2<\cdots<i_j\leq i''_j=i$.
  Note that the last index in the sequence has been updated: $i''_j=i$.
  Since $x_j<c-b$ when $y_j$ is reset, $C_{i+1}$ satisfies $x_j-y_j<c-b$.
  
  Assume now that $C'_i$ satisfies $x_{j}\geq c-b$. 
  Assume towards a contradiction that $j=k$.
  Then, $C'_{i}$ satisfies $x_2-x_k\geq(k-2)(c-b)$.
  Since we had no clock shift, we also have $x_2\leq b$.
  Together with $c-b\leq x_{k}$ we get $(k-1)(c-b)\leq x_2\leq b$.
  This is a contradiction with $k=2+\Big\lfloor\dfrac{b}{c-b}\Big\rfloor$.
  Hence, $j<k$ and Line 11 will reset clocks $x_{j+1},y_{j+1}$ to $0$.
  We easily check that Item 3 is satisfied with $j+1$ and the sequence of indices
  $i_1\leq i'_1<i_2\leq i'_2<\cdots<i_j\leq i'_j\leq i_{j+1}=i'_{j+1}=i$.
  In particular, $C_{j+1}$ satisfies $x_j-x_{j+1}=x_j\geq c-b$.
  And also, if $j+1=2$, we have $b\nleqlt 0=x_{2}$ at $C_{i+1}$ since $0\notin I$.
  
  If $\argtwo\notin a_i$ then Line 13 of \Cref{algo:Since-I-general}
  is taken and we easily check that item (3) holds for $i$ with the same $j$ and 
  indices $i_1\leq i'_1<i_2\leq i'_2<\cdots<i_j\leq i'_j$. In particular, since clocks 
  are not updated, for all $1\leq j'\leq j$, we get
  $C_{i+1}(x_{j'})=C_i(x_{j'})+\tau_{i}-\tau_{i-1}=\tau_{i}-\tau_{i_{j'}}$ 
  and similarly for $y_{j'}$.

  \medskip
  
  Finally, assume that a clock shift occurs: the test of Line 6 evaluates to true with
  some $1<n\leq j$.  Let $j''=j-n+1$.
  Then, $1\leq j''<j$ and after the clock shift we have $x_{j''}\neq+\infty$ and
  $x_{j''+1}=+\infty$.  We can easily check that, \emph{after the clock shift}, Item 3
  holds with $j''$ and the sequence of indices
  $i_{n}\leq i'_{n}<i_{n+1}\leq i'_{n+1}<\cdots <i_{j}\leq i'_{j}\leq i$. 
  In particular, the value of $x_1$ after the clock shift is the value of $x_{n}$ at
  $C'_{i}$, and $b$ is smaller than $(\leqlt$) this value due to the condition of the
  clock shift. Similarly, if $j''>1$ then $b$ is smaller than $(\leqlt$) the value of $x_2$ 
  after the clock shift which is the value of $x_{n+1}$ at $C'_i$.
  Therefore, the last bullet of Item 3 holds after the clock shift, hence also at
  $C_{i+1}$.
    
  \medskip
  
  Now, we can prove as above, in the case of no clock shift, that after executing Lines 
  (9--14) of \Cref{algo:Since-I-general}, Item 3 holds at $C_{i+1}$.
  This concludes the inductive proof of Item (3).
\end{proof}

\begin{proof}[Proof of \Cref{lem:since-general}]
  Fix $i\geq0$. By \Cref{lem:since-general-invariants}, position $i$ satisfies the 
  untimed since $\argone\S\argtwo$ iff $C_{i+1}(\pre\AS\state)=\ell_1$ (corresponding to 
  $\post\AS\state=\ell_1$ in $\outv_{\S_I}$). We assume that we are in this case.
  The only question is whether there is a witnessing $\argtwo$ within the interval $I$.
  We claim that this can be determined by looking solely at the values of clocks $x_1$ 
  and $y_1$ at configuration $C_{i+1}$ (corresponding to the test
  $\post\ASgenI{x_1}\in I \vee \post\ASgenI{y_1}\in I$ in $\outv_{\S_I}$).
  
  We can view this as follows.  
  Suppose there is in the past a $\argtwo$ event witnessing $\argone\S_{I}\argtwo$:
  $\argtwo\in a_{i''}$ for some $i''\leq i$ and
  $\argone\in a_{i'}$ for all $i''<i'\leq i$ and $\tau_i-\tau_{i''}\in I$.
  Let $1\leq j\leq k$ and $i_1,i'_1,\ldots,i_j,i'_j$ be the indices given by
  \Cref{lem:since-general-invariants}~(3).
  
  If $i''<i_1$ then $(w,i_1)\models\argone\wedge\Y(\argone\S\argtwo)$.
  Therefore in this case we have $b\leqlt x_1$ at $C_{i+1}$. 
  Now, $b\leqlt \tau_i-\tau_{i_{1}}\leq\tau_i-\tau_{i''}\in I$.
  We deduce that $i_1$ is also a witness of $(w,i)\models\argone\S_{I}\argtwo$.
  
  If $i'_1<i''$ then we must have $j\geq 2$.  Moreover, $i_2\leq i''$ since there are no
  $\argtwo$ in $(i'_1,i_2)$.  Using $b\nleqlt x_2$ at $C_{i+1}$, we deduce that $b\nleqlt
  \tau_i-\tau_{i_2}$ and then $b\nleqlt \tau_i-\tau_{i''}$, a contradiction with
  $\tau_i-\tau_{i''}\in I$.
  
  Finally, if $i_1<i''<i'_1$ then using $x_1-y_1<c-b$ at $C_{i+1}$ we get
  $\tau_{i'_1}-\tau_{i_1}<c-b$.  Using $\tau_{i_1}\leq\tau_{i''}\leq\tau_{i'_1}$ and
  $\tau_i-\tau_{i''}\in I$, we deduce that either $i_1$ or $i'_1$ (or both) is also a
  witness of $(w,i)\models\argone\S_{I}\argtwo$.
\end{proof}

\subsection{Missing proofs from \Cref{sec:future}}

\lemuntilsidedbound*

\begin{proof}
  Both the automata $\AUfirst$ and $\AUlast$ contain a single state and two transitions.
  The tests in these two transitions are negations of each other and hence ensure a
  deterministic choice at position $i$.  However, in $\AU$, there is non-determinism in
  the choice of transitions.  Item $1$ above claims that $C_i(\pre\AU\state) =
  C'_i(\pre\AU\state) = \ell_1$ iff position $(w, i)$ satisfies $\argone \U \argtwo$.
  Hence, if the claim is true, it follows that there is an unambiguous choice even in
  this automaton, leading to a unique run of the network on the given word $w$.  Hence, we
  move on to proving the claim.

  We start with a property of clock $x$ that can be directly deduced from the transitions
  of $\AUfirst$.  If $C'_i(\pre\AU{x}) \neq -\infty$, then there is a point $j \ge i$ such
  that the following hold (i) $\argtwo \in a_j$, (ii) for all $i \leq k < j$, $\argtwo \notin
  a_k$ and (iii) $\Absolut{C'_i(\pre\AU{x})} = \tau_j - \tau_i$.  Else, $C'_i(\pre\AU{x}) =
  -\infty$, and there is no $j \ge i$ such that $\argtwo \in a_j$.

  \begin{enumerate}
    \item Suppose $C_i(\pre\AU\state) = C'_i(\pre\AU\state) = \ell_0$.  All outgoing
    transitions from $\ell_0$ satisfy $\neg \argtwo$.  Hence, $\argtwo \notin a_i$.  If
    for all $j \ge i$, $C_j(\pre\AU\state) = C'_j(\pre\AU\state) = \ell_0$, then clearly
    $\argtwo \notin a_j$ for all $j \ge i$, showing that $(w, i) \models \neg (\argone \U
    \argtwo)$.  Else, pick the smallest $j > i$ such that $C_j(\pre\AU\state) =
    C'_j(\pre\AU\state) = \ell_1$.  Since the only transition that leads out of $\ell_0$
    satisfies $\neg \argone \wedge \neg \argtwo$, we deduce that $\argone \notin a_{j-1}$.
    Furthermore, $\argtwo \notin a_k$ for all $i \le k \le j-1$.  This proves that $(w, i)
    \models \neg (\argone \U \argtwo)$.
    
    Now suppose $C_i(\pre\AU\state) = C'_i(\pre\AU\state) = \ell_1$.  The two transitions
    out of $\ell_1$ require that either $\argone$ or $\argtwo$ is present in $a_i$.  From
    the invariant $\pre\AUfirst{x} \neq -\infty$, we can infer that there is an earliest
    point $j \ge i$ such that $\argtwo \in a_j$.  Coupling the two arguments, we deduce
    that $\argone \in a_k$ for all $i \le k < j$.  Hence $(w, i) \models \argone \U
    \argtwo$.
    
    \item We have seen that $x$ maintains the time to the earliest $\argtwo$.  From the
    previous item, we know $(w, i) \models \argone \U \argtwo$ iff $C'_i(\pre\AU\state) =
    \ell_1$.  Hence, if $\outv_{\U_{I}}$ holds at $C'_i$, clearly $(w, i) \models \argone
    \U_I \argtwo$.  Conversely, if $(w, i) \models \argone \U_I \argtwo$, then as $I$ is
    an upper bounded interval, the earliest witness should lie in the interval $I$, and
    therefore $\outv_{\U_{I}}$ should hold at $C'_i$.
    
    \item Again, from Item 1, we know that $(w,i)\models\argone\U\argtwo$ if and only if
    $C'_i(\pre\AU\state)=\ell_1$, which is the first part of the test $\outv_{\U_{I}}$ at
    $C'_i$.  We assume below that this is the case.
    
    From the transitions of $\AUlast$, we see that clock $y$ tracks the time to the next
    position, if it exists, that satisfies $\argtwo \wedge \neg (\argone \wedge
    \post\AU\state = \ell_1)$. 
    
    If such a position $j\geq i$ exists, it is the latest witness of $\argone\U\argtwo$ at
    $i$ and we have $\Absolut{C'_i(\pre\AUlast{y})}=\tau_j-\tau_i$.  Since $I$ is a lower
    bounded interval, we deduce that in this case $(w,i)\models\argone\U_I\argtwo$ if and
    only if $\outv_{\U_{I}}$ holds at $C'_i$.
    
    If there is no such position, since we have assumed that
    $(w,i)\models\argone\U\argtwo$, it means that there are infinitely many witnesses:
    $(w,i)\models\G\argone\wedge\G\F\argtwo$.  In this case we have
    $\Absolut{C'_i(\pre\AUlast{y})}=+\infty$. Also, since $I$ is a lower bounded interval 
    and the word is non-zeno, there are witnesses $j\geq i$ with $\argtwo\in a_j$ and 
    $\tau_j-\tau_i\in I$. Again, this proves that $(w,i)\models\argone\U_I\argtwo$ if and
    only if $\outv_{\U_{I}}$ holds at $C'_i$.
    \qedhere
  \end{enumerate}
\end{proof}

\section{Appendix for liveness}
\label{sec:couvreur-algorithm}

For the purpose of the algorithm, we consider the finite graph $\finiteabs$ defined as the
zone graph of $\gtasystem$ quotiented by the mutual simulation relation.  The finiteness
and soundness of this abstraction for liveness was shown in~\cite{AGGS-CONCUR24}.

\begin{algorithm}[h]
  \caption{SCC decomposition based algorithm for GTA liveness with generalized B\"uchi acceptance condition.}
  \label{alg:liveness}
  \begin{algorithmic}[1]
    \State \textbf{Global variables}
    \State $\countint$: Integer initialized to $0$
    \State $\Roots$: Stack storing tuples $(\eta, \labels, \releasedclocks)$ where $\eta$ is a potential root node, 
    $\labels$ is the set of labels seen in SCC rooted at $\eta$, 
    and $\releasedclocks$ is the set of clocks released in SCC rooted at $\eta$ 
    \State $\CurrentReleasedClocks$: Stack storing the set of future clocks that have been released in the current SCC
    \State $\Active$: Stack storing the nodes which are active in the SCC exploration of the current root node
    \State $\Todo$: Stack storing pairs $(\eta, \suc)$ where $\eta$ is a node and $\suc$ is a list of successor node and transition. 
    \Function{GTASCC}{$\finiteabs, L$}
      \Comment Outputs if there is a live component in $\finiteabs$ which visits all labels from $L$ 
      \ForOne{Each initial node $\eta$ in $\finiteabs$}{$\Call{Dfs}{\eta, L}$}\EndForOne
    \EndFunction
    \Function{Dfs}{$\eta, L$}
      \Comment{Perform DFS starting from $\eta$}
      \State $\Call{Push}{\eta, \emptyset}$
      \While{$\Todo$ is not empty}
        \State $(\nu, \suc) \gets \Call{Todo.top}$ 
        \If{$\suc$ is empty}
          \If{$(\nu, \_, \_) = \Call{Roots.top}$} \Comment{SCC rooted at $u$ is explored}
            \State $\Call{Close}{s}$
          \EndIf
          \State $\Call{Todo.pop}$
        \Else
          \State $(t, \chi) \gets \Call{succ.FirstAndRemove}$  \Comment{Explore along transition $\nu \xra{t} \chi$}
          \If{$\chi.\dfsnum = 0$} \Comment{New node}
            \State $\Call{Push}{\chi, \released(t)}$ \Comment{$\released(t)$ is the set of future clocks released in $t$}
          \ElsIf{$\chi.\curr$} \Comment{Cycle detected}
            \State $\Call{CurrentReleasedClocks.push}{released(t)}$
            \State $\Call{MergeSCC}{\chi, L}$
          \EndIf
        \EndIf
      \EndWhile
    \EndFunction
    \algstore{bkbreak}
  \end{algorithmic}
  \end{algorithm}
  
  \begin{algorithm}[h]
  \begin{algorithmic}[1]
  \algrestore{bkbreak}
    \newpage
    \Function{Push}{$\eta, X_{\mathsf{rel}}$} \Comment{$\eta$ is a new node, $X_{\mathsf{rel}}$ is the set of clocks released on transition to $\eta$}
      \State $\countint \gets \countint + 1$; $\eta.\dfsnum \gets \countint$; $\eta.\curr \gets true$
      \State $\Call{Active.push}{\eta}$
      \State $\Call{Todo.push}{(\eta, \nxt(\eta))}$ \Comment{$\nxt(\eta)$ returns the successors of $\eta$}
      \State $\Call{\CurrentReleasedClocks.push}{X_{\mathsf{rel}}}$
      \State $\Call{Roots.push}{(\eta, \labels(\eta), \emptyset)}$
    \EndFunction
    \Function{Close}{$\eta$} \Comment{Closes the SCC of $\eta$}
    \State $\Call{Roots.Pop}$
    \State $\Call{CurrentReleasedClocks.Pop}$
    \State $\mu \gets null$ 
    \Repeat
      \State $\mu \gets \Call{Active.top}$
      \State $\mu.\curr = false$
      \State $\Call{Active.pop}$
    \Until{$\mu \neq \eta$}
    \EndFunction
    \Function{MergeSCC}{$v, L$}
      \State $A \gets \emptyset$, $R \gets \emptyset$
      \State $(\mu, \_, \_) \gets \Call{Roots.top}$
      \Repeat
        \State $c \gets \Call{CurrentReleasedClocks.Top}$
        \State $\Call{CurrentReleasedClocks.Pop}$
        \State $R \gets R \cup c$
        \State $(\mu, l, c) \gets \Call{Roots.Top}$
        \State $\Call{Roots.Pop}$
        \State $A \gets A \cup l$, $R \gets R \cup c$
      \Until{$\mu.\dfsnum > \eta.\dfsnum$}
      \If {$L \subseteq A$ and $X_{F}=R \cup \Call{InfClocks}{\mu} = X_F$} 
        \Comment{$\Call{InfClocks}{\mu}$ is the set of future clocks which can take
        value $-\infty$ in zone of $u$}
        \State \Output{Accepting SCC detected}
      \EndIf
      \State $\Call{Roots.Push}{(\mu, A, R)}$
    \EndFunction
  \end{algorithmic}
\end{algorithm}

\clearpage

\section{Appendix for Implementation and Experimental evaluation}
\label{sec:app:experiments}

\subsection{Implementation details.}

Our tool, $\tool$, takes as input a $\MTLfp$ formula $\psi$ and generates a synchronous
network of generalized timed automata with shared variables based on the translation
described in \Cref{sec:fastmtl-to-ta,sec:future}.  The execution of this network is simulated using an array
of shared variables.  Each generalized timed automaton writes its output to a dedicated
shared variable, which can then be accessed by its parent automaton.  Synchronization is
enforced by sequentially executing each automaton in bottom-to-top order.  A shared
control variable tracks the active automaton during execution, and a transition in an
automaton is executed when its ID matches the value of the control variable.  Upon taking
the transition, the control variable is updated to the ID of the next automaton.  A
controller is used to reset the execution after every automaton takes a transition and
alternates between time elapse and automaton execution.  A master automaton evaluates the
top-level boolean subformula, ensuring the formula is satisfied at the start and, for
finite timed words, that all automata have performed a transition in the last round.

To avoid constructing automata for boolean operators, we use symbolic labels to
represent temporal operators and atomic propositions, generating boolean expressions with
these labels.  For model checking, we construct the network of automata for $\neg \psi$,
where $\psi$ is the input specification.  The input model and the network of automata
generated by our translation are executed alternately, with transitions of the model
updating atomic propositions in the shared variable array, which the network of timed
automata reads.

\subsection{Comparison of translation times}

In~\Cref{table:translation-time}, we discuss the time for the translation for a few
instances.  We remark that the time for the translation is rather small in most cases.

\begin{table}[hbtp]
  \centering
  \resizebox{\textwidth}{!}{
    \begin{tabular}{|l|l|r|r|}
      \hline
      Sl. 
      & \multicolumn{1}{c|}{Formula}
      & \multicolumn{2}{|c|}{Time taken by}\\
      \cline{3-4}
      No. 
      & 
      & \MightyL\ (ms) 
      &  \tool\ (ms) \\
    \hline
    1 & $(((((p0 \U_{[1, 2]} q1) \U_{[1, 2]} q2) \U_{[1, 2]} q3) \U_{[1, 2]} q4) \U_{[1, 2]} q5)$ & 50 & 5\\
      \hline
      2 & $\Always_{[0, 10]} \Eventually_{[1, 2]} ( (a \wedge (\Next b) ) \vee  (\neg a \wedge !(\Next b) ))$ & 15 & 3\\
      \hline
      3 & $\Always_{[2, \infty)} ( (\neg request \vee  \Eventually_{[4, 5]} grant))$ & 42 & 4\\
      \hline
      4 & $\Eventually_{[1, 2]} p1 \wedge \Eventually_{[1, 2]} p2 \wedge \Eventually_{[1, 2]} p3 \wedge \Eventually_{[1, 2]} p4 \wedge \Eventually_{[1, 2]} p5$ & 49 & 3\\
      \hline
  \end{tabular}  
  } 
  \caption{Comparison of time taken by \tool, and \MightyL\ for translation of $\mitl$
  formulae to \GTAfull\ (timed automata respectively).}
\label{table:translation-time}
\end{table}

\subsection{\bf Engineering optimizations.}

We have implemented several optimizations to improve performance, of which we discuss one
such key optimization here.  First, in addition to the locations of each automaton, the
global state in our translation also stores the output of each automaton in the array of
shared variables.  Each shared variable doubles the global state space of the network.
However, at any point in the sequential execution, the subsequent behaviour of the network
depends only on the outputs of the previous automaton which are the direct dependencies of
the remaining automata in the round.  Therefore, once the last parent of an automaton
reads a shared variable, its value can be discarded and the variable reused to store the
output of a different automaton.  We use static liveness analysis to determine when shared
variables become dead and reuse them, reducing the number of auxiliary variables.  This
optimization minimizes the state size and the overall state space of the zone graph.

\subsection{Ablation Study}
\label{sec:app:ablation}

To better understand the contribution of individual improvements in \tool, we perform
ablation studies by selectively disabling them and measuring their impact on performance.
For the ablation study, we consider (a) one of our theoretical contributions, namely
\emph{sharing of predictions}, (b) \emph{specialized automata for the $\Eventually$
operator} and (c) finally, an engineering optimization, namely \emph{memory reuse}
discussed above.

For each feature, we compare the performance of \tool\ with the feature, \tool\ with the
feature turned off, and the state-of-the-art tool \MightyL. Our results demonstrate that
each of these improvements plays a critical role in enabling scalability, often yielding exponential improvements.

\paragraph*{Sharing Predictions.}
In the absence of prediction sharing, each until subformula in the network of automata
independently guesses the time of satisfaction.  As a result, multiple non-deterministic
predictions must be resolved during reachability analysis, significantly increasing the
complexity of the state space.  With prediction sharing enabled, a single dedicated
automaton generates the prediction, and the remaining components deterministically consume
this information.

This reduces the overall non-determinism in the system and leads to a much smaller product
automaton.  Empirical results confirm this effect: benchmark families involving multiple
until operators show exponential gains in verification time when prediction sharing is
enabled.

We depict the results in \Cref{fig:ablation-plots} (left) for $\phi(n) = 
\bigwedge_{i=1}^{n} p \U_{[0, i]} q$. 

\begin{figure}[htbp]
    \centering
    \begin{minipage}{0.3\textwidth}
        \centering
        \includegraphics[scale=0.16]{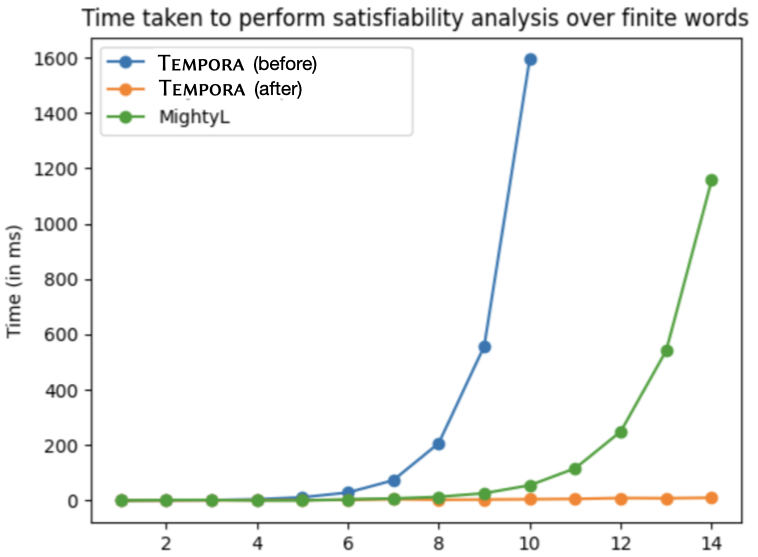}  
    \end{minipage}
    \hspace{0.05in}
    \begin{minipage}{0.3\textwidth}
      \centering
      \includegraphics[scale=0.16]{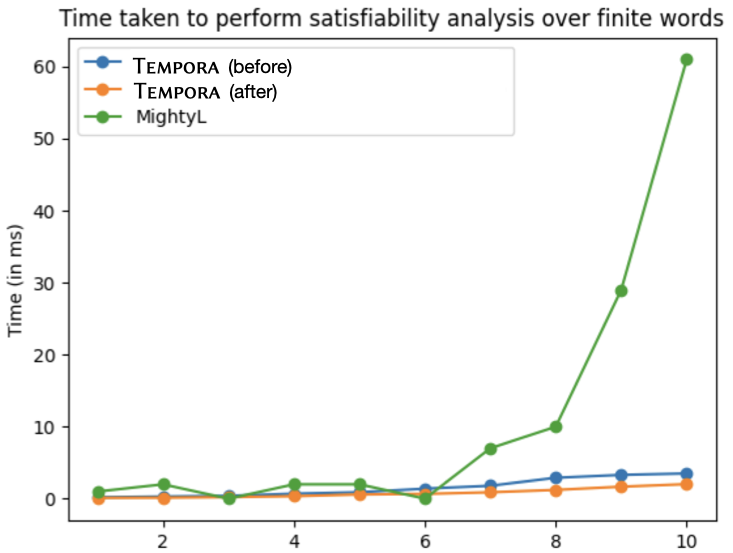}  
    \end{minipage}
    \hspace{0.05in}
    \begin{minipage}{0.3\textwidth}
      \centering
      \includegraphics[scale=0.16]{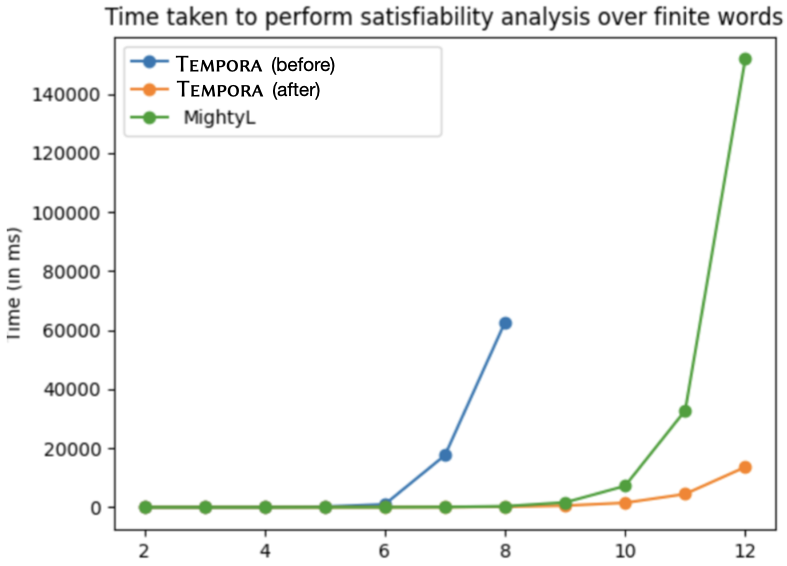}  
    \end{minipage}
  \caption{Comparison of performance of \tool\ with and without features. 
  The left figure shows the impact of sharing predictions. 
  The middle figure shows the impact of specialized automata for $\Eventually$ operator.
  The right figure shows the impact of memory reuse.
  The x-axis represents the parameter in the formula, while the y-axis shows the time taken in milliseconds.}
    \label{fig:ablation-plots}
  \end{figure}
  
\medskip\noindent {\bf Specialized Automaton for $\Eventually$}
While the eventually operator $\F_I p$ can be encoded as $\true \U_I p$, this transformation
can be an overkill for untimed or simple timing intervals like $[0,\infty)$.  \tool\
implements specialized automata for such cases that are smaller and structurally simpler
than general-purpose until-based encodings.

For instance, the automaton for $\F_{[0,a]} p$ simply waits up to $a$ time units for $p$ to
become true, without introducing intermediate obligations or control states.  These
specialized constructions avoid unnecessary non-determinism and reduce clock usage.  Our
benchmarks show that formulas involving eventualities with small intervals benefit
substantially from this, often reducing the number of locations and transitions in the
resulting automaton by an order of magnitude.  The results are shown in
\Cref{fig:ablation-plots} (middle) for the formula $\phi(n) = \bigwedge_{i=1}^N \F_{[0,
i]} p$.

\medskip\noindent {\bf Reusing Memory}
The memory optimization addresses inefficiencies in how truth values of atomic
propositions are stored during automaton execution.  Consider the formula $\phi(n) =
\bigwedge_{i=1}^n \X p_i$, where each $\X p_i$ depends only on the value of $p_i$ at the
previous timestamp.  In a naive implementation, each $p_i$ is assigned a unique index in
the state array, even though their \emph{lifetimes} do not overlap.

\tool\ incorporates a live-range analysis to reuse memory indices once a variable becomes
\emph{dead} — i.e., no longer accessed in the remaining evaluation of the round.  In the
above example, all $p_i$ can safely share a single index.  This drastically reduces the
size of the memory array and thus the number of reachable states.  Our experiments
demonstrate that this reuse yields exponential improvements in model-checking time for
formulas with long chains of independent next operators, as depicted in
\Cref{fig:ablation-plots} (right).

\subsection{\bf Benchmarks and Extended Experimental Results}

We next briefly describe the additional benchmarks we use to evaluate our performance.

\medskip\noindent{\bf Our benchmarks.}
We constructed several $\MTLfp$ formulae to show the benefits of our improvements.  To
show the benefit of sharing clocks, we use a disjunction of multiple $\Next$ operators:
$\G_{[0, 2]}(\bigvee_{i=1}^{n} \Next_{[2i, 2i + 1]} p_i)$, where each Next automaton uses a
single shared clock.  A conjunction of multiple $\U$ operators with the same left and
right children: $\Always(\bigwedge_{i=1}^{n} (p_i \rightarrow a \U_{[0, i] } b)) \wedge
\Always (\bigvee_{i=1}^{n} p_i)$ shows the benefit of using sharer automata, where each Until
automaton uses a common sharer automaton, as well as using a smaller 2 state automaton for
Until.  The power of using deterministic automata for top level operators can be seen for
a formula with conjunction of multiple $\U$ operators having general intervals at the top
level: $\bigwedge_{i=1}^{n} (p \U_{[2i, 2i + 1]} q)$.  We show the benefit of reusing shared
variables for formulae where atomic propositions have short live ranges:
$\bigvee_{i=1}^{n}\Always_{[2, \infty)} (p_i)$.  In all these formulae, we try to ensure that
they cannot be easily rewritten to a simpler formula, and are not trivially
satisfiable/unsatisfiable.

\medskip\noindent{\bf Acacia Benchmarks.}
Each LTL formula in the original Acacia benchmark is of the form $\psi \iff \Always
\Eventually acc$, where $\psi$ is an LTL specification.  We modify these examples by adding
intervals to temporal operators of $\psi$.  For several of these specifications, we add
intervals to the operators specifying occurrences of atomic propositions, which adds
stricter constraints for the atomic propositions to occur within an interval.  For some of
the benchmarks, we also add general intervals to nested $\U$ operators to add complexity
to the formula.

\medskip

The complete experimental results are presented in Table A (for finite words) and Table B
for infinite words.

\clearpage\noindent\hspace{-25mm}    
\smash{\raisebox{-240mm}{\includegraphics[width=1.3\textwidth]{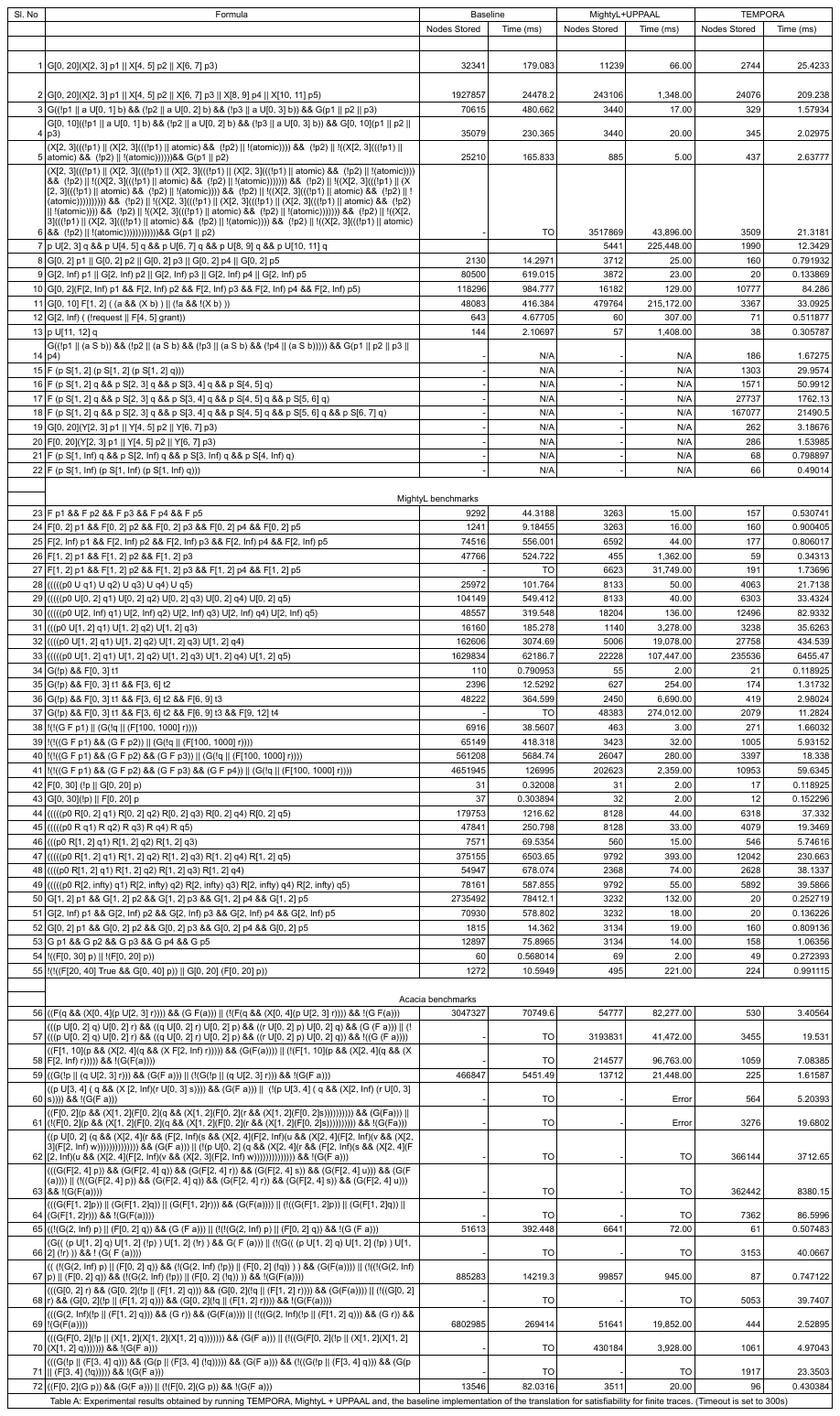}}}

\clearpage\noindent\hspace{-2cm}    
\smash{\raisebox{-240mm}{\includegraphics[width=1.3\textwidth]{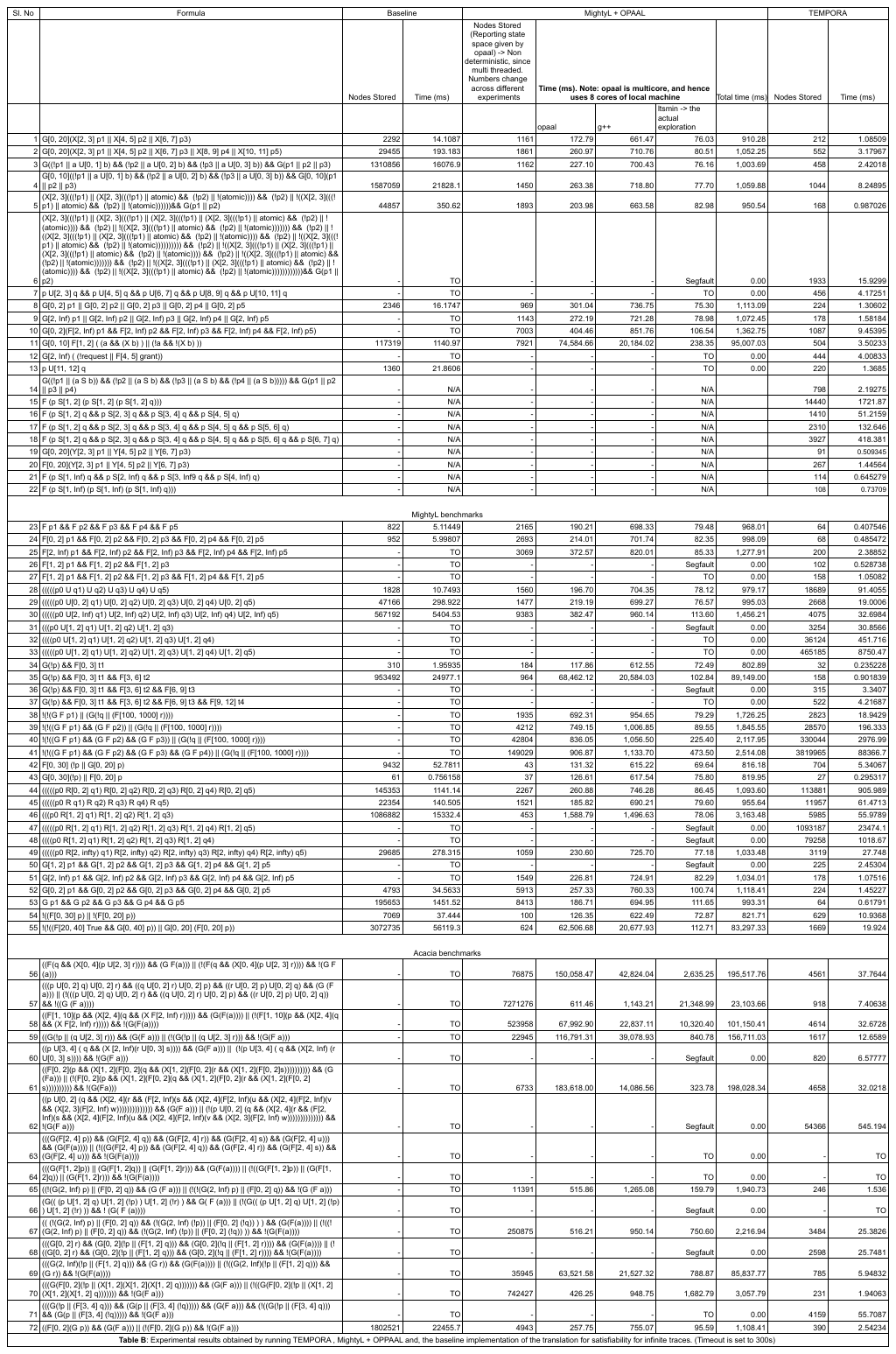}}}

\end{document}